\pdfoutput=1

\documentclass[11pt]{article}
\usepackage{xspace}
\usepackage{mathptmx}
\usepackage{epsfig}
\usepackage{epstopdf}
\usepackage[sort]{cite}
\usepackage{graphicx}
\usepackage{amsfonts}
\usepackage{amsmath}
\usepackage{amsthm,mathtools}
\usepackage{latexsym}
\usepackage{amssymb}
\usepackage{color}
\usepackage{comment}
\usepackage[margin=1in]{geometry}
\usepackage{extras,tikz}
\usepackage{psfrag}
\usepackage{comment}
\usepackage[noend]{algorithmic}
\usepackage{algorithm}

\usepackage[protrusion=true,expansion=true]{microtype}

\usepackage{hyperref}

\definecolor{darkblue}{rgb}{0,0,0.45}
\definecolor{darkred}{rgb}{0.6,0,0}
\definecolor{darkgreen}{rgb}{0.13,0.5,0}
\hypersetup{colorlinks, linkcolor=darkblue, citecolor=darkgreen,
urlcolor=darkblue}

\newtheorem{theorem}{Theorem}[section]
\newtheorem{lemma}[theorem]{Lemma}
\newtheorem{claim}[theorem]{Claim}

\newtheorem{observation}[theorem]{Observation}

\newtheorem{definition}[theorem]{Definition}

\DeclareMathAlphabet{\mathcal}{OMS}{cmsy}{m}{n}
\newcommand{\cqed}{\renewcommand{\qedsymbol}{$\lrcorner$}}

\newtheorem*{rep@theorem}{\rep@title} \newcommand{\newreptheorem}[2]{%
\newenvironment{rep#1}[1]{%
\def\rep@title{\autoref{##1}}%
\begin{rep@theorem} }%
{\end{rep@theorem} } }
\makeatother
\newreptheorem{theorem}{}
\newreptheorem{lemma}{}
\newreptheorem{corollary}{}

\newcommand{\executeiffilenewer}[3]{%
\ifnum\pdfstrcmp{\pdffilemoddate{#1}}%
{\pdffilemoddate{#2}}>0%
{\immediate\write18{#3}}\fi%
}

\newcommand{\poly}{{\mathrm{poly}}}

\newcommand{\mcc}{\textsc{Multicolored Clique} \xspace}
\newcommand{\st}{\textsc{Steiner Tree} \xspace}
\newcommand{\dst}{\textsc{Directed Steiner Tree} \xspace}
\newcommand{\scss}{\textsc{Strongly Connected Steiner Subgraph} \xspace}
\newcommand{\dsn}{\textsc{Directed Steiner Network}\xspace}
\newcommand{\csi}{\textsc{Partitioned Subgraph Isomorphism}\xspace}
\newcommand{\gt}{\textsc{Grid Tiling}\xspace}

\newcommand{\T}{\mathcal{T}}

\newcommand{\tw}{\ensuremath{\textbf{tw}}}

\newcommand{\inn}{\texttt{in}}
\newcommand{\out}{\texttt{out}}

\newif\ifappendix
\newif\ifmain
\newif\iffullorapp
\maintrue
\fullorappfalse
\newif\ifabstract
\abstracttrue
 \abstractfalse 
\newif\iffull
\ifabstract \fullfalse \else \fulltrue \fullorapptrue \fi

\psfragscanon

{\makeatletter
 \gdef\xxxmark{%
   \expandafter\ifx\csname @mpargs\endcsname\relax 
     \expandafter\ifx\csname @captype\endcsname\relax 
       \marginpar{xxx}
     \else
       xxx 
     \fi
   \else
     xxx 
   \fi}
 \gdef\xxx{\@ifnextchar[\xxx@lab\xxx@nolab}
 \long\gdef\xxx@lab[#1]#2{{\bf [\xxxmark #2 ---{\sc #1}]}}
 \long\gdef\xxx@nolab#1{{\bf [\xxxmark #1]}}
 \long\gdef\xxx@lab[#1]#2{}\long\gdef\xxx@nolab#1{}%
}

\begin{document}

\title{Tight Bounds for Planar Strongly Connected Steiner Subgraph with Fixed Number of Terminals (and Extensions)\thanks{A preliminary version of this paper appeared in SODA 2014~\cite{DBLP:conf/soda/ChitnisHM14}}}

\author{Rajesh Chitnis\thanks{School of Computer Science, University of
Birmingham. Part of this work was done while at the University of Maryland, USA and University of Warwick, UK (supported by ERC grant 2014-CoG 647557). Email: \texttt{rajeshchitnis@gmail.com}} \and
Andreas Emil Feldmann\thanks{Department of Applied Mathematics, Charles
University, Prague, Czechia. Supported by project CE-ITI (GA\v{C}R
no.~P202/12/G061) of the Czech Science Foundation, and by the Center for
Foundations of Modern Computer Science (Charles Univ.\ project UNCE/SCI/004).
Email: \texttt{feldmann.a.e@gmail.com} } \and MohammadTaghi
Hajiaghayi\thanks{Department of Computer Science , University of Maryland at
College Park, USA. Supported in part by NSF CAREER award 1053605, ONR YIP award
N000141110662, DARPA/AFRL award FA8650-11-1-7162 and a
University of Maryland Research and Scholarship Award (RASA). Email: \texttt{hajiagha@cs.umd.edu}} \and D{\'a}niel Marx\thanks{Max Planck Institute for Informatics, Saarbr\"ucken, Germany. Supported by ERC Starting Grant PARAMTIGHT (No. 280152) and ERC Consolidator Grant SYSTEMATICGRAPH (No. 725978). Email: \texttt{dmarx@mpi-inf.mpg.de}}}

\date{}

\maketitle
\thispagestyle{empty}

\begin{abstract}
Given a vertex-weighted directed graph $G=(V,E)$ and a set $T=\{t_1, t_2, \ldots t_k\}$ of $k$ terminals, the objective of the \scss (SCSS)
problem is to find a vertex set $H\subseteq V$ of minimum weight such that $G[H]$ contains a $t_{i}\rightarrow t_j$
path for each $i\neq j$. The problem is NP-hard, but Feldman and Ruhl [FOCS '99; SICOMP '06] gave a novel $n^{O(k)}$
algorithm for the SCSS problem, where $n$ is the number of vertices in the graph and $k$ is the number of terminals.
We explore how much easier the problem becomes on planar directed graphs.

\begin{itemize}
\item Our main algorithmic result is a $2^{O(k)}\cdot
  n^{O(\sqrt{k})}$ algorithm for planar SCSS, which is an improvement
  of a factor of $O(\sqrt{k})$ in the exponent over the algorithm of
  Feldman and Ruhl.
\item Our main hardness result is a matching lower bound for our
  algorithm: we show that planar SCSS does not have an $f(k)\cdot
  n^{o(\sqrt{k})}$ algorithm for any computable function $f$, unless
  the Exponential Time Hypothesis (ETH) fails.
\end{itemize}

To obtain our algorithm, we first show combinatorially that there is a minimal
solution whose treewidth is $O(\sqrt{k})$, and then use the dynamic-programming
based algorithm for finding bounded-treewidth solutions due to Feldmann and Marx
[ICALP '16].
%
%
To obtain the lower bound matching the
algorithm, we need a delicate construction of gadgets arranged in a
grid-like fashion to tightly control the number of terminals in the
created instance.

The following additional results put our upper and lower bounds in context:
\begin{itemize}
\item Our $2^{O(k)}\cdot n^{O(\sqrt{k})}$ algorithm for planar directed graphs
can be generalized to graphs excluding a fixed minor. Additionally, we can
obtain this running time for the problem of finding an optimal planar solution
even if the input graph is not planar.
\item In general graphs, we cannot hope for such a dramatic improvement over
the $n^{O(k)}$ algorithm of Feldman and Ruhl: assuming ETH,
SCSS in general graphs does not have an $f(k)\cdot n^{o(k/\log k)}$ algorithm
for any computable function $f$.
\item Feldman and Ruhl generalized their $n^{O(k)}$ algorithm to the more
general \dsn (DSN) problem; here the task is to find a subgraph of minimum
weight such that for every source $s_i$ there is a path to the corresponding
terminal $t_i$. We show that, assuming ETH, there is no $f(k)\cdot
n^{o(k)}$ time algorithm for DSN on acyclic planar graphs.
\end{itemize}

All our lower bounds hold for the edge-unweighted version, while the algorithm works for the more general vertex-(un)weighted version.
\end{abstract}





\section{Introduction}
\label{sec:intro}

%

The \textsc{Steiner Tree} (ST) problem is one of the earliest and
most fundamental problems in combinatorial optimization: given an
undirected graph $G=(V,E)$ and a set $T\subseteq V$ of terminals, the
objective is to find a tree of minimum size which connects all the
terminals. The ST problem is believed to have been
first formally defined by Gauss in a letter in 1836, and the first
combinatorial formulation is attributed
independently to Hakimi~\cite{hakimi} and Levin~\cite{levin} in
1971. The ST problem is known to be NP-complete, and was in fact part of
Karp's original list~\cite{karp1972reducibility} of 21
NP-complete problems. In the directed version of the ST problem,
called \textsc{Directed Steiner Tree} (DST), we are
also given a root vertex $r$ and the objective is to find a minimum
size arborescence which connects the root $r$ to each terminal from
$T$. An easy reduction from \textsc{Set Cover} shows that the DST
problem is also NP-complete.

Steiner-type problems arise in the design of networks. Since many networks are
symmetric, the directed versions of Steiner-type problems were mostly of
theoretical interest. However, in recent years, it has been
observed~\cite{ramanathan1996multicast, salama1997evaluation} that the
connection cost in various networks such as satellite or radio networks are not
symmetric. Therefore, directed graphs form the most suitable model for such
networks. In addition, Ramanathan~\cite{ramanathan1996multicast} also used the
DST problem to find low-cost multicast trees, which have applications in
point-to-multipoint communication in high bandwidth networks. We refer the
interested reader to Winter~\cite{winter1987steiner} for a survey on
applications of Steiner problems in networks.

In this paper we consider two
well-studied Steiner-type problems in directed graphs, namely the \scss and the
\dsn problems. In the (vertex-unweighted) \scss (SCSS) problem, given a directed graph $G=(V,E)$ and
a set $T=\{t_{1}, t_{2}, \ldots, t_{k}\}$ of $k$ terminals, the objective is to
find a set $S\subseteq V$ of minimum size such that $G[S]$ contains a $t_{i}\rightarrow t_{j}$
path for each $1\leq i\neq j\leq k$. Thus, just as DST, the SCSS problem is
another directed version of the ST problem, where all terminals need to be
connected to each other. The (vertex-unweighted) \dsn (DSN) problem generalizes both DST and SCSS:
given a directed graph $G=(V,E)$ and a set $T=\{(s_1, t_{1}), (s_{2}, t_{2}),
\ldots, (s_k, t_{k})\}$ of $k$ pairs of terminals, the objective is to find a
set $S\subseteq V$ of minimum size such that $G[S]$ contains an $s_{i}\rightarrow t_{i}$ path
for each $1\leq i\leq k$.
We first describe the known results for both SCSS and DSN before stating our
results and techniques.

\subsection{Previous work}

Since both DSN and SCSS are NP-complete, one can try to design polynomial-time
approximation
algorithms for these problems. An $\alpha$-approximation for DST implies a $2\alpha$-approximation for SCSS as follows: fix a
terminal $t\in T$ and take the union of the solutions of the DST instances $(G,t, T\setminus t)$ and $(G_{\textup{rev}}, t, T\setminus
t)$, where $G_{\textup{{rev}}}$ is the graph obtained from $G$ by reversing the orientations of all edges. The best known approximation
ratio in polynomial time for SCSS is $k^{\epsilon}$ for any $\epsilon>0$~\cite{DBLP:journals/jal/CharikarCCDGGL99}. A result
of Halperin and Krauthgamer~\cite{DBLP:conf/stoc/HalperinK03} implies SCSS has no $\Omega(\log^{2-\epsilon} n)$-approximation
for any $\epsilon>0$, unless NP has quasi-polynomial Las Vegas algorithms. For
the more general DSN problem, the best approximation ratio known is $n^{2/3 +
\epsilon}$ for any $\epsilon>0$. Berman et
al.~\cite{DBLP:journals/iandc/BermanBMRY13} showed
that DSN has no $\Omega(2^{\log^{1-\epsilon} n})$-approximation for any $0< \epsilon <1$, unless NP has quasi-polynomial
time algorithms.

Rather than finding approximate solutions in polynomial time, one can
look for exact solutions in time that is still better than the running
time obtained by brute force algorithms. For (unweighted versions of) both the SCSS and DSN
problems, brute force can be used to check in time $n^{O(p)}$ if a
solution of size at most $p$ exists: one can go through all sets of
size at most $p$. A more efficient algorithm would have runtime $f(p)\cdot n^{O(1)}$,
where $f$ is some computable function depending only on $p$. A problem is said to be \emph{fixed-parameter tractable} (FPT) with a particular parameter $p$ if it
admits such an algorithm;
see~\cite{fpt-book,downey-fellows,flum-grohe,niedermeier} for more
background on FPT algorithms. A natural parameter for our considered problems is the number $k$
of terminals or terminal pairs; with this parameterization, it is not
even clear if there is a polynomial-time algorithm for every fixed
$k$, much less if the problem is FPT.  It is known that \st on
undirected graphs is FPT parameterized by the number $k$ of terminals: the classical algorithm of Dreyfus and
Wagner \cite{DBLP:journals/networks/DreyfusW71} solves the problem in
time $3^k\cdot n^{O(1)}$. The
running time was recently improved to $2^k\cdot n^{O(1)}$ by
Bj\"orklund et al.~\cite{DBLP:conf/stoc/BjorklundHKK07}. The same
algorithms work for \dst as well.

For the SCSS and DSN problems, we cannot expect fixed-parameter
tractability: Guo et al.~\cite{guo-et-al} showed
that SCSS is W[1]-hard parameterized by the number of terminals $k$,
and DSN is W[1]-hard parameterized by the number of terminal pairs
$k$. In fact, it is not even clear how to solve these problems in
polynomial time for small fixed values of the number $k$ of
terminals/pairs.  The case of $k=1$ in DSN is the well-known shortest
path problem in directed graphs, which is known to be polynomial time
solvable. For the case $k=2$ in DSN, an $O(n^5)$ algorithm was given
by Li et al.~\cite{Li1992267} which was later improved to
$O(mn+n^{2}\log n)$ by Natu and Fang~\cite{Natu1997207}. The question
regarding the existence of a polynomial time algorithm for DSN when $k=3$
was open. Feldman and Ruhl~\cite{feldman-ruhl} solved this question by
giving an $n^{O(k)}$ algorithm for DSN, where $k$ is the number of
terminal pairs. They first designed an $n^{O(k)}$ algorithm for SCSS,
where $k$ is the number of terminals, and used it as a subroutine in
the algorithm for the more general DSN problem.

\subsection{Our results and techniques}

Given the amount of attention
the planar version of Steiner-type problems received from the
viewpoint of approximation (see, e.g.,
\cite{DBLP:conf/soda/BateniCEHKM11,DBLP:journals/jacm/BateniHM11,DBLP:journals/talg/BorradaileKM09,DBLP:conf/icalp/DemaineHK09a,DBLP:conf/soda/EisenstatKM12})
and the availability of techniques for parameterized algorithms on
planar graphs (see, e.g.,
\cite{DBLP:conf/focs/BodlaenderFLPST09,DBLP:journals/cj/DemaineH08,DBLP:journals/jacm/FrickG01,DBLP:conf/icalp/KleinM12,MarxPP-FOCS2018}),
it is natural to explore SCSS and DSN restricted to planar graphs\footnote{Planarity for directed graph problems refers to the underlying undirected graph being planar}. In
general, one can have the expectation that the problems restricted to
planar graphs become easier, but sophisticated techniques might be
needed to exploit planarity. In particular, a certain {\em square root phenomenon} was observed for a wide range of algorithmic problems: the exponent of the running time can be improved from $O(k)$ to $O(\sqrt{k})$ (or to $O(\sqrt{k}\log k)$) and lower bounds indicate that this improvement is essentially best possible \cite{DBLP:conf/icalp/Marx12,DBLP:conf/icalp/KleinM12,MarxPP-FOCS2018,KleinM14,FominLMPPS16,DemaineFHT05,DBLP:conf/stacs/PilipczukPSL13,DBLP:conf/esa/MarxP15,DBLP:conf/fsttcs/LokshtanovSW12,DBLP:journals/corr/AboulkerBHMT15,FominKLPS16}. Our main algorithmic result is also an improvement of this form:
\begin{theorem}
\label{thm:algo-sqrt-k-h-minor-free} An instance $(G,T)$ of the vertex-weighted \scss problem with $|G|=n$ and $|T|=k$ can be solved in $2^{O(k)}\cdot
n^{O(\sqrt{k})}$ time, when the underlying undirected graph of $G$ is planar.
\end{theorem}


This algorithm presents a major improvement over the Feldman-Ruhl algorithm for
SCSS in general graphs which runs in
$n^{O(k)}$ time. A preliminary version of this paper~\cite{DBLP:conf/soda/ChitnisHM14} by a subset of the authors contained a complicated algorithm with a worse running time of $2^{O(k\cdot \log k)}\cdot   n^{O(\sqrt{k})}$. It relied on modifying the Feldman-Ruhl token game, and then using the excluded grid theorem for planar graphs followed by treewidth-based techniques. We briefly give some intuition behind this algorithm and the original $n^{O(k)}$ algorithm of Feldman-Ruhl. The algorithm of Feldman-Ruhl for SCSS is based on defining a game with $2k$ tokens and costs associated with the moves of the tokens such that the minimum cost of the game is equivalent to
the minimum cost of a solution of the SCSS problem; then the minimum cost of the game can be computed by exploring a state space of size
$n^{O(k)}$. The $2^{O(k\cdot \log k)}\cdot   n^{O(\sqrt{k})}$ algorithm was obtained by generalizing the Feldman-Ruhl token game via introducing \emph{supermoves}, which are sequences of certain types of moves. The
generalized game still has a state space of $n^{O(k)}$, but it has the advantage that we can now give a bound of $O(k)$ on the number of
supermoves required for the game (such a bound is not possible for the original
version of the game). This gives an $O(k)$-sized summary of the token game, and
hence has treewidth~$O(\sqrt{k})$. However, this summary is ``unlabeled", i.e.,
we do not explicitly know which vertices occur where in the summary. Guessing by
brute force requires $n^{O(k)}$ time, and the improvement to $2^{O(k\cdot \log
k)}\cdot   n^{O(\sqrt{k})}$ is obtained by using an embedding theorem of Klein and Marx~\cite{DBLP:conf/icalp/KleinM12}.

Unlike the $2^{O(k\cdot \log k)}\cdot n^{O(\sqrt{k})}$ algorithm
of~\cite{DBLP:conf/soda/ChitnisHM14}, the $2^{O(k)}\cdot n^{O(\sqrt{k})}$
algorithm from Theorem~\ref{thm:algo-sqrt-k-h-minor-free} does not depend on the
Feldman-Ruhl algorithm. It is conceptually much simpler: first we show
combinatorially (see Lemma~\ref{lem:tw-sqrt-k}) that there is a minimal solution
whose treewidth is $O(\sqrt{k})$, and then use the dynamic-programming based
algorithm for finding bounded-treewidth solutions for DSN due to Feldmann and
Marx~\cite[Theorem 5]{andreas-dm-arxiv}. The simplicity of our new approach
also allows transparent generalizations in two directions:
\begin{itemize}
\item \underline{From planar to $H$-minor-free graphs:} we may use
the excluded grid minor theorem for $H$-minor-free
graphs~\cite{H-minor-free-grid-theorem} instead of the excluded grid minor
theorem for planar graphs~\cite{planar-grid-theorem} to prove the existence of
a minimal solution of treewidth $O(\sqrt{k})$, which again implies a
$2^{O(k)}\cdot n^{O(\sqrt{k})}$ time algorithm.

\item \underline{Between restricted inputs and restricted solutions:}
our algorithm only exploits the $H$-minor-freeness of an optimum solution, and not of
the whole input graph. Thus, only the existence of one optimum $H$-minor-free solution in an otherwise
unrestricted input graph is enough to show that some optimum solution (which might not necessarily be $H$-minor-free) can be found in $2^{O(k)}\cdot n^{O(\sqrt{k})}$ time\footnote{This can be thought of as an intermediate requirement between the two extremes of either forcing the input graph itself to be $H$-minor-free versus the other extreme of finding an optimum solution which is $H$-minor-free. There has been some recent work in this direction~\cite{rajesh-andreas-ipec,rajesh-andreas-pasin}}.

\end{itemize}

Can we get a better speedup in planar graphs than the improvement from $O(k)$ to $O(\sqrt{k})$ in the exponent of $n$?
Our main hardness result matches our algorithm: it shows that $O(\sqrt{k})$ is
best possible under the Exponential Time Hypothesis (ETH).
\begin{theorem}
\label{thm:scss-main-hardness-planar-graphs} The edge-unweighted version of the SCSS problem is W[1]-hard
parameterized by the number of terminals $k$, even when the underlying
undirected graph is planar. Moreover, under ETH, the SCSS problem on planar
graphs cannot be solved in $f(k)\cdot n^{o(\sqrt{k})}$ time where $f$ is any
computable function, $k$ is the number of terminals and $n$ is the number of
vertices in the instance.
\end{theorem}
This also answers the question of Guo et al.~\cite{guo-et-al}, who showed the
W[1]-hardness of these problems on general graphs and left the fixed-parameter
tractability status on planar graphs as an open question.
Recall that ETH has the consequence that
$n$-variable 3SAT cannot be solved in time $2^{o(n)}$~\cite{eth,eth-2}.
There are relatively few parameterized problems that
are W[1]-hard on planar
graphs~\cite{DBLP:conf/iwpec/BodlaenderLP09,DBLP:journals/mst/CaiFJR07,
DBLP:conf/iwpec/EncisoFGKRS09,DBLP:conf/icalp/Marx12}. The
reason for the scarcity of such hardness results is mainly because for most problems, the fixed-parameter tractability of finding a solution
of size $k$ in a planar graph can be reduced to a bounded-treewidth problem by standard layering techniques. However, in our case the
parameter $k$ is the number of terminals, hence such a simple reduction to the bounded-treewidth case does not seem to be
possible. Our reduction is from the \gt problem formulated by Marx~\cite{daniel-grid-tiling,DBLP:conf/icalp/Marx12} (see also \cite{fpt-book}), which is a
convenient starting point for parameterized reductions for planar problems.
For our reduction we need to construct two types of gadgets, namely the
connector gadget and main gadget, which are then arranged in a grid-like
structure (see Figure~\ref{fig:big-picture}). The main technical
part of the reduction is the structural result regarding the existence and
construction of particular types of connector gadgets
and main gadgets (Lemma~\ref{lem:connector-gadget} and Lemma~\ref{lem:main-gadget}). Interestingly, the construction of the
connector gadget poses a greater challenge: here we exploit in a fairly
delicate way the fact that the $t_i \leadsto t_j$ and the reverse $t_j\leadsto
t_i$ paths appearing in the solution subgraph might need to share edges to reduce the weight.

We present additional results that put our algorithm and lower bound for SCSS in a wider context.  Given our speedup for SCSS in planar
graphs, one may ask if it is possible to get any similar speedup in general graphs.  Our next result shows that the $n^{O(k)}$ algorithm of
Feldman-Ruhl is almost optimal in general graphs:
\begin{theorem}
\label{thm:scss-main-hardness-general-graphs} Under ETH, the edge-unweighted version of the SCSS problem cannot be solved in time $f(k)\cdot n^{o(k/\log k)}$ where $f$ is an arbitrary
computable function, $k$ is the number of terminals and $n$ is the number of vertices in the instance.
\end{theorem}


Our proof of Theorem~\ref{thm:scss-main-hardness-general-graphs} is similar to the W[1]-hardness proof of Guo et al.~\cite{guo-et-al}. They showed the W[1]-hardness of SCSS on general graphs parameterized by the number $k$ of terminals by giving a reduction from \textsc{$k$-Clique}. However, this reduction uses ``edge selection gadgets'' and since a $k$-clique has
$\Theta(k^2)$ edges, the parameter is increased at least to $\Theta(k^2)$. Combining with the result of Chen et al.~\cite{chen-hardness} regarding the non-existence of an $f(k)\cdot
n^{o(k)}$ algorithm for ${k}$-Clique under ETH, this gives a lower bound of $f(k)\cdot n^{o(\sqrt{k})}$ for SCSS on general graphs.
To avoid the quadratic blowup in the parameter and thereby get a stronger lower bound, we use the \csi (PSI) problem as the source problem of our reduction. For this problem, Marx~\cite{marx-beat-treewidth} gave a $f(k)\cdot n^{o(k/\log k)}$ lower bound under ETH, where $k=|E(G)|$ is the number of edges of the subgraph $G$ to be found in graph $H$. The reduction of Guo et al.~\cite{guo-et-al} from \textsc{Clique} can be turned into a reduction from PSI which uses only $|E(G)|$ edge selection
gadgets, and hence the parameter is $\Theta(|E(G)|)$. Then the lower bound of $f(k)\cdot n^{o(k/\log k)}$ transfers from PSI to SCSS.
A natural question is whether we can close the $O(\log k)$ factor in the exponent: however, our reduction is from the PSI problem and the best known lower bound for PSI also has such a gap~\cite{marx-beat-treewidth}. Note that there are many other parameterized problems for which the only known way of proving almost tight lower bounds is by a similar reduction from PSI, and hence an $O(\log k)$ gap appears for these problems as well~\cite{DBLP:conf/esa/MarxP15,DBLP:journals/jcss/JansenKMS13,DBLP:conf/stoc/CurticapeanDM17,DBLP:journals/siamdm/JonesLRSS17,DBLP:conf/focs/CurticapeanX15,DBLP:conf/esa/BonnetM16,DBLP:journals/algorithmica/GuoHNS13,DBLP:journals/dam/BonnetS17,DBLP:conf/esa/Bringmann0MN16,DBLP:journals/corr/abs-1808-02162,DBLP:journals/corr/LokshtanovRSZ17,DBLP:conf/iwpec/BonnetGL17,logk-1,rajesh-andreas-pasin,fahad-dsn,logk-2}.


Even though Feldman and Ruhl were able to generalize their $n^{O(k)}$ time algorithm from SCSS to DSN, we show that, surprisingly, such a
generalization is not possible for our $2^{O(k)}\cdot n^{O(\sqrt{k})}$ time algorithm for planar SCSS.
\begin{theorem}
\label{thm:dsn-w[1]-hardness} The edge-unweighted version of the \dsn problem is W[1]-hard parameterized by
the number $k$ of terminal pairs, even when the input is restricted to planar
directed acyclic graphs (planar DAGs). Moreover, there is no $f(k)\cdot
n^{o(k)}$ time algorithm for any computable function $f$, unless the ETH fails.
\end{theorem}
This implies that the Feldman-Ruhl algorithm for DSN is optimal, even on planar directed acyclic graphs. As in our lower bound for planar SCSS, the
proof is by a reduction from an instance of the $k\times k$~\gt problem.
However, unlike in the reduction to SCSS where we needed
$O(k^2)$ terminals, the reduction to DSN needs only $O(k)$ pairs of terminals (see Figure~\ref{fig:dsn}). Since the parameter blowup is linear, the
$f(k)\cdot n^{o(k)}$ lower bound for \gt from~\cite{daniel-grid-tiling} transfers to DSN.

\textbf{Remark:} All our hardness results (Theorem~\ref{thm:scss-main-hardness-planar-graphs}, Theorem~\ref{thm:scss-main-hardness-general-graphs} and Theorem~\ref{thm:dsn-w[1]-hardness}) are presented for weighted-edge versions with polynomially-bounded integer weights (including edges with weight zero). By splitting each edge of weight $W$ into $W$ edges of weight one, all the results also hold for the unweighted-edge version. Our algorithm (Theorem~\ref{thm:algo-sqrt-k-h-minor-free}) is presented for the weighted-vertex version.
Appendix~\ref{appendix:vertex-general-than-edge} shows that the unweighted-vertex version is more general than the weighted-edge version. Hence all our lower bounds also hold for the (un)weighted-vertex version too.

Finally, instead of parameterizing by the number of terminals, we can consider parameterization by the number of edges/vertices of the solution. Let us
briefly and informally discuss this parameterization.  Note that the number of terminals is a lower bound on the number of edges/vertices
of the solution (up to a factor of 2 in the case of DSN parameterized by the number of edges), thus fixed-parameter tractability could be
easier to obtain by parameterizing with the number of edges/vertices. However,
our lower bound for SCSS on general graphs
(as well as the W[1]-hardness of Guo et al.~\cite{guo-et-al}) actually proves hardness also
with these parameterizations, making fixed-parameter tractability unlikely. On the other hand, it follows from standard techniques that
both SCSS and DSN are FPT on planar graphs when parameterizing by the number $k$ of edges/vertices in the solution. The main argument here
is that the solution is fully contained in the $k$-neighborhood of the
terminals, whose number is at most $2k$.
It is known that the $k$-neighborhood of $O(k)$ vertices in a planar graph has
treewidth $O(k)$, and thus one can use standard techniques on bounded-treewidth
graphs (dynamic programming or Courcelle's Theorem). Alternatively, at least in the unweighted case, one can formulate the problem as a first
order formula of size depending only on $k$ and then invoke the result of Frick and Grohe~\cite{DBLP:journals/jacm/FrickG01} stating that
such problems are FPT. Therefore, as fixed-parameter tractability is easy to establish on planar graphs, the challenge here is to obtain
optimal dependence on $k$. One would expect that sub-exponential dependence on~$k$
(e.g., $2^{O(\sqrt{k})}$ or $k^{O(\sqrt{k})}$) should be possible at least for SCSS, but this is
not yet fully understood even for undirected
\st~\cite{DBLP:conf/stacs/PilipczukPSL13}. A slightly different parameterization
is to consider the number $k$ of {\em non-terminal} vertices in the solution,
which can be much smaller than the number of terminals. This leads to problems
of somewhat different flavour, see
e.g.~\cite{DBLP:conf/stacs/DvorakFKMTV18,DBLP:journals/siamdm/JonesLRSS17}.

\subsection{Further related work}

Subsequent to the conference version~\cite{DBLP:conf/soda/ChitnisHM14} of this
paper, there have been several related results. Chitnis et al.~\cite{khandekar}
considered a variant of SCSS with only 2 terminals but with a requirement of
multiple paths. Formally, in the $2$-SCSS-$(k_1, k_2)$ problem we are given two
vertices $s,t$ and the goal is to find a min weight subset $H\subseteq E(G)$
such that $H$ has $k_1, k_2$ paths from $s\leadsto t, t\leadsto s$,
respectively. The objective function is given by $\text{cost}(H) = \sum_{e\in H}
\phi(e)\cdot \text{cost}(e)$ where $\phi(e)$ is the maximum number of times $e$
appears on $s\leadsto t$ paths and $t\leadsto s$ paths. Chitnis et
al.~\cite{khandekar} showed that the $2$-SCSS-$(k, 1)$ problem can be solved in
$n^{O(k)}$ time for any $k\geq 1$, and has a $f(k)\cdot n^{o(k)}$ lower bound
under ETH.

Such{\'y}~\cite{suchy-wg} introduced a generalization of DST and SCSS called
the {\sc $q$-Root Steiner Tree} ($q$-RST) problem. In this problem, given a set
of $q$ roots and a set of $k$ leaves, the task is to find a minimum-cost
network where the roots are in the same strongly connected component and every
leaf can be reached from every root. Generalizing the token game of Feldman and
Ruhl~\cite{feldman-ruhl}, Such{\'y}~\cite{suchy-wg} designed a $2^{O(q)}\cdot
n^{O(k)}$ algorithm for $q$-RST.

Recently, Chitnis et al.~\cite{rajesh-andreas-pasin} considered the SCSS and DSN
problems on bidirected graphs: these are directed graphs with the guarantee
that for every edge $(u,v)$ the reverse edge $(v,u)$ exists and has the same
weight. They showed that on bidirected graphs, the DSN problem stays W[1]-hard
parameterized by $k$ but SCSS becomes FPT (while still being NP-hard). In fact,
under ETH, no $f(k)n^{o(k/\log k)}$ time algorithm for DSN on bidirected graphs
exists, and thus the problem is essentially as hard as for general directed
graphs. For bidirected planar graphs however, Chitnis et
al.~\cite{rajesh-andreas-pasin} show that DSN can be solved in
$2^{O(k^{3/2}\log k)}n^{O(\sqrt{k})}$, which is in contrast to
Theorem~\ref{thm:dsn-w[1]-hardness}.
Some FPT approximability and inapproximability results for SCSS and DSN were
also shown in~\cite{rajesh-andreas-pasin,rajesh-andreas-ipec}.

\textbf{Pattern graphs and DSN:} The set of pairs $\{(s_i, t_i)\ :\ i\in [k]\}$ in the input of DSN can be
interpreted as a directed (unweighted) pattern graph on a set
$R=\bigcup_{i=1}^{k} \{s_i, t_i\}$ of terminals. For a graph class
$\mathcal{H}$, the $\mathcal{H}$-DSN problem takes as input a directed graph
$H\in \mathcal{H}$ on vertex set $R$ and the goal is to find a minimum cost
subgraph $N\subseteq E(G)$ such that $N$ has an $s\leadsto t$ path for each
$(s,t)\in E(H)$. Thus for a fixed class $\mathcal{H}$ of pattern graphs, the $\mathcal{H}$-DSN problem is a restricted special case of the general DSN problem, and it is possible that  $\mathcal{H}$-DSN is FPT (for example, if $\mathcal{H}$ is the class of out-stars). Feldmann and Marx~\cite{andreas-dm-arxiv} gave a complete
dichotomy for which graph classes the $\mathcal{H}$-DSN problem is FPT or
W[1]-hard parameterized by $|R|$.

Given an instance of DSN with the pattern graph $H=(R,A)$ on the terminal set
$R$ with $|A|=k$, the algorithm of Feldman and Ruhl~\cite{feldman-ruhl} runs in
$n^{O(k)}$ time. The $f(k)\cdot n^{o(k)}$ lower bound under ETH for DSN in this
paper (Theorem~\ref{thm:dsn-w[1]-hardness}) has $|A|=O(|R|)$. Hence, for the
parameter $|R|$ we have a lower bound of $f(|R|)\cdot n^{o(|R|)}$ and an upper
bound of $n^{O(|R|^2)}$ (since $|A|=O(|R|^2)$ in the worst case). Recently,
Eiben et al.~\cite{fahad-dsn} essentially closed this gap by showing a lower
bound of $f(|R|)\cdot n^{o(|R|^2 /\log |R|)}$ under ETH for DSN. They also gave
an algorithm for DSN on bounded genus graphs: for graphs of genus $g$, the
algorithm runs in $f(|R|)\cdot n^{O_{g}(|R|)}$ time where $O_{g}(\cdot)$ hides
constants depending only on $g$.


%
%
%
%
%


\section{Improved algorithm for SCSS on planar graphs}

In this section we describe the proof to
Theorem~\ref{thm:algo-sqrt-k-h-minor-free}, i.e., we present an algorithm to
solve SCSS on planar graphs in $2^{O(k)}\cdot n^{O(\sqrt{k})}$ time. The definitions
of some of the graph-theoretic notions used in this section such as treewidth
and minors are deferred to Appendix~\ref{appendix:gt-defns} to maintain the flow
of presentation. The key is to analyze the structure of \emph{edge-minimal
solutions}, i.e., subgraphs of the input graph $G$ (induced by some set
$S\subseteq V$) containing all terminals for which no edge can be removed
without also removing all $s\leadsto t$ paths for some terminal pair $(s,t)$. We
show that for an edge-minimal solution $M$ of the SCSS problem there is a vertex set
$W\subseteq V(M)$ of size $O(k)$ such that, after removing $W$ from $M$, each
component has constant treewidth. More formally, we define a
\emph{$W_M$-component} as a connected component of the (underlying undirected) graph induced by
$V(M)\setminus W$ in $M$, and prove the following.




\begin{lemma}\label{lem:W}
For any edge-minimal solution $M$ to the edge-weighted SCSS problem there is a set of at
most $9k$ vertices $W\subseteq V(M)$ for which every $W_M$-component has treewidth at most $2$.
\end{lemma}

We defer the proof of Lemma~\ref{lem:W} to Section~\ref{subsec:bounding-treewidth}. First, we see how we can use Lemma~\ref{lem:W} to bound the treewidth of the minimal solution $M$.

\begin{lemma}
If an edge-minimal solution $M$ to edge-weighted SCSS is planar (or excludes some fixed minor),
then its treewidth is~$O(\sqrt{k})$.
\label{lem:tw-sqrt-k}
\end{lemma}
\begin{proof}
By the planar grid theorem~\cite{planar-grid-theorem}, there is a constant $c_{\text{Planar}}$ such that any planar graph $G$ with treewidth $c_{\text{Planar}}\cdot \omega$ has a $\omega \times \omega$ grid minor.
If the treewidth of $M$ is at least $c_{\text{Planar}}\cdot \lceil 20\sqrt{k} \rceil$, then it follows that $M$ has a $\lceil 20\sqrt{k} \rceil\times
\lceil 20\sqrt{k} \rceil$ grid minor $M'$. It is easy to see that $M'$ contains at least $\lfloor \frac{\lceil 20\sqrt{k}\rceil}{3} \rfloor   \cdot \lfloor \frac{\lceil 20\sqrt{k}\rceil}{3} \rfloor$ (pairwise vertex-disjoint) grids of size~$3\times 3$. For each $k\geq 1$, one can easily verify that $\lfloor \frac{\lceil 20\sqrt{k}\rceil}{3} \rfloor \geq \lceil 4\sqrt{k}\rceil$, and hence the number of pairwise vertex-disjoint $3\times 3$ grids is at least $\lceil 4\sqrt{k}\rceil \cdot \lceil 4\sqrt{k}\rceil \geq 4\sqrt{k}\cdot 4\sqrt{k} =16k$.
By Lemma~\ref{lem:W}, there is a set of vertices $W$ of size $9k$
whose deletion makes every $W_M$-component have treewidth at most $2$.
Since $16k>9k$, it follows that $W$ does not contain a vertex from at least one
of the (pairwise vertex-disjoint) $16k$ grid minors of size $3\times 3$ in~$M$. Hence, there is a
$W_{M}$-component, which contains a $3\times 3$ grid minor, and hence has
treewidth at least~$3$, which is a contradiction.

For the case when the input graph is $H$-minor-free for some fixed graph $H$,
we can instead use the excluded grid-minor theorem of Demaine and
Hajiaghayi~\cite{H-minor-free-grid-theorem}: for any fixed graph $H$, there is a constant $c_{H}$ (which depends only on $|H|$) such that any $H$-minor-free graph of treewidth at least $c_{H}\cdot \omega$ has a
$\omega\times \omega$ grid as a minor.
\end{proof}


To prove Theorem~\ref{thm:algo-sqrt-k-h-minor-free}, which is restated below,
we invoke an algorithm of~\cite{andreas-dm-arxiv} to find the optimum solution
of bounded treewidth. The algorithm of~\cite{andreas-dm-arxiv} is designed for the edge-weighted version, and we state below the corresponding statement for the more general unweighted vertex version (so that it may also be of future use).

\begin{theorem}
\label{thm:dm-andreas-vertex-version} (\textbf{generalization of~\cite[Theorem~5]{andreas-dm-arxiv}})
If there is an optimum solution to an instance on $k$ terminals of the vertex-weighted version of SCSS
which has treewidth at most $\omega$, then an optimum solution\footnote{Not necessarily the same optimum solution as the one mentioned in the first part of this theorem. For example, the actual optimum found by this algorithm might have treewidth much larger than $\omega$.} can be found in $2^{O(k+\omega \cdot \log
\omega)}\cdot n^{O(\omega)}$ time.
\end{theorem}
\begin{proof}
In the given graph $G$, we start by subdividing each edge by adding a non-terminal vertex of weight $0$ (note that this does not increase the treewidth). Let us call these vertices we have added as dummy vertices, and the graph obtained at this point be $G^*$. Note that each dummy vertex has in-degree one and out-degree one. Now we reduce the vertex-weighted version of SCSS to the edge-weighted
version, using a standard reduction: substitute each non-terminal vertex $u\in G$ of weight $W$ with
two new non-terminal vertices $u^-$ and $u^+$ and an edge $(u^-,u^+)$ of the same weight $W$.
Every edge that had $u$ as its head will now have $u^-$ as its head instead, and
every edge that had $u$ as its tail will now have $u^+$ as its tail. We set the weight of all these
edges to be zero. Let the graph obtained after these modifications be $G'$.

Consider an optimum solution $S$ for the vertex-weighted version of SCSS, and without loss of generality we can assume that $S$ is minimal under vertex deletions (if it is not, then make it minimal by deleting unnecessary vertices). Let $S' = (S\cap T)\cup \{\{u^-, u^+\}\ : u\in S\cap (V\setminus T)\}$. We now show that the induced graph $G'[S']$ is an edge-minimal solution (with same weight as that of $S$) for the edge-weighted version of SCSS: we do this by showing that deletion of any edge from $G'[S']$ creates a non-terminal source or a non-terminal sink which contradicts the fact that $S$ was a vertex-minimal solution for vertex-weighted version of SCSS. The construction of the graph $G'$ from $G$ implies that any edge $e$ in $G$ must be of either of the following two types:
\begin{itemize}
  \item Without loss of generality\footnote{The other case is the edge being $(v^+,y)$ for some dummy vertex $y$ and some non-terminal $v\in G$}, the edge is $(y,v^-)$ for some dummy vertex $y$ and some non-terminal $v\in G$ in which case deleting this edge makes the non-terminal $y$ to be a sink.
  \item The edge is $(z^-,z^+)$ for some non-terminal $z\in G$ in which case deleting this edge makes the non-terminal $z^+$ a source and the non-terminal $z^-$ a sink.
\end{itemize}
Note that $G'$ is not necessarily planar (or $H$-minor-free)
even if $G$ is. However, the treewidth of $G'[S']$ is at most twice the treewidth of
$G[S]$ since we can simply replace each non-terminal vertex $u$ in the bags of the tree
decomposition of $N$ by the two vertices $u^-$ and $u^+$.
Feldmann and Marx~\cite[Theorem~5]{andreas-dm-arxiv} showed that if the optimum
solution to an instance on $k$ terminals of the edge-weighted version of SCSS
has treewidth $\omega$, then it can be found in $2^{O(k+\omega \cdot \log
\omega)}\cdot n^{O(\omega)}$ time. Hence, the claimed running time for the vertex-weighted version follows.
\end{proof}


Finally, we are now ready to prove Theorem~\ref{thm:algo-sqrt-k-h-minor-free}

\begin{reptheorem}{thm:algo-sqrt-k-h-minor-free}
An instance $(G,T)$ of the vertex-weighted \scss problem with $|G|=n$ and $|T|=k$ can be solved
in $2^{O(k)}\cdot n^{O(\sqrt{k})}$ time, when the underlying undirected graph of
$G$ is planar (or more generally, $H$-minor-free for any fixed graph $H$).
\end{reptheorem}
\begin{proof}
%
Consider a subgraph $M$ of $G$ induced by the optimum solution $S\subseteq V$,
which is also minimal, i.e., no edge of $M$ can be removed without destroying
the connectivity between some terminal pair $(s,t)$. By
Lemma~\ref{lem:tw-sqrt-k} we know that the treewidth of $M$ is $O(\sqrt{k})$. Hence, the claimed running time follows from Theorem~\ref{thm:dm-andreas-vertex-version}.
%
\end{proof}

Note that Lemma~\ref{lem:tw-sqrt-k} only used the planarity (or
$H$-minor-freeness) of $M$, and not of the input graph. Hence, the algorithm of
Theorem~\ref{thm:algo-sqrt-k-h-minor-free} also works for the weaker restriction
of finding an optimal planar (or $H$-minor-free) solution in an otherwise
unrestricted input graph, rather than finding an optimal solution in a planar
(or $H$-minor-free respectively) graph. It only remains to prove
Lemma~\ref{lem:W}, which is done in the next section.

\subsection{Proof of Lemma~\ref{lem:W}}
\label{subsec:bounding-treewidth}

Fix an arbitrary terminal $r\in T$. It is easy to see (observed for example by Feldman and Ruhl~\cite{feldman-ruhl}) that any minimal SCSS
solution $M$ is the union of an in-arborescence $A_{\inn}$ and an
out-arborescence~$A_{\out}$, both having the same root $r\in T$ and
only terminals as leaves, since every terminal of $T$ can be reached from $r$,
and conversely every terminal can reach $r$ in $M$. We construct the set $W_M$ by including
three different kinds of vertices. First, $W_M$ contains every \emph{branching
point} of $A_{\inn}$ and $A_{\out}$, i.e.\ every vertex with in-degree at least
$2$ in $A_{\inn}$ and every vertex with out-degree at least $2$ in~$A_{\out}$.
Since $A_{\inn}$ and $A_{\out}$ are arborescences with at most $k$ leaves (the
terminals), they each have at most $k$ branching points. Secondly, $W_M$
contains all terminals of $T$, which adds another $k$ vertices to the
set~$W_M$. The third kind of vertices in $W_M$ is the following. Note that
every component of the intersection of $A_{\inn}$ and $A_{\out}$ forms a path (possibly consisting only of a single vertex),
since
every vertex of $A_{\inn}$ has out-degree at most~$1$, while every vertex of
$A_{\out}$ has in-degree at most~$1$. We call such a component a \emph{shared
path}. If a shared path contains a branching point or a terminal, we add the endpoints of the shared path to $W_M$. For a branching point or terminal~$v$ on such a shared path,
we can map the endpoints of the shared path to $v$. This maps at most two
endpoints of shared paths to each branching point or terminal, so that the
number of vertices of the third kind in $W_M$ is at most $6k$ (as there are $k$
terminals and at most $2k$ branching points). Thus the total size of $W_M$ is at
most~$9k$.

%

We claim that every $W_M$-component consists of at most two interacting paths,
one from $A_{\inn}$ and one from $A_{\out}$. More formally, consider a $u\to v$
path $P$ of $A_{\inn}$ such that $u$ and $v$ are either terminals or branching
points of $A_{\inn}$, and such that no internal vertex of $P$ is a terminal or
branching point of $A_{\inn}$. We call any such path $P$ an \emph{essential path
of} $A_{\inn}$. Note that we ignore the branching points of $A_{\out}$ in this
definition, and that the edge set of the arborescence~$A_{\inn}$ is the disjoint
union of the edge sets of its essential paths. Analogously we define the
\emph{essential paths of} $A_{\out}$ as those $u\to v$ paths $P$ in~$A_{\out}$
for which $u$ and $v$ are terminals or branching points of $A_{\out}$, and no
internal vertices of $P$ are of such a type.

\begin{claim}\label{clm:2paths}
Every $W_M$-component contains edges of at most two essential paths, one
from~$A_{\inn}$ and one from~$A_{\out}$.
\end{claim}
\begin{proof}
Any vertex at which two essential paths of the same arborescence intersect is a
terminal or branching point. These vertices are in $W_M$ and therefore not
contained in any $W_M$-component. Thus if a $W_M$-component $H$ contains at
least
two essential paths then they either coincide on every edge of $H$, in
which case the claim is clearly true, or $H$ contains the endpoint $v$ of a shared path, i.e., there are two essential paths, one
from each arborescence, that both contain vertex $v$. We will show that there is only one pair of
essential paths that can meet at an endpoint of a shared path in $H$, from which
the claim follows.

In order to prove this, we label every essential $u\to v$ path $P$ of $A_{\inn}$ with those
terminals $T_P\subseteq T$ that can reach the start vertex of $P$ in the in-arborescence, i.e.\
$t\in T_P$ if and only if there exists a $t\to u$ path in $A_{\inn}$. Note that no
two essential paths of $A_{\inn}$ can have the same label. We also label any
essential $u\to v$ path $Q$ of $A_{\out}$ analogously, by setting the label
$T_Q\subseteq T$ to be the terminals which can be reached from the end vertex of $Q$ in the
out-arborescence, i.e.\ there is a $v\to t$ path in~$A_{\out}$ if and only if
$t\in T_Q$. Even though no two essential paths of an individual arborescence
have the same label, there can be pairs of essential paths from $A_{\inn}$ and
$A_{\out}$ with the same label. Let $P$ and $Q$ be essential paths of $A_{\inn}$
and $A_{\out}$, respectively. We prove that if $P$ and $Q$ meet at an endpoint
$v$ of a shared path, then $v\in W_M$ or $T_P=T_Q$.

Assume this is not the case so that $v\notin W_M$ and $T_P\neq T_Q$. Let $I$ be
the shared path in the intersection of $A_{\inn}$ and $A_{\out}$ for which $v$ is
an endpoint. If $u$ is the other endpoint of $I$, assume w.l.o.g.\ that $I$ is
a $u\to v$ path (the other case is symmetric). If there were any branching
points or terminals on $I$ then~$v\in W_M$, since $v$ would then be one of the
third kind of vertices in $W_M$. As this is not the case, $I$ lies in the
intersection of $P$ and $Q$, there are edges $e_v\in E(P)$ and $f_v\in E(Q)$
leaving $v$ such that $e_v\notin E(A_{\out})$ and $f_v\notin E(A_{\inn})$, and
there are edges $e_u\in E(P)$ and $f_u\in E(Q)$ entering $u$ such that
$e_u\notin A_{\out}$ and $f_u\notin E(A_{\inn})$.

As $T_P\neq T_Q$ there is a terminal $t$ contained in one of the two sets but
not the other. Consider the case when $t\in T_Q\setminus T_P$, i.e.\ there is a
$v\to t$ path in $A_{\out}$ but no $t\to u$ path in~$A_{\inn}$. The latter implies
that $e_v$ cannot be reached from $t$ in $A_{\inn}$, as the $u\to v$ path $I$
contains no branching point of $A_{\inn}$. The in-arborescence $A_{\inn}$ does
however contain a $t\to r$ path from $t$ to the root~$r$. Since $e_v\notin
E(A_{\out})$, this means that the root $r$ can be reached from $v$ through the
$v\to t$ path of $A_{\out}$ and the $t\to r$ path without passing through~$e_v$.
Hence $e_v$ can safely be removed without making the solution $M$ infeasible.
This contradicts the minimality of $M$.

In case $t\in T_P\setminus T_Q$ a symmetric argument shows that the edge $f_u$
is redundant in $M$, which again contradicts its minimality. We have thus shown
that $P$ and $Q$ are the only essential paths that meet in any endpoint of a
shared path in the $W_M$-component $H$. Hence $H$ consists of exactly these two
paths $P$ and $Q$, and the claim follows.
\cqed\end{proof}

Consider the case when there is at most one shared path of $M$ that intersects
with a $W_M$-component~$H$. Since by Claim~\ref{clm:2paths}, $H$ consists of at
most
two essential paths, it is easy to see that in this case $H$ is a tree, and
thus its treewidth is $1$. If at least two shared paths of $M$ intersect with
$H$, by Claim~\ref{clm:2paths}, $H$~contains edges of two essential paths $P$ and
$Q$ of $A_{\inn}$ and $A_{\out}$ respectively. To show that in this case the
treewidth of $H$ is at most $2$, we need the following observation on $P$ and
$Q$:

\begin{claim}\label{clm:order}
Let $I_1,\ldots,I_h$ be the connected components in the intersection of $P$ and
$Q$, ordered in the way that $P$ visits them, i.e.\ for any $i\in\{1,\ldots,
h-1\}$ there is a subpath of $P$ with prefix $I_i$ and suffix~$I_{i+1}$. The
path $Q$ visits the shared paths in the opposite order, i.e.\ for any
$i\in\{1,\ldots, h-1\}$ there is a subpath of $Q$ with suffix $I_i$ and prefix
$I_{i+1}$.
\end{claim}
\begin{proof}
Assume this is not the case, so that there is an index $i\in\{1,\ldots, h-1\}$
such that both $P$ and $Q$ contain subpaths with prefix $I_i$ and suffix
$I_{i+1}$. This means that there are edges $e\in E(P)\setminus E(Q)$ and $f\in
E(Q)\setminus E(P)$ which share the last vertex $u$ of $I_i$. Hence $Q$ contains
a $u\to v$ subpath $Q'$ to the first vertex $v$ of $I_{i+1}$, which does not contain
the edge $e$, and also $P$ contains a $u\to v$ subpath $P'$, which does not
contain the edge $f$. As $Q$, and therefore also $Q'$, contains no branching
point of $A_{\out}$, any terminal reachable from $u$ through $Q'$ in $A_{\out}$ is
also reachable from $u$ via the detour $P'$. We can therefore remove edge $f\in
E(A_{\out})$ without violating the feasibility of $M$. This however contradicts
its minimality.
\cqed
\end{proof}

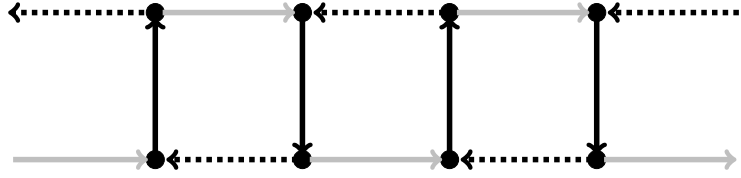
\begin{figure}[t]
\centering
\resizebox{10.0cm}{!}{

\begin{tikzpicture}[rotate=90]

\centering

\draw [black] plot [only marks, mark size=7, mark=*] coordinates {(4,0)}
node[label={[xshift=0mm,yshift=0mm] $$}] {} ;

\draw [black] plot [only marks, mark size=7, mark=*] coordinates {(8,0)}
node[label={[xshift=0mm,yshift=0mm] $$}] {} ;

\draw [black] plot [only marks, mark size=7, mark=*] coordinates {(4,-4)}
node[label={[xshift=0mm,yshift=0mm] $$}] {} ;

\draw [black] plot [only marks, mark size=7, mark=*] coordinates {(8,-4)}
node[label={[xshift=0mm,yshift=0mm] $$}] {} ;

\draw [black] plot [only marks, mark size=7, mark=*] coordinates {(4,-8)}
node[label={[xshift=0mm,yshift=0mm] $$}] {} ;

\draw [black] plot [only marks, mark size=7, mark=*] coordinates {(8,-8)}
node[label={[xshift=0mm,yshift=0mm] $$}] {} ;

\draw [black] plot [only marks, mark size=7, mark=*] coordinates {(4,-12)}
node[label={[xshift=0mm,yshift=0mm] $$}] {} ;

\draw [black] plot [only marks, mark size=7, mark=*] coordinates {(8,-12)}
node[label={[xshift=0mm,yshift=0mm] $$}] {} ;


\path (4,0) node(a) {} (8,0) node(b) {};
        \draw[line width=1.5mm,black,->] (a) -- (b);

\path (4,-4) node(a) {} (8,-4) node(b) {};
        \draw[line width=1.5mm,black,->] (b) -- (a);

\path (4,-8) node(a) {} (8,-8) node(b) {};
        \draw[line width=1.5mm,black,->] (a) -- (b);

\path (4,-12) node(a) {} (8,-12) node(b) {};
        \draw[line width=1.5mm,black,->] (b) -- (a);


\path (4,4) node(a) {} (4,0) node(b) {};
        \draw[line width=1.5mm,lightgray,->] (a) -- (4,0.2);

\path (8,0) node(a) {} (8,-4) node(b) {};
        \draw[line width=1.5mm,lightgray,->] (8,-0.2) -- (8,-3.8);

\path (4,-4) node(a) {} (4,-8) node(b) {};
        \draw[line width=1.5mm,lightgray,->] (4,-4.2) -- (4,-7.8);

\path (8,-8) node(a) {} (8,-12) node(b) {};
        \draw[line width=1.5mm,lightgray,->] (8,-8.2) -- (8,-11.8);

\path (4,-12) node(a) {} (4,-16) node(b) {};
        \draw[line width=1.5mm,lightgray,->] (4,-12.2) -- (4,-15.8);


\path (8,0) node(a) {} (8,4) node(b) {};
        \draw[line width=1.5mm,loosely dotted,->] (8,0.3) -- (8,4);

\path (4,-4) node(a) {} (4,0) node(b) {};
        \draw[line width=1.5mm,loosely dotted,->] (a) -- (4,-0.3);

\path (8,-8) node(a) {} (8,-4) node(b) {};
        \draw[line width=1.5mm,loosely dotted,->] (a) -- (8,-4.3);

\path (4,-12) node(a) {} (4,-8) node(b) {};
        \draw[line width=1.5mm,loosely dotted,->] (a) -- (4,-8.3);

\path (8,-16) node(a) {} (8,-12) node(b) {};
        \draw[line width=1.5mm,loosely dotted,->] (a) -- (8,-12.3);

\end{tikzpicture}
}
\caption{Let $H$ be a $W_M$-component $M$ such that at least two shared paths of
$M$ intersect with $H$. By Claim~\ref{clm:2paths}, $H$~contains edges of two
essential paths $P$ and $Q$ of $A_{\inn}$ and $A_{\out}$, respectively.
Claim~\ref{clm:order} shows that $H$ looks roughly as shown in the figure: let
$P, Q$ be denoted by light+black paths and dotted+black paths respectively. Then the
black paths are exactly the shared paths. Note though that a shared path may
have length~$0$.
}
\label{fig:tw-2}
\end{figure}

Claim~\ref{clm:order} implies that the structure of $H$ is roughly as shown in
Figure~\ref{fig:tw-2} (the black edges shown in Figure~\ref{fig:tw-2} correspond
to paths of length~$0$ or more, while the light and dotted edges correspond to paths
of length at least~$1$). 
If we contract each path of length at least $1$ to a path of length $1$, then
the resulting graph is planar and all vertices belong to the outer face. Such
graphs are called outerplanar graphs. In other words, $H$ is a subdivision of
an outerplanar graph. Lemma~\ref{lem:tw-outerplanar} shows that treewidth of
subdivisions of outerplanar graphs is at most $2$, 
which proves Lemma~\ref{lem:W}.

\section{W[1]-hardness for SCSS in planar graphs}

The goal of this section is to prove Theorem~\ref{thm:scss-main-hardness-planar-graphs}. We reduce from the \gt problem\footnote{The \gt problem has been defined in two (symmetrical) ways in the literature: either the first coordinate or the second coordinate remains the same in a row. Here, we follow the notation of~\cite{daniel-grid-tiling}, but the other definition also appears in some places (e.g.~\cite{fpt-book}).} introduced by Marx~\cite{daniel-grid-tiling}:

\begin{center}
\noindent\framebox{\begin{minipage}{6.00in}
\textbf{$k\times k$~\textsc{Grid Tiling}}\\
\emph{Input }: Integers $k, n$, and $k^2$ non-empty sets $S_{i,j}\subseteq [n]\times [n]$ where $1\leq i, j\leq k$\\
\emph{Question}: For each $1\leq i, j\leq k$ does there exist an entry $\gamma_{i,j}\in S_{i,j}$ such that
\begin{itemize}
\item If $\gamma_{i,j}=(x,y)$ and $\gamma_{i,j+1}=(x',y')$ then $x=x'$.
\item If $\gamma_{i,j}=(x,y)$ and $\gamma_{i+1,j}=(x',y')$ then $y=y'$.
\end{itemize}
\end{minipage}}
\end{center}

Under ETH~\cite{eth,eth-2}, it was shown by Chen et
al.~\cite{chen-hardness} that $k$-\textsc{Clique}\footnote{The
$k$-\textsc{Clique} problem asks whether there is a clique of size $\geq k$.}
does not admit an algorithm running in time $f(k)\cdot n^{o(k)}$ for any computable
function $f$. There is a simple reduction~\cite[Theorem 14.28]{fpt-book} from
$k$-\textsc{Clique} to $k\times k$~\gt implying the same runtime lower bound
for the latter problem.
%
%
To prove Theorem~\ref{thm:scss-main-hardness-planar-graphs}, we give a reduction which transforms the problem of $k\times k$
\gt into a planar instance of SCSS with $O(k^2)$ terminals.
We design two types of gadgets: the \emph{connector gadget} and the \emph{main gadget}. The reduction from \gt represents each cell of the grid with a copy of the main gadget, with a connector gadget between main gadgets that are adjacent either horizontally or vertically (see Figure~\ref{fig:big-picture}).

The proof of Theorem~\ref{thm:scss-main-hardness-planar-graphs} is divided into the following steps: In
Section~\ref{sec:connector:exist}, we first introduce the connector gadget and Lemma~\ref{lem:connector-gadget} states the
existence of a particular type of connector gadget. In Section~\ref{sec:main:exists}, we introduce the main gadget and
Lemma~\ref{lem:main-gadget} states the existence of a particular type of main gadget. Section~\ref{sec:construction-of-g^*-scss-planar} describes the construction of the planar instance $(G^*, T^*)$ of SCSS. The two directions implying the reduction from \gt to SCSS are proved in Section~\ref{proof:scss-planar-easy} and Section~\ref{proof:scss-planar-hard} respectively. Using Lemmas~\ref{lem:connector-gadget}
and ~\ref{lem:main-gadget} as a blackbox, we prove Theorem~\ref{thm:scss-main-hardness-planar-graphs} in
Section~\ref{sec:main-scss-planar-proof-subsec}. The proofs of Lemmas~\ref{lem:connector-gadget} and Lemma~\ref{lem:main-gadget} are given in
Sections~\ref{connector:proof} and \ref{main:proof} respectively.

\subsection{Existence of connector gadgets}
\label{sec:connector:exist}

A connector gadget $CG_{n}$ is a directed (embedded) planar graph with $O(n^2)$ vertices and positive integer weights\footnote{Weights are polynomial in $n$.} on its edges. It has a total of
$2n+2$ distinguished vertices divided into the following 3 types:
\begin{itemize}
\item The vertices $p, q$ are called \emph{internal-distinguished} vertices
\item The vertices $p_1, p_2, \ldots, p_n$ are called \emph{source-distinguished} vertices
\item The vertices $q_1, q_2, \ldots, q_n$ are called \emph{sink-distinguished} vertices
\end{itemize}
Let $P=\{p_1, p_2,\ldots, p_n\}$ and $Q=\{q_1, q_2,\ldots, q_n\}$. The vertices $P\cup Q$ appear in the clockwise order $p_1$, $\dots$, $p_n$, $q_n$, $\dots$, $q_1$ on the boundary of the gadget. In the connector gadget $CG_n$, every vertex in $P$ is a source
and has exactly one outgoing edge. Also every vertex in $Q$ is a sink and has exactly one incoming edge.

\begin{definition}
An edge set $E'\subseteq E(CG_{n})$ satisfies the \textbf{\emph{connectedness}} property if
each of the following four conditions hold for the graph $CG_{n}[E']$:
\begin{enumerate}
\item $p$ can be reached from some vertex in $P$
\item $q$ can be reached from some vertex in $P$
\item $p$ can reach some vertex in $Q$
\item $q$ can reach some vertex in $Q$
\end{enumerate}
\label{defn:connectedness}
\end{definition}

\begin{definition}
An edge set $E'$ satisfying the connectedness property \textbf{represents} an integer $i\in [n]$ if in $E'$ the only
outgoing edge from $P$ is the one incident to $p_i$ and the only incoming edge into $Q$ is the one incident to $q_i$.
\label{defn:represents-connector}
\end{definition}

The next lemma shows we can construct a particular type of connector gadget:

\begin{lemma}
Given an integer $n$ one can construct in polynomial time a connector gadget $CG_{n}$ and an integer $C^*_{n}$ such that the following two properties hold \footnote{We use the notation $C^*_n$ to emphasize that $C^*$ depends only on $n$.}:
\begin{enumerate}
\item For every $i\in [n]$, there is an edge set $E_i \subseteq E(CG_{n})$ of weight $C^*_n$ such that $E_i$ satisfies the
    connectedness property and represents $i$. Note that, in particular, $E_i$ contains a $p_i \leadsto q_i$ path (via $p$ or $q$).
\item If there is an edge set $E'\subseteq E(CG_{n})$ such that $E'$ has weight at most $C^*_n$ and $E'$ satisfies the connectedness
    property, then $E'$ has weight exactly $C^*_n$ and it represents some $i\in [n]$.
\end{enumerate}
\label{lem:connector-gadget}
\end{lemma}

\subsection{Existence of main gadgets}
\label{sec:main:exists}

A main gadget $MG$ is a directed (embedded) planar graph with $O(n^3)$ vertices and positive integer weights on its edges. It has $4n$ distinguished
vertices given by the following four sets:
\begin{itemize}
\item The set $L=\{\ell_1, \ell_2, \ldots, \ell_n\}$ of \emph{left-distinguished} vertices.
\item The set $R=\{r_1, r_2, \ldots, r_n\}$ of \emph{right-distinguished} vertices.
\item The set $T=\{t_1, t_2, \ldots, t_n\}$ of \emph{top-distinguished} vertices.
\item The set $B=\{b_1, b_2, \ldots, b_n\}$ of \emph{bottom-distinguished} vertices.
\end{itemize}

The distinguished vertices appear in the (clockwise) order $t_1$, $\dots$, $t_n$,
$r_1$, $\dots$, $r_n$, $b_n$, $\dots$, $b_1$, $\ell_n$, $\dots$,
$\ell_1$ on the boundary of the gadget.  In the main gadget $MG$, every
vertex in $L\cup T$ is a source and has exactly one outgoing
edge. Also each vertex in $R\cup B$ is a sink and has exactly one
incoming edge.

\begin{definition}
An edge set $E'\subseteq E(MG)$ satisfies the \textbf{\emph{connectedness}} property if
 each of the following four conditions hold for the graph $MG[E']$:
\begin{enumerate}
\item There is a directed path from some vertex in $L$ to $R\cup B$
\item There is a directed path from some vertex in $T$ to $R\cup B$
\item Some vertex in $R$ can be reached from $L\cup T$
\item Some vertex in $B$ can be reached from $L\cup T$
\end{enumerate}
\label{defn:source-sink-connectivity}
\end{definition}

\begin{definition}
An edge set $E'\subseteq E(MG)$ \textbf{represents} a pair $(i,j)\in
[n]\times [n]$ if each of the following five conditions holds:
\begin{itemize}
\item The only edge of $E'$ leaving $L$ is the one incident to $\ell_i$
\item The only edge of $E'$ entering $R$ is the one incident to $r_i$
\item The only edge of $E'$ leaving  $T$ is the one incident to $t_j$
\item The only edge of $E'$ entering $B$ is the one incident to $b_j$
\item $E'$ contains an $\ell_i \leadsto r_i$ path and an $t_j \leadsto b_j$ path
\end{itemize}
\label{defn:represents-main}
\end{definition}

The next lemma shows we can construct a particular type of main gadget:

\begin{lemma}
Given a subset $S\subseteq [n]\times [n]$, one can construct in polynomial time a main gadget $MG_{S}$ and an integer
$M^*_n$ such that the following two properties hold \footnote{We use the notation $M^*_n$ to emphasize that $M^*$ depends only on $n$, and not on the set $S$.}:
\begin{enumerate}
\item For every $(i,j)\in S$ there is an edge set $E_{i,j} \subseteq E(MG_{S})$ of weight $M^*_n$ such that $E_{i,j}$
    represents $(i,j)$. Note that the last condition of Definition~\ref{defn:represents-main} implies that $E_{i,j}$ satisfies the connectedness property.
\item If there is an edge set $E'\subseteq E(MG_{S})$ such that $E'$ has weight at most $M^*_n$ and satisfies the connectedness
property, then $E'$ has weight exactly $M^*_n$ and represents some $(i,j)\in S$.
\end{enumerate}
\label{lem:main-gadget}
\end{lemma}


%
%

\subsection{Construction of the SCSS instance}
\label{sec:construction-of-g^*-scss-planar}

In order to prove Theorem~\ref{thm:scss-main-hardness-planar-graphs}, we reduce from the \textsc{Grid Tiling} problem. The following assumption will be helpful in handling some of the
border cases of the gadget construction. We may assume that $1 < x,y< n$ holds for every $(x,y)\in S_{i,j}$: indeed, we can increase $n$ by two
and replace every $(x,y)$ by $(x+1,y+1)$ without changing the problem. This is just a minor technical modification\footnote{For the interested reader, what this modification does is to ensure no shortcut edge added in Section~\ref{subsec:description-of-edges-in-main-gadget} has either endpoint on the unbounded face of the planar embedding of the main gadget provided in Figure~\ref{fig:main-gadget}. This helps to streamline the proofs by avoiding the need to have to consider any special cases.} which is introduced to make some of the arguments easier in Section~\ref{main:proof} cleaner.


Given an instance $(k,n, \{S_{i,j}\ :\ i,j\in [k]\})$ of \gt, we construct an instance $(G^*, T^*)$ of SCSS the following way (see Figure~\ref{fig:big-picture}):

 \begin{figure}[t]
 \centering
 \includegraphics[height=5in]{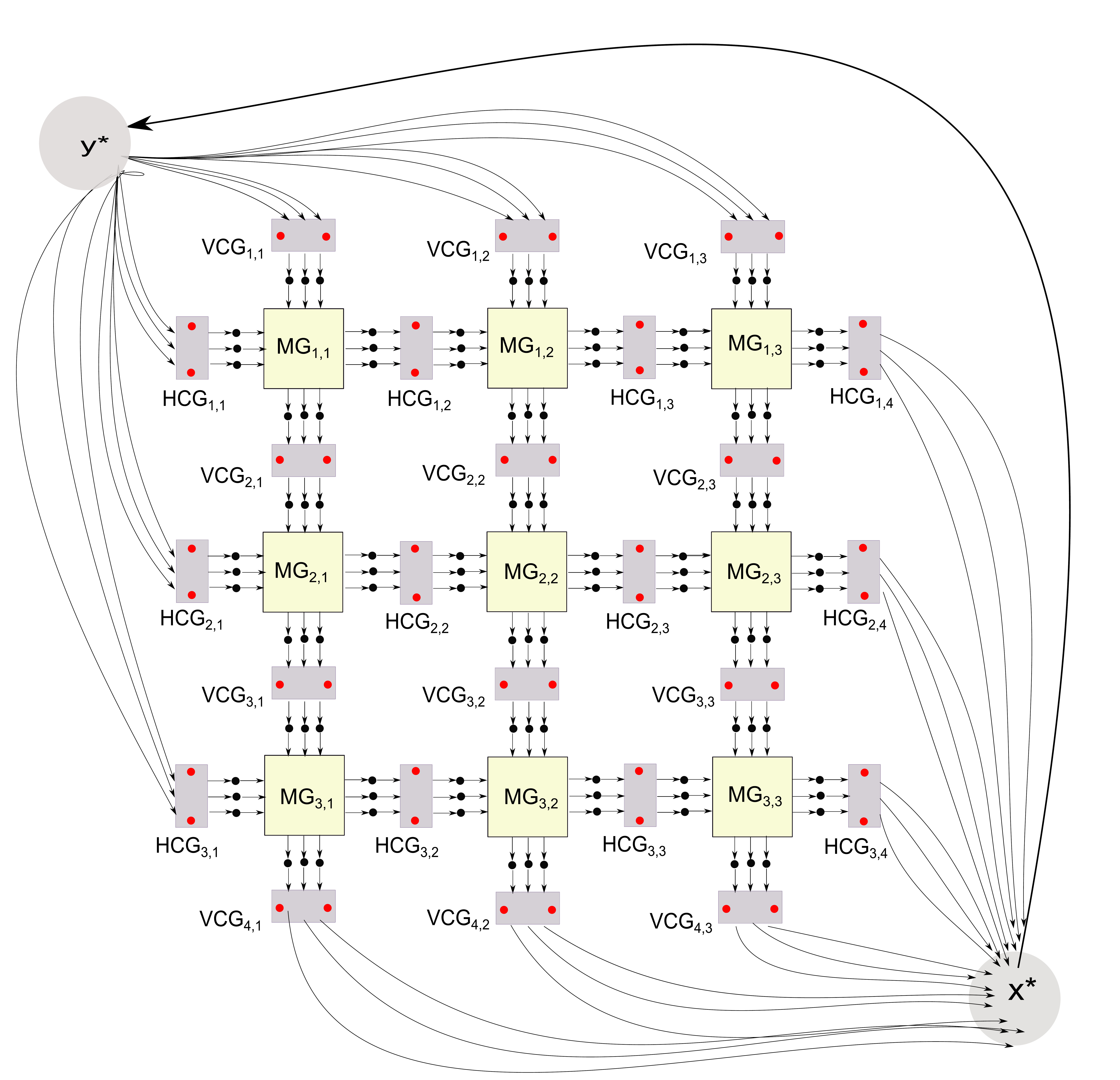}
 \caption{An illustration of the reduction from \textsc{Grid Tiling} to SCSS on planar graphs.
 \label{fig:big-picture}}
 \end{figure}

\begin{itemize}
\item We introduce a total of $k^2$ main gadgets and $2k(k+1)$ connector gadgets.
\item For every set $S_{i,j}$ in the \gt instance, we construct a main gadget $MG_{i,j}$ using
    Lemma~\ref{lem:main-gadget} for the subset $S_{i,j}$.
\item Half of the connector gadgets have the same orientation as in Figure~\ref{fig:connector-gadget} (with the $p_i$ vertices on the top side and the $q_i$ vertices on the bottom side), and we call them
    $HCG$ to denote \emph{horizontal connector gadgets}\footnote{The horizontal connector gadgets are so called because they connect things horizontally as seen by the reader.}. The other half of the connector gadgets are rotated anti-clockwise by 90 degrees
    with respect to the orientation of Figure~\ref{fig:connector-gadget}, and we call them $VCG$ to denote \emph{vertical
    connector gadgets}. The internal-distinguished vertices of the connector gadgets are shown in
    Figure~\ref{fig:big-picture}.
\item For each $1\leq i,j\leq k$, the main gadget $MG_{i,j}$ is surrounded by the following four connector gadgets:
    \begin{enumerate}
    \item The \emph{vertical connector gadgets} $VCG_{i,j}$ is on the top and $VCG_{i+1,j}$ is on the bottom.
        Identify (or glue together) each sink-distinguished vertex of $VCG_{i,j}$ with the top-distinguished vertex of
        $MCG_{i,j}$ of the same index. Similarly identify each source-distinguished vertex of $VCG_{i+1,j}$ with the
        bottom-distinguished vertex of $MCG_{i,j}$ of the same index.
    \item The \emph{horizontal connector gadgets} $HCG_{i,j}$ is on the left and $HCG_{i,j+1}$ is on the right.
        Identify (or glue together) each sink-distinguished vertex of $HCG_{i,j}$ with the left-distinguished vertex
        of $MCG_{i,j}$ of the same index. Similarly, identify each source-distinguished vertex of $HCG_{i,j+1}$ with
        the right-distinguished vertex of $MCG_{i,j}$ of the same index.
    \end{enumerate}
\item We introduce two special vertices $x^*, y^*$ and add an edge $(x^*, y^*)$ of weight 0.
\item For each $1\leq i\leq k$, add an edge of weight 0 from $y^*$ to each source-distinguished vertex of the vertical connector gadget $VCG_{1,i}$.
\item For each $1\leq j\leq k$, add an edge of weight 0 from $y^*$ to each source-distinguished vertex of the horizontal connector gadget $HCG_{j,1}$.
\item For each $1\leq i\leq k$, add an edge of weight 0 from each sink-distinguished vertex of the vertical connector gadget $VCG_{k+1,i}$ to $x^*$.
\item For each $1\leq j\leq k$, add an edge of weight 0 from each sink-distinguished vertex of the horizontal connector gadget $HCG_{j,k+1}$ to $x^*$.

\item For each $i\in [k], j\in [k+1]$, denote the two internal-distinguished vertices of $HCG_{i,j}$ by $\{p^{h}_{i,j}, q^{h}_{i,j}\}$
\item For each $i\in [k+1], j\in [k]$, denote the two internal-distinguished vertices of $VCG_{i,j}$ by $\{p^{v}_{i,j},
    q^{v}_{i,j}\}$

\item The set of terminals $T^*$ for the SCSS instance on $G^*$ is $\{x^*, y^*\}\cup \{p^{h}_{i,j}, q^{h}_{i,j}\ |\ 1\leq i\leq k+1,
    1\leq j\leq k\}\cup \{p^{v}_{i,j}, q^{v}_{i,j}\ |\ 1\leq i\leq k, 1\leq j\leq k+1\}$.
\item We note that the total number of terminals  is $|T^*|=4k(k+1)+2 = O(k^2)$
\item The edge set of $G^*$ is a disjoint union of
    \begin{itemize}
      \item Edges of main gadgets
      \item Edges of horizontal connector gadgets
      \item Edges of vertical connector gadgets

      \item Edges from $y^*$ to source-distinguished vertices of the vertical connector gadgets $VCG_{1,i}$ (for each $i\in [k]$), and from $y^*$ to source-distinguished vertices of horizontal connector gadgets $HCG_{j,1}$ (for each $j\in [k]$)

      \item Edges from sink-distinguished vertices of the vertical connector gadgets $VCG_{k+1,i}$ (for each $i\in [k]$) to $x^*$, and from sink-distinguished vertices of horizontal connector gadgets $HCG_{j,k+1}$ (for each $j\in [k]$) to $x^*$
      \item The edge $(x^*, y^*)$
    \end{itemize}
\end{itemize}

Define the following quantity: 
\begin{equation}
\label{eqn:w-star} W^*_n = k^{2}\cdot M^*_n + 2k(k+1)\cdot C^*_n.
\end{equation}
In the next two sections, we show that \gt has a solution if and only if the SCSS instance $(G^*,T^*)$ has a solution of weight at
most $W^{*}_n$.


\subsection{\gt has a solution $\Rightarrow$ SCSS has a solution of weight $\leq W^*_n$}
\label{proof:scss-planar-easy}

\begin{lemma}
If the \gt instance $(k,n, \{S_{i,j}\ :\ i,j\in [k]\})$ has a solution, then the SCSS instance $(G^*, T^*)$ has a solution of weight at most $W^*_n$.
\label{lem:scss-planar-easy}
\end{lemma}
\begin{proof}
Since \gt has a solution, for each $1\leq i,j\leq k$ there is an entry $(x_{i,j}, y_{i,j})=\gamma_{i,j}\in S_{i,j}$ such that
\begin{itemize}
\item For every $i\in [k]$, we have $x_{i,1}=x_{i,2}=x_{i,3}=\ldots=x_{i,k} = \alpha_i$
\item For every $j\in [k]$, we have $y_{1,j}=y_{2,j}=y_{3,j}=\ldots=y_{k,j} =  \beta_j$
\end{itemize}
We build a solution $E^*$ for the SCSS instance $(G^*, T^*)$ and show that it has weight at most $W^*_n$. In the edge set $E^*$, we
take the following edges:
\begin{enumerate}
\item The edge $(x^*, y^*)$ which has weight 0.

\item For each $j\in [k]$ the edge of weight 0 from $y^*$ to the source-distinguished vertex of $VCG_{1,j}$ of index $\beta_j$, and the edge of weight 0 from the sink-distinguished vertex of $VCG_{k+1,j}$ of index $\beta_j$ to $x^*$.

\item For each $i\in [k]$ the edge of weight 0 from $y^*$ to the source-distinguished vertex of $HCG_{i,1}$ of index $\alpha_i$, and the edge of weight 0 from the sink-distinguished vertex of $HCG_{i,k+1}$ of index $\alpha_i$ to $x^*$

\item For each $1\leq i,j\leq k$ for the main gadget $MG_{i,j}$, use Lemma~\ref{lem:main-gadget}(1) to generate a solution
    $E^{M}_{i,j}$ which has weight $M^*_n$ and represents $(\alpha_i, \beta_j)$.

\item For each $1\leq i\leq k$ and $1\leq j\leq k+1$ for the horizontal connector gadget $HCG_{i,j}$, use
    Lemma~\ref{lem:connector-gadget}(1) to generate a solution $E^{HC}_{i,j}$ of weight $C^*_n$ which represents $\alpha_i$.

\item For each $1\leq j\leq k$ and $1\leq i\leq k+1$ for the vertical connector gadget $VCG_{i,j}$, use
    Lemma~\ref{lem:connector-gadget}(1) to generate a solution $E^{VC}_{i,j}$ of weight $C^*_n$ which represents $\beta_j$.
\end{enumerate}

The weight of $E^*$ is $k^2\cdot M^*_n + k(k+1)\cdot C^*_n+ k(k+1)\cdot
C^*_n = W^*_n$. It remains to show that $E^*$ is a solution for the SCSS
instance $(G^*, T^*)$. Since we have already picked up the edge $(x^*,
y^*)$, it is enough to show that for any terminal $t\in T^{*}\setminus
\{x^*, y^*\}$, both $t\leadsto x^*$ and $y^*\leadsto t$ paths exist in
$E^*$. Then for any two terminals $t_1, t_2$, there is a $t_1 \leadsto
t_2$ path given by $t_1\leadsto x^*\rightarrow y^*\leadsto t_2$.

We now show the existence of both a $t\leadsto x^*$ path and a $y^*\leadsto t$ path in $E^*$ when $t$ is a terminal in a vertical connector gadget. Without loss of generality, let $t$ be the terminal $p^{v}_{i,j}$ for some $1\leq i\leq k, 1\leq j\leq k+1$.
\begin{itemize}
\item \underline{Existence of $p^{v}_{i,j}\leadsto x^*$ path in $E^*$}: By Lemma~\ref{lem:connector-gadget}(1), the terminal $p^{v}_{i,j}$ can
    reach the sink-distinguished vertex of $VCG_{i,j}$ which has the index $\beta_{j}$. This vertex is the
    top-distinguished vertex of the index $\beta_j$ of the main gadget $MG_{i,j}$. By Definition~\ref{defn:represents-main}, there
    is a path from this vertex to the bottom-distinguished vertex of the index $\beta_j$ of the main gadget $MG_{i,j}$.
    However this vertex is exactly the same as the source-distinguished vertex of the index $\beta_j$ of $VCG_{i+1,j}$. By
    Lemma~\ref{lem:connector-gadget}(1), the source-distinguished vertex of the index $\beta_j$ of $VCG_{i+1,j}$ can reach
    the sink-distinguished vertex of the index $\beta_j$ of $VCG_{i+1,j}$. This vertex is exactly the top-distinguished
    vertex of $MG_{i+1,j}$. Continuing in this way we can reach the source-distinguished vertex of the index $\beta_j$ of
    $VCG_{k+1,j}$. By Lemma~\ref{lem:connector-gadget}(1), this vertex can reach the sink-distinguished vertex of the
    index $\beta_j$ of $VCG_{k+1,j}$. But $E^*$ also contains an edge (of weight 0) from this sink-distinguished vertex to $x^*$, and hence there is a
    $p^{v}_{i,j}\leadsto x^*$ path in $E^*$.

\item \underline{Existence of $y^*\leadsto p^{v}_{i,j}$ path in $E^*$}: Recall that $E^*$ contains an edge (of weight 0) from $y^*$ to the source-distinguished
    vertex of the index $\beta_j$ of $VCG_{1,j}$. If $i=1$, then by Lemma~\ref{lem:connector-gadget}(1), there is a path
    from this vertex to $p^{v}_{1,j}$. If $i\geq 2$, then by Lemma~\ref{lem:connector-gadget}(1), there is a path from
    source-distinguished vertex of the index $\beta_j$ of $VCG_{1,j}$ to the sink-distinguished vertex of the index
    $\beta_j$ of $VCG_{1,j}$. But this is the top-distinguished vertex of $MG_{1,j}$ of the index $\beta_j$. By
    Definition~\ref{defn:represents-main}, from this vertex we can reach the bottom-distinguished vertex of the index $\beta_{j}$
    of $MG_{1,j}$. However, this vertex is exactly the source-distinguished vertex of index $\beta_j$ of $VCG_{2,j}$. Continuing this way we can reach the source-distinguished vertex of the index $\beta_j$ of $VCG_{i,j}$.
    By Lemma~\ref{lem:connector-gadget}(1), from this vertex we can reach $p^{v}_{i,j}$. Hence there is a $y^*\leadsto p^{v}_{i,j}$
    path in $E^*$.
\end{itemize}
The arguments when $t$ is a terminal in a horizontal connector gadget are similar, and we omit the details here.
\end{proof}

\subsection{SCSS has a solution of weight $\leq W^*_n$ $\Rightarrow$ \gt has a solution }
\label{proof:scss-planar-hard}

First we show that the following preliminary claim:
\begin{claim}
Let $E'$  be any solution to the SCSS instance $(G^*, T^*)$. Then
\begin{itemize}
\item $E'$ restricted to each connector gadget satisfies the connectedness property (see Definition~\ref{defn:connectedness}).
\item $E'$ restricted to each main gadget satisfies the connectedness property (see Definition~\ref{defn:source-sink-connectivity}).
\end{itemize}
\label{claim:scss-planar-gadgets-are-connected}
\end{claim}
\begin{proof}
First we show that the edge set $E'$ restricted to each connector gadget satisfies the connectedness property.
Consider a horizontal connector gadget $HCG_{i,j}$ for some $1\leq j\leq k+1, 1\leq i\leq k$. This gadget contains two terminals: $p_{i,j}^{h}$ and $q_{i,j}^{h}$. The only incoming edges from $G^*\setminus HCG_{i,j}$ into $HCG_{i,j}$ are incident onto the source-distinguished vertices of $HCG_{i,j}$, and the only outgoing edges from $HCG_{i,j}$ into $G^*\setminus HCG_{i,j}$ are incident on the sink-distinguished vertices of $HCG_{i,j}$. Since $E'$ is a solution of the SCSS instance $(G^*, T^*)$ it follows that $E^*$ contains a path from $p_{i,j}^{h}$ to the terminals in $T^* \setminus \{p_{i,j}^{h}\cup q_{i,j}^{h}\}$. Since the only outgoing edges from $HCG_{i,j}$ into $G^*\setminus HCG_{i,j}$ are incident on the sink-distinguished vertices of $HCG_{i,j}$, it follows that $p_{i,j}^{h}$ can reach some sink-distinguished vertex of $HCG_{i,j}$ in the solution $E'$. The other three conditions of Definition~\ref{defn:connectedness} can be verified similarly, and hence $E'$ restricted to each main connector satisfies the
connectedness property.

Next we argue that $E'$ restricted to each main gadget satisfies the connectedness property. Consider a main gadget
$MG_{i,j}$. Since $E'$ is a solution for the SCSS instance $(G^*, T^*)$ it follows that the terminal $p_{i,j}^{h}$ from $HCG_{i,j}$ is able to reach other terminals of $T^*$. However, the only outgoing edges from $HCG_{i,j}$ into $G^*\setminus HCG_{i,j}$ are incident on the sink-distinguished vertices of $HCG_{i,j}$. Moreover, each sink-distinguished vertex of $HCG_{i,j}$ is identified with a left-distinguished vertex of $MG_{i,j}$ of the same index. Hence, these outward paths from $p_{i,j}^{h}$ to other terminals of $T^*$ must continue through the left-distinguished vertices of $MG_{i,j}$. However, the only outgoing edges from $MG_{i,j}$ into $G^*\setminus MG_{i,j}$ are incident on the right-distinguished vertices or bottom-distinguished vertices of $HCG_{i,j}$. Hence, some left-distinguished vertex of $MG_{i,j}$ can reach some vertex in the set given by the union of right-distinguished and bottom-distinguished vertices of $MG_{i,j}$.
Hence the first condition of Definition~\ref{defn:source-sink-connectivity} is satisfied. Similarly it can be shown the other three conditions
of Definition~\ref{defn:source-sink-connectivity} also hold, and hence $E'$ restricted to each main gadget satisfies the
connectedness property.
\end{proof}

Now we are ready to prove the following lemma:

\begin{lemma}
If the SCSS instance $(G^*, T^*)$ has a solution $E''$ of weight at most $W^*_n$, then the \gt instance $(k,n, \{S_{i,j}\ :\ i,j\in [k]\})$ has a solution.
\label{lem:scss-planar-hard}
\end{lemma}
\begin{proof}
By Claim~\ref{claim:scss-planar-gadgets-are-connected}, the edge set $E''$ restricted to any connector gadget satisfies the
connectedness property and the edge set $E''$ restricted to any main gadget satisfies the connectedness property. Let $\mathcal{C}$ and $\mathcal{M}$ be the sets of connector and main gadgets respectively. Recall that $|\mathcal{C}|=2k(k+1)$ and $|\mathcal{M}|=k^2$. Recall that we have defined $W^*_n$ as $k^{2}\cdot M^*_n + 2k(k+1)C^*_n$. Let $\mathcal{C}'\subseteq \mathcal{C}$ be the set of connector gadgets that have weight at most $C^*_n$ in $E''$. By Lemma~\ref{lem:connector-gadget}(2), each connector gadget from the set $\mathcal{C}'$ has weight exactly $C^*_n$. Since all edge-weights in connector gadgets are positive integers, each connector gadget from the set $\mathcal{C}\setminus \mathcal{C}'$ has weight at least $C^{*}_n+1$. Similarly, let $\mathcal{M}'\subseteq \mathcal{M}$ be the set of main gadgets which have weight at most $M^*_n$ in~$E''$. By Lemma~\ref{lem:main-gadget}(2), each main gadget from the set $\mathcal{M}'$ has weight exactly $M^*_n$. Since all edge-weights in main gadgets are positive integers, each main gadget from the set $\mathcal{M}\setminus \mathcal{M}'$ has weight at least $M^{*}_n+1$. As any two gadgets are pairwise edge-disjoint, we have
\begin{align*}
W^*_n &= k^{2}\cdot M^*_n + 2k(k+1)C^*_n \\
&\geq |\mathcal{M}\setminus \mathcal{M}'|\cdot (M^{*}_n+1) + |\mathcal{M}'|\cdot M^{*}_n + |\mathcal{C}\setminus \mathcal{C}'|\cdot (C^{*}_n+1) + |\mathcal{C}'|\cdot C^{*}_n  \\
&= |\mathcal{M}|\cdot M^*_n + |\mathcal{C}|\cdot C^*_n + |\mathcal{M}\setminus \mathcal{M}'| + |\mathcal{C}\setminus \mathcal{C}'| \\
&= k^{2}\cdot M^*_n + 2k(k+1)\cdot C^*_n + |\mathcal{M}\setminus \mathcal{M}'| + |\mathcal{C}\setminus \mathcal{C}'|\\
&= W^*_n + |\mathcal{M}\setminus \mathcal{M}'| + |\mathcal{C}\setminus \mathcal{C}'|.
\end{align*}
This implies $|\mathcal{M}\setminus \mathcal{M}'| = 0 = |\mathcal{C}\setminus \mathcal{C}'|$. However, we had $\mathcal{M'}\subseteq \mathcal{M}$ and $\mathcal{C'}\subseteq \mathcal{C}$. Therefore, $\mathcal{M'}= \mathcal{M}$ and $\mathcal{C'}= \mathcal{C}$. Hence in $E''$, each connector gadget has weight $C^*_n$ and each main gadget has weight $M^*_n$.
%
%
From
Lemma~\ref{lem:connector-gadget}(2) and Lemma~\ref{lem:main-gadget}(2), we have
\begin{itemize}
\item For each vertical connector gadget $VCG_{i,j}$, the restriction of the edge set $E''$ to $VCG_{i,j}$ represents an
    integer $\beta_{i,j}\in [n]$ where $i\in [k+1], j\in [k]$.

\item For each horizontal connector gadget $HCG_{i,j}$, the restriction of the edge set $E''$ to $HCG_{i,j}$ represents an
    integer $\alpha_{i,j}$ where $i\in [k], j\in [k+1]$.

\item For each main gadget $MG_{i,j}$, the restriction of the edge set $E''$ to $MG_{i,j}$ represents an ordered pair $(\alpha'_{i,j}, \beta'_{i,j})\in S_{i,j}$ where $i, j\in [k]$.

\end{itemize}
Consider the main gadget $MG_{i,j}$ for any $1\leq i,j\leq k$. We can make the following observations:
\begin{itemize}
\item \underline{$\beta_{i,j}=\beta'_{i,j}$}: By Lemma~\ref{lem:connector-gadget}(2) and
    Definition~\ref{defn:represents-connector}, the terminal vertices in $VCG_{i,j}$ can exit the vertical connector
    gadget only via the unique edge entering the sink-distinguished vertex of index $\beta_{i,j}$. By
    Lemma~\ref{lem:main-gadget}(2) and Definition~\ref{defn:represents-main}, the only edge in $E''$ incident to any
    top-distinguished vertex of $MG_{i,j}$ is the unique edge leaving the top-distinguished vertex of the index
    $\beta'_{i,j}$. Hence if $\beta_{i,j}\neq \beta'_{i,j}$ then the terminals in $VCG_{i,j}$ will not be able to exit $VCG_{i,j}$ and reach
    other terminals.

\item \underline{$\beta'_{i,j}=\beta_{i+1,j}$}: By Lemma~\ref{lem:connector-gadget}(2) and
    Definition~\ref{defn:represents-connector}, the unique edge entering $VCG_{i+1,j}$ is the edge entering the
    source-distinguished vertex of the index $\beta_{i+1,j}$. By Lemma~\ref{lem:main-gadget}(2) and
    Definition~\ref{defn:represents-main}, the only edge in $E''$ incident to any bottom-distinguished vertex of
    $MG_{i,j}$ is the unique edge entering the bottom-distinguished vertex of index $\beta'_{i,j}$. Hence if
    $\beta'_{i,j}\neq \beta_{i+1,j}$, then the terminals in $VCG_{i+1,j}$ cannot be reached from the other terminals.

\item \underline{$\alpha_{i,j}=\alpha'_{i,j}$}: By Lemma~\ref{lem:connector-gadget}(2) and
    Definition~\ref{defn:represents-connector}, the paths starting at the terminal vertices in $HCG_{i,j}$ can leave the horizontal connector
    gadget only via the unique edge entering the sink-distinguished vertex of index $\alpha_{i,j}$. By
    Lemma~\ref{lem:main-gadget}(2) and Definition~\ref{defn:represents-main}, the only edge in $E''$ incident to any
    left-distinguished vertex of $MG_{i,j}$ is the unique edge leaving the left-distinguished vertex of the index
    $\alpha'_{i,j}$. Hence if $\alpha_{i,j}\neq \alpha'_{i,j}$ then the terminals in $HCG_{i,j}$ will not be able to reach
    other terminals.

\item \underline{$\alpha'_{i,j}=\alpha_{i,j+1}$}: By Lemma~\ref{lem:connector-gadget}(2) and
    Definition~\ref{defn:represents-connector}, the unique edge entering $HCG_{i,j+1}$ is the edge entering the
    source-distinguished vertex of index $\alpha_{i,j+1}$. By Lemma~\ref{lem:main-gadget}(2) and
    Definition~\ref{defn:represents-main}, the only edge in $E''$ incident to any right-distinguished vertex of $MG_{i,j}$
    is the unique edge entering the right-distinguished vertex of index $\alpha'_{i,j}$. Hence if
    $\alpha'_{i,j}\neq \alpha_{i,j+1}$, then the terminals in $HCG_{i,j+1}$ cannot be reached from the other terminals.
\end{itemize}
We claim that for $1\leq i,j\leq k$, the entries $(\alpha'_{i,j}, \beta'_{i,j})\in S_{i,j}$ form a solution for the \gt instance.
For this we need to check two conditions:
\begin{itemize}
\item \underline{$\alpha'_{i,j} = \alpha'_{i,j+1}$}: This holds because $\alpha_{i,j} = \alpha'_{i,j} = \alpha_{i,j+1} =
    \alpha'_{i,j+1}$.
\item \underline{$\beta'_{i,j} = \beta'_{i+1,j}$}: This holds because $\beta_{i,j} = \beta'_{i,j} = \beta_{i+1,j} =
    \beta'_{i+1,j}$.
\end{itemize}
This completes the proof of the lemma.
\end{proof}

\subsection{Proof of Theorem~\ref{thm:scss-main-hardness-planar-graphs}}
\label{sec:main-scss-planar-proof-subsec}

Finally we are ready to prove Theorem~\ref{thm:scss-main-hardness-planar-graphs} which is restated below:

\begin{reptheorem}{thm:scss-main-hardness-planar-graphs}
The edge-unweighted version of the SCSS problem is W[1]-hard parameterized by the number of terminals $k$, even when the underlying undirected graph is planar. Moreover, under the ETH, the SCSS problem on planar graphs cannot be solved in $f(k)\cdot n^{o(\sqrt{k})}$ time where $f$ is any computable function, $k$ is the number of terminals
  and $n$ is the number of vertices in the instance.
\end{reptheorem}
%
%
\begin{proof}
Each connector gadget has $O(n^2)$ vertices and $G^*$ has $O(k^2)$ connector gadgets. Each main gadget has $O(n^3)$ vertices and $G^*$ has $O(k^2)$ main gadgets. It is easy to see that the graph $G^*$ has $O(n^{3}k^{2})=\poly(n,k)$ vertices. Moreover, the graph $G^*$ can be constructed in $\poly(n+k)$ time: recall that each connector gadget (Lemma~\ref{lem:connector-gadget}) and main gadget (Lemma~\ref{lem:main-gadget}) can be constructed in polynomial time. Each main gadget and connector gadget is planar, and any two gadgets are pairwise edge-disjoint. Moreover, the 0-weight edges incident on $x^*$ or $y^*$ do not affect planarity (see Figure~\ref{fig:big-picture} for a planar embedding). Hence, $G^*$ is planar.

It is known~\cite[Theorem 14.28]{fpt-book} that $k\times k$ \gt is W[1]-hard parameterized by $k$, and under ETH cannot be solved in $f(k)\cdot n^{o(k)}$ for any computable function $f$. Combining the two directions from Section~\ref{proof:scss-planar-easy} and Section~\ref{proof:scss-planar-hard}, we get a parameterized reduction from $k\times k$ \gt to a planar instance of SCSS with $O(k^2)$ terminals.  Hence, it follows that SCSS on planar graphs is W[1]-hard and under ETH cannot be solved in $f(k)\cdot n^{o(\sqrt{k})}$ time for any computable function $f$.
\end{proof}

This shows that the $2^{O(k)}\cdot n^{O(\sqrt{k})}$ algorithm for SCSS on planar graphs given in
Theorem~\ref{thm:algo-sqrt-k-h-minor-free} is asymptotically optimal.

\section{Proof of Lemma~\ref{lem:connector-gadget}: constructing connector gadgets} \label{connector:proof}


We prove Lemma~\ref{lem:connector-gadget} in this section, by constructing a connector gadget satisfying the specifications of Section~\ref{sec:connector:exist}.

\begin{figure}[H]
\centering
\def\svgwidth{\linewidth}%
\executeiffilenewer{connector-new.svg}{connector-new.pdf}%
{inkscape -z -D --file=connector-new.svg %
--export-pdf=connector-new.pdf --export-latex}%
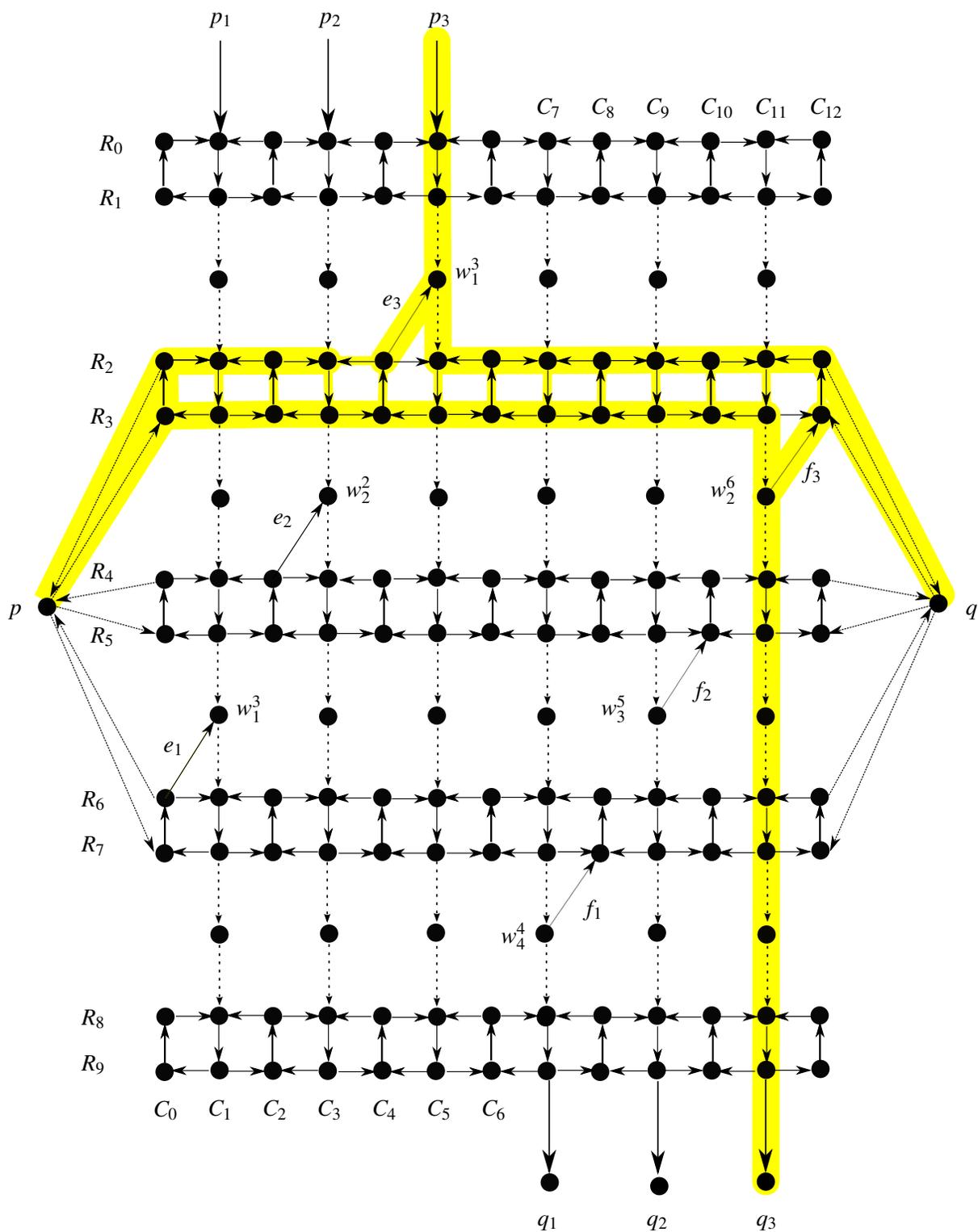%

\vspace{-35mm}
\caption{The connector gadget \label{fig:connector-gadget} for $n=3$. A set of edges representing 3 is highlighted in the figure.}
\end{figure}

\subsection{Different types of edges in connector gadget}

Before proving Lemma~\ref{lem:connector-gadget}, we first describe the construction of the connector gadget in more detail (see
Figure~\ref{fig:connector-gadget}). The connector gadget has $2n+4$ rows denoted by $R_0, R_1, R_2, \ldots, R_{2n+3}$ and $4n+1$ columns
denoted by $C_0, C_1, \ldots, C_{4n}$. Let us denote the vertex at the intersection of row $R_i$ and column $C_j$ by
$v_{i}^j$. We now describe the different kinds of edges present in the connector gadget.
\begin{enumerate}
\item \textbf{Source Edges}: For each $i\in [n]$, there is an edge $(p_i, v_{0}^{2i-1})$. These edges are together called source edges.

\item \textbf{Sink Edges}: For each $i\in [n]$, there is an edge $(v_{2n+3}^{2n+2i-1}, q_i)$. These edges are together called sink edges.

\item \textbf{Terminal Edges}: The union of the sets of edges incident to the terminals $p$ or $q$ are called terminal edges. The set of edges incident on $p$ is $\{(p, v^{0}_{2i+1}\ :\ i\in [n])\} \cup \{(v^{0}_{2i}, p\ :\ i\in [n])\}$. The set of edges incident on $q$ is $\{(q, v^{4n}_{2i+1}\ :\ i\in [n])\} \cup \{(v^{4n}_{2i}, q\ :\ i\in [n])\}$.

\item \textbf{Inrow Edges}:
\begin{itemize}
\item \underline{Inrow Up Edges}: For each $0\leq i\leq n+1$, we call the $\uparrow$ edges connecting vertices of row $R_{2i+1}$ to $R_{2i}$ as
inrow up edges. Explicitly, this set of edges is given by $\{(v_{2i+1}^{2j}, v_{2i}^{2j})\ :\ 0\leq j\leq 2n \}$.

\item \underline{Inrow Down Edges}: For each $0\leq i\leq n+1$, we call the $\downarrow$ edges connecting vertices of row $R_{2i}$ to
$R_{2i+1}$ as inrow down edges. Explicitly, this set of edges is given by $\{(v_{2i}^{2j-1}, v_{2i+1}^{2j-1})\ :\ 1\leq j\leq 2n \}$.

\item \underline{Inrow Left Edges}: For each $0\leq i\leq 2n+3$, we call the $\leftarrow$ edges connecting vertices of row $R_{i}$ as inrow
left edges. We explicitly list the set of inrow left edges for even-numbered and odd-numbered rows below:
        \begin{itemize}
          \item For each $0\leq i\leq n+1$, the set of inrow left edges for the row $R_{2i}$ is given by $\{(v_{2i}^{2j}, v_{2i}^{2j-1})\ :\ j\in [2n] \}$
          \item For each $0\leq i\leq n+1$, the set of inrow left edges for the row $R_{2i+1}$ is given by $\{(v_{2i+1}^{2j-1}, v_{2i+1}^{2j-2})\ :\ j\in [2n] \}$
        \end{itemize}

\item \underline{Inrow Right Edges}: For each $0\leq i\leq 2n+3$, we call the $\rightarrow$ edges connecting vertices of row $R_{i}$ as inrow
right edges. We explicitly list the set of inrow right edges for even-numbered and odd-numbered rows below:
        \begin{itemize}
          \item For each $0\leq i\leq n+1$, the set of inrow right edges for the row $R_{2i}$ is given by $\{(v_{2i}^{2j-2}, v_{2i}^{2j-1})\ :\ j\in [2n] \}$
          \item For each $0\leq i\leq n+1$, the set of inrow right edges for the row $R_{2i+1}$ is given by $\{(v_{2i+1}^{2j-1}, v_{2i+1}^{2j})\ :\ j\in [2n] \}$
        \end{itemize}


\end{itemize}

\item \textbf{Interrow Edges}: For each $i\in [n+1]$ and each $j\in [2n]$, we subdivide the edge $(v_{2i-1}^{2j-1}, v_{2i}^{2j-1})$ by introducing a new vertex $w_{i}^{j}$ and adding the edges $(v_{2i-1}^{2j-1}, w_{i}^j)$ and $(w^{j}_i, v_{2i}^{2j-1})$. All these edges are together called interrow edges. Note that there is a total of $4n(n+1)$ interrow edges.

\item \textbf{Shortcuts}: There are $2n$ shortcut edges, namely $e_1, e_2, \ldots, e_n$ and $f_1, f_2, \ldots,f_n$. They are drawn
as follows:
\begin{itemize}
\item The edge $e_i$ is given by $(v_{2n-2i+2}^{2i-2}, w_{n-i+1}^i)$.
\item The edge $f_i$ is given by $(w_{n-i+2}^{n+i }, v_{2n-2i+3}^{2n+2i})$.
\end{itemize}

\end{enumerate}

\subsection{Assigning weights in the connector gadget}
Fix the quantity $B=18n^2$. We assign weights to the edges as follows
\begin{enumerate}
\item For $i\in [n]$, the source edge $(p_i, v_{0}^{2i-1})$ has weight $B^{5}+(n-i+1)$.
\item For $i\in [n]$, the sink edge $(v_{2n+3}^{2n+2i-1}, q_i)$ has weight $B^{5}+i$.

\item Each terminal edge has weight $B^4$.
\item Each inrow up edge has weight $B^3$.

\item Each interrow edge has weight $\dfrac{B^2}{2}$ each.

\item Each inrow right edge has weight $B$.

\item For each $i\in [n]$, the shortcut edge $e_i$ has weight $n\cdot i$.
\item For each $j\in [n]$, the shortcut edge $f_j$ has weight $n(n-j+1)$.


\item Each inrow left edge and inrow down edge has weight 0.

\end{enumerate}

Now we define the quantity $C^*_n$ stated in statement of Lemma~\ref{lem:connector-gadget}:
\begin{equation}
C^*_n = 2B^{5} + 4B^{4} + (2n+1)B^{3} + (n+1)B^{2} + (4n-2)B + (n+1)^2.
\label{eqn:c-star}
\end{equation}

In the next two sections, we prove the two statements of Lemma~\ref{lem:connector-gadget}.

\subsection{For every $i\in [n]$, there is a solution $E_i$ of weight $C^*_n$ that satisfies the connectedness property and represents $i$}
\label{subsection:connector-first-statement}

Let $E_i$ be the union of the following sets of edges:
\begin{itemize}
\item
    Select the edges
    $(p_i, v_{0}^{2i-1})$ and $(v_{2n+3}^{2n+2i-1}, q_i)$. This incurs a weight of $B^{5}+ (n-i+1) +  B^{5}+
    i = 2B^{5}+(n+1)$.

\item The two terminal edges $(p, v_{2n-2i+3}^{0})$ and $(v_{2n-2i+2}^{0}, p)$. This incurs a weight of $2B^4$.
\item The two terminal edges $(q, v_{2n-2i+3}^{4n})$ and $(v_{2n-2i+2}^{4n}, q )$. This incurs a weight of $2B^4$.

\item All $2n$ inrow right edges and $2n$ inrow left edges which occur between vertices of $R_{2n-2i+2}$. This incurs a
    weight of $2n\cdot B$ since each inrow left edge has weight 0 and each inrow right edge has weight~$B$.
\item All $2n$ inrow right edges and $2n$ inrow left edges which occur between vertices of $R_{2n-2i+3}$. This incurs a weight of
    $2n \cdot B$ since each inrow left edge has weight 0 and each inrow right edge has weight~$B$.

\item All the $2n+1$ inrow up edges that are between vertices of $R_{2n-2i+2}$ and $R_{2n-2i+3}$. These edges are given
    by $(v_{2n-2i+3}^{2j}, v_{2n-2i+2}^{2j})$ for $0\leq j\leq 2n$. This incurs a weight of $(2n+1)B^3$.
\item All $2n$ inrow down edges that occur between vertices of row $R_{2n-2i+2}$ and row $R_{2n-2i+3}$. This incurs a weight of 0, since each inrow down edge has weight 0.

\item The vertically downward $v_{0}^{2i-1}\leadsto v_{2n-2i+3}^{2i-1}$ path $P_1$ formed by interrow edges and inrow down edges, and the vertically downward $v_{2n-2i+2}^{2n+2i-1}\leadsto v_{2n+3}^{2n+2i-1}$ path $P_2$ formed by interrow edges and inrow down edges.
    These two paths together incur a total weight of $(n+1)B^2$, since the inrow down edges have weight 0.

\item The edges $e_i$ and $f_i$. This incurs a weight of $n\cdot i+ n(n-i+1) = n(n+1)$.


\end{itemize}

Finally, \textbf{remove} the two inrow right edges $(v_{2n-2i+2}^{2i-2},v_{2n-2i+2}^{2i-1})$ and $(v_{2n-2i+3}^{2n+2i-1}, v_{2n-2i+3}^{2n+2i})$ from
$E_i$. This saves a weight of $2B$. From the above paragraph and Equation~\ref{eqn:c-star} it follows that the total weight of $E_i$
is exactly $C^*_n$. Note that even though we removed the edge $(v_{2n-2i+2}^{2i-2},v_{2n-2i+2}^{2i-1})$ we can still travel from $v_{2n-2i+2}^{2i-2}$ to $v_{2n-2i+2}^{2i-1}$ in $E_i$ using the edge $e_i$ as follows: take the path $v_{2n-2i+2}^{2i-2}\rightarrow w_{n-i+1}^i \rightarrow v_{2n-2i+2}^{2i-1}$. Similarly, even though we removed the edge
$(v_{2n-2i+3}^{2n+2i-1}, v_{2n-2i+3}^{2n+2i})$ we can still travel from $v_{2n-2i+3}^{2n+2i-1}$ to $v_{2n-2i+3}^{2n+2i}$ in
$E_i$ using the edge $f_i$ as follows: take the path $v_{2n-2i+3}^{2n+2i-1} \rightarrow w_{n-i+2}^{n+i} \rightarrow v_{2n-2i+3}^{2n+2i}$.

It remains to show that $E_i$ satisfies the connectedness property and it represents $i$. It is easy to see $E_i$ represents $i$
since the only edge in $E_i$ which is incident to $P$ is the edge leaving $p_i$. Similarly, the only edge in $E_i$
incident to $Q$ is the one entering $q_i$. We show that the connectedness property holds as follows (recall Definition~\ref{defn:connectedness}):
\begin{enumerate}
\item There is a $p_i \leadsto p$ path in $E_i$ by starting with the source edge leaving $p_i$ and then following downward path $P_1$
    from $v_{0}^{2i-1} \leadsto v_{2n-2i+3}^{2i-1}$. Then travel towards the left from $v_{2n-2i+3}^{2i-1}$ to $p$ by
    using inrow left, inrow up and inrow down edges from rows $R_{2n-2i+2}$ and $R_{2n-2i+3}$. Finally, use the edge
    $(v_{2n-2i+2}^{0}, p)$

\item For the existence of a $p_i \leadsto q$ path in $E_i$, we have seen above that there is a $p_i \leadsto
    v_{2n-2i+3}^{2i-1}$ path. Then travel towards the right from $v_{2n-2i+2}^{2i-1}$ to $q$ by using inrow right, inrow
    up and inrow down edges from rows $R_{2n-2i+2}$ and $R_{2n-2i+3}$ to reach the vertex $v_{2n-2i+2}^{4n}$. The only potential issue is that the inrow right edge $(v_{2n-2i+3}^{2n+2i-1}, v_{2n-2i+3}^{2n+2i})$ is missing in $E_i$: however this is not a problem since we have the path $v_{2n-2i+3}^{2n+2i-1}\rightarrow w_{n-i+2}^{n+i} \rightarrow v_{2n-2i+3}^{2n+2i}$ in $E_i$. Finally, use the edge $(v_{2n-2i+2}^{4n}, q)$.

\item For the existence of a $p \leadsto q_i$ path in $E_i$, first use the edge $(p, v_{2n-2i+3}^{0})$. Then travel
    towards the right by using inrow up, inrow right and inrow down edges from rows $R_{2n-2i+2}$ and $R_{2n-2i+3}$ to reach the vertex $v_{2n-2i+2}^{2n+2i-1}$. The only potential issue is that the inrow right edge $(v_{2n-2i+2}^{2i-2},v_{2n-2i+2}^{2i-1})$ is missing in $E_i$: however this is not a problem since we have the path $v_{2n-2i+2}^{2i-2}\rightarrow w_{n-i+1}^{i} \rightarrow v_{2n-2i+2}^{2i-1}$ in $E_i$. Then take the downward path $P_2$ from $v_{2n-2i+2}^{2n+2i-1}$ to
    $v_{2n+3}^{2n+2i-1}$. Finally, use the sink edge $(v_{2n+3}^{2n+2i-1}, q_i)$ incident to $q_i$.

\item For the existence of a $q \leadsto q_i$ path in $E_i$, first use the terminal edge $(q, v_{2n-2i+3}^{4n})$. Then travel
    towards the left by using inrow up, inrow left and inrow down edges from rows $R_{2n-2i+2}$ and $R_{2n-2i+3}$ until you reach
    the vertex $v_{2n-2i+2}^{2n+2i-1}$. Then take the downward path $P_2$ from $v_{2n-2i+2}^{2n+2i-1}$ to
    $v_{2n+3}^{2n+2i-1}$. Finally, use the sink edge $(v_{2n+3}^{2n+2i-1}, q_i)$ incident to $q_i$.
\end{enumerate}

Therefore, $E_i$ indeed satisfies the connectedness property.

\subsection{$E'$ satisfies the connectedness property and has weight at most $C^*_n \Rightarrow E'$ represents some $\beta\in [n]$ and has weight exactly $C^*_n$}
\label{subsection:connector-second-statement}

Next we show that if a set of edges $E'$ satisfies the connectedness property and has weight at most $C^*_n$, then in fact the weight of $E'$ is exactly $C^*_n$ and it represents some
$\beta\in [n]$.
We do this via the following series of claims and observations.

\begin{claim}
\label{claim:one-in-out} $E'$ contains exactly one source edge and one sink edge.
\end{claim}
\begin{proof}
Since $E'$ satisfies the connectedness property it must contain at least one source edge and at least one sink edge. Without
loss of generality, suppose that there are at least two source edges in $E'$. Then the weight of $E'$ is a least the sum of the weights
of these two source edges plus the weight of at least one sink edge.
Thus if $E'$ contains at least two source edges, then its weight is at least $3B^{5}$. However, from Equation~\ref{eqn:c-star} we get that
\begin{align*}
C^*_n &= 2B^{5} + 4B^{4} + (2n+1)B^{3} + (n+1)B^{2} + (4n-2)B + (n+1)^2\\
&\leq 2B^5 + 4n\cdot B^4 + 3n\cdot B^4 + 2n\cdot B^4 + 4n\cdot B^4 +4n\cdot B^4\\
&\leq 2B^5 + 17n\cdot B^4\\
&< 3B^5,
\end{align*}
since $B=18n^2 > 17n$.
\end{proof}

Thus we know that $E'$ contains exactly one source edge and exactly
one sink edge. Let the source edge be incident to $p_{i'}$ and the
sink edge be incident to $q_{j'}$.

\begin{claim}
\label{claim:4-blue-dotted} $E'$ contains exactly four terminal edges.
\end{claim}
\begin{proof}
  Since $E'$ satisfies the connectedness property, it must contain at
  least one incoming and one outgoing edge for both $p$ and
  $q$. Therefore, we need at least four terminal edges. Suppose we
  have at least five terminal edges in $E'$. We already know that the
  source and sink edges contribute at least $2B^{5}$ to weight of
  $E'$, hence the weight of $E'$ is at least $2B^{5}+ 5B^{4}$. However, from
  Equation~\ref{eqn:c-star}, we get that
\begin{align*}
C^*_n &= 2B^{5} + 4B^{4} + (2n+1)B^{3} + (n+1)B^{2} + (4n-2)B + (n+1)^2\\
&\leq 2B^5 + 4B^4 + 3n\cdot B^3 + 2n\cdot B^3 + 4n\cdot B^3 + 4n\cdot B^3\\
&= 2B^5 + 4B^4 + 13n\cdot B^3\\
&< 2B^5 + 5B^4,
\end{align*}
since $B=18n^2 > 13n$.
\end{proof}

Hence we know that $E'$ contains exactly four terminal edges.

\begin{claim}
\label{claim:one-red-inrow} $E'$ contains exactly $2n+1$ inrow up edges, one from each column $C_{2i}$ for $0\leq i\leq 2n$.
\end{claim}
\begin{proof}
Observe that for each $1\leq j\leq 2n-1$, the inrow up edges in column $C_{2j}$ form a cut between vertices from columns $C_{2j-1}$
and $C_{2j+1}$. Since $E'$ must have a $p_{i'} \leadsto p$ path, we need to use at least one inrow up edge from each of the
columns $C_0, C_2, \ldots, C_{2i'-2}$. Since $E'$ must have a $p_{i'} \leadsto q$, path we need to use at least one inrow up
edge from each of the columns $C_{2i'}, C_{2i'+2}, \ldots, C_{4n}$. Hence $E'$ has at least $2n+1$ inrow up edges, as we
require at least one inrow up edge from each of the columns $C_{0}, C_{2}, \ldots, C_{4n}$.

Suppose $E'$ contains at least $2n+2$ inrow up edges. We already know that $E'$ has a contribution of $2B^{5}+4B^{4}$ from source,
sink, and terminal edges. Hence the weight of $E'$ is at least  $2B^{5}+ 4B^{4}+ (2n+2)B^3$. However, from Equation~\ref{eqn:c-star}, we get that
\begin{align*}
C^*_n &= 2B^{5} + 4B^{4} + (2n+1)B^{3} + (n+1)B^{2} + (4n-2)B + (n+1)^2\\
&\leq 2B^5 + 4B^4 + (2n+1)B^3 + 2n\cdot B^2 + 4n\cdot B^2 + 4n\cdot B^2\\
&= 2B^5 + 4B^4 + (2n+1)B^3 + 10n\cdot B^2\\
&< 2B^5 + 4B^4 + (2n+2)B^3,
\end{align*}
since $B=18n^2 > 10n$.
\end{proof}

Therefore, we know that $E'$ contains exactly one inrow edge per column $C_{2i}$ for every $0\leq i\leq 2n$. By
Claim~\ref{claim:4-blue-dotted}, we know that exactly two terminal edges incident to $p$ are selected in $E'$. Observe that the terminal edge leaving $p$ should be followed by an inrow up edge, and similarly, the terminal edge entering $p$ follows an inrow up edge. Since we select exactly
one inrow up edge from column $C_0$, it follows that the two terminal edges in $E'$ incident to $p$ must be incident to the rows $R_{2\ell+1}$
and $R_{2\ell}$ respectively for some $\ell \in [n]$. Similarly, the two terminal edges in $E'$ incident to $q$ must be
incident to the rows $R_{2\ell'+1}$ and $R_{2\ell'}$ for some $\ell' \in [n]$. We summarize this in the following claim:

\begin{observation}
\label{obs:same-for-s-and-t} There exist integers $\ell, \ell' \in [n]$ such that
\begin{itemize}
\item the only two terminal edges in $E'$ incident to $p$ are $(p, v_{2\ell+1}^{0})$ and $(v_{2\ell}^{0}, p)$, and
\item the only two terminal edges in $E'$ incident to $q$ are $(q, v_{2\ell'+1}^{4n})$ and $(v_{2\ell'}^{4n}, q)$.
\end{itemize}
\end{observation}

\begin{definition}
For $i\in [n+1]$, we call the $4n$ interrow edges which connect vertices from row $R_{2i-1}$ to vertices from $R_{2i}$ as
Type$(i)$ interrow edges. We can divide the Type$(i)$ interrow edges into $2n$ ``pairs" of adjacent interrow edges given by $(v_{2i-1}^{2j-1}, w_{i}^{j})$ and $(w_{i}^{j}, v_{2i}^{2j-1})$ for each $1\leq j\leq 2n$
\end{definition}

Note that there are a total of $n+1$ types of interrow edges. 
%
%
%
%
%

\begin{claim}
\label{claim:one-interrow} $E'$ contains a pair of interrow edges of Type$(r)$ for each $r\in [n+1]$. Moreover, these two edges are the only interrow edges of Type$(r)$ chosen in $E'$.
\end{claim}
\begin{proof}
First we show that $E'$ contains at least one pair of interrow edges of each type. Observation~\ref{obs:same-for-s-and-t} implies that we cannot avoid using interrow edges of any type by, for example, going into $p$ via an edge from some $R_{2i}$ and then exiting $p$ via an edge to some $R_{2j+1}$ for any $j> i$ (similarly for $q$). By the
connectedness property, set $E'$ contains a $p_{i'}\leadsto p$ path $P_1$. By Observation~\ref{obs:same-for-s-and-t}, the only edge entering $p$ is $(v_{2\ell}^{0}, p)$. Hence $E'$ must contain at least one pair of interrow edges of Type$(r)$ for $1\leq r\leq \ell$ since the only way to travel from row $R_{2r-1}$ to $R_{2r}$  (for each $r\in [\ell]$) is by using a pair of interrow edges of Type$(r)$ . Similarly $E'$ contains a $p\leadsto q_i$ path and the only outgoing edge from $p$ is $(p, v_{2\ell+1}^{0})$. Hence
$E'$ must contain at least one pair of interrow edges of Type$(r)$ for each $\ell+1 \leq r\leq n+1$ since the only way to travel from row $R_{2r-1}$ to $R_{2r}$ is by using a pair of interrow edges of Type$(r)$. Therefore, the edge set $E'$ contains at least one pair of interrow edges of each Type$(r)$ for $1\leq r\leq n+1$.

Next we show that $E'$ contains exactly two interrow edges of Type$(r)$ for each $r\in [n+1]$. Suppose $E'$ contains at least three interrow edges of some Type$(r)$ for some $r\in [n+1]$. Since weight of each interrow edge is $B^2/2$, this implies $E'$ gets a weight of at least $(n+1 +\frac{1}{2})\cdot B^2$ from the interrow edges. We have already seen $E'$ has contribution of $2B^{5}+4B^{4}+(2n+1)B^3$ from source, sink, terminal, and inrow up
edges. Hence the weight of $E'$ is at least $2B^{5}+ 4B^{4}+ (2n+1)B^3+ (n+1 + \frac{1}{2})\cdot B^2$. However, from Equation~\ref{eqn:c-star}, we get that
\begin{align*}
C^*_n &= 2B^{5} + 4B^{4} + (2n+1)B^{3} + (n+1)B^{2} + (4n-2)B + (n+1)^2\\
&\leq 2B^5 + 4B^4 + (2n+1)B^3 + (n+1)B^2 + 4n\cdot B + 4n\cdot B\\
&= 2B^5 + 4B^4 + (2n+1)B^3 + (n+1)B^2 + 8n\cdot B\\
&< 2B^5 + 4B^4 + (2n+1)B^3 + (n+1 + \frac{1}{2})B^2,
\end{align*}
since $\frac{B}{2}=9n^2 > 8n$. Hence, $E'$ contains exactly two interrow edges of Type$(r)$ for each $r\in [n+1]$.
%
\end{proof}

\begin{claim}
For each $r\in [n+1]$, let the unique pair of interrow edges in $E'$ (guaranteed by Claim~\ref{claim:one-interrow}) of Type$(r)$ belong
to column $C_{2\ell_{r}-1}$. If the unique source and sink edges in $E'$ (guaranteed by Claim~\ref{claim:one-in-out}) are incident to $p_{i'}$ and
$q_{j'}$, respectively, then we have $i'\leq \ell_1 \leq \ell_2 \leq \ldots \leq \ell_{n+1} \leq n+j'.$
\label{claim:interrow-monotone}
\end{claim}
\begin{proof}
Observation~\ref{obs:same-for-s-and-t} implies the only way to get from row $R_{2i-1}$ to $R_{2i}$ is to use a pair of interrow edges
of Type$(i)$. By Claim~\ref{claim:one-interrow}, we use exactly one pair of interrow edges of each type. Recall that there is a
walk $P=p_{i'}\leadsto p\leadsto q_{j'}$ in $E'$, and this walk must use each of these interrow edges.

First we show that $\ell_1 \geq i'$. Suppose $\ell_1 <i'\leq n$. Since we use the source edge incident to $p_{i'}$, we must reach
vertex $v_{0}^{2i'-1}$. Since $i'>\ell_1$, to use the pair of interrow edges to travel from $v_{1}^{2\ell_{1}-1}$ to $v_{2}^{2\ell_{1}-1}$, the walk $P$ must contain a
$v_{0}^{2i'-1}\leadsto v_{1}^{2\ell_{1}-1}$ subwalk $P'$. By the construction of the connector gadget this subwalk $P'$ must
use the inrow up edge $(v_{1}^{2i'-2}, v_{0}^{2i'-2})$. Now the walk $P$ has to reach column $C_{2n+2j'-1}$ from column
$C_{2\ell_{1}-1}$, and so it must use another inrow edge from column $C_{2i'-2}$ (between rows $R_{2i}$ and $R_{2i+1}$ for some $i\geq 1$), which contradicts
Claim~\ref{claim:one-red-inrow}.

Now we prove $\ell_{n+1}\leq n+j'$. Suppose to the contrary that $\ell_{n+1}>n+j'$. Then by reasoning similar to that of above one
can show that at least two inrow up edges from column $C_{2n+2j'}$ are used, which contradicts Claim~\ref{claim:one-red-inrow}.

Finally suppose there exists $r\in [n]$ such that $\ell_{r} > \ell_{r+1}$. We consider the following three cases:
\begin{itemize}
\item \underline{$\ell_{r+1} < \ell_{r}\leq n$}: Then by using the fact that there is a $p_{i'}\leadsto q_{j'}$ walk in $E'$ we
    get at least two inrow up edges are used from column $C_{2\ell_{r}-2}$, which contradicts
    Claim~\ref{claim:one-red-inrow}.
\item \underline{$n< \ell_{r} \leq n+j'$}: Then we need to use at least two inrow up edges from column $C_{2\ell_{r}-2}$, which
    contradicts Claim~\ref{claim:one-red-inrow}.
\item \underline{$\ell_{r} > n+j'$}: Then we need to use at least two inrow up edges from column $C_{2n+2j'}$, which
    contradicts Claim~\ref{claim:one-red-inrow}.
\end{itemize}
\end{proof}

\begin{claim}
$E'$ contains at most two shortcut edges. \label{claim:at-most-2-shortcuts}
\end{claim}
\begin{proof}
For the proof we will use Claim~\ref{claim:interrow-monotone}. We will show that we can use at most one $e$-shortcut. The proof for $f$-shortcut is similar.

Suppose we use two $e$-shortcuts viz. $e_{x}$ and $e_{y}$ such that $x> y$. Note that it makes sense to include a shortcut into $E'$ only if
we use the interrow edge that continues it. Hence $\ell_x = x$ and $\ell_y = y$. By Claim~\ref{claim:interrow-monotone}, we have $y=\ell_y\geq \ell_x = x$, which is a contradiction.
\end{proof}

\begin{claim}
\label{claim:exactly-4n-2-inrow-right} $E'$ contains exactly $4n-2$ inrow right edges.
\end{claim}
\begin{proof}
%
%
Since $E'$ contains a $p\leadsto q_{j'}$ path, it follows that $E'$ has a path connecting some vertex from the column $C_{i}$ to some vertex from column $C_{i+1}$ for each $0\leq
i\leq 2n+2j'-2$. Since $E'$ contains a $p_{i'}\leadsto q$ path, it follows that $E'$ has a path connecting some vertex from the column $C_{j}$ to some vertex from the column
$C_{j+1}$ for each $2i'-1\leq j\leq 4n-1$.

Since $2n+2j'-2\geq 2n$ and $2i'-1\leq 2n$, it follows that for each $0\leq i\leq 4n-1$ the solution $E'$ must contain a path connecting some vertex from column $C_i$ to some vertex from column $C_{i+1}$. Each such path has to either be a path of one which must be an inrow right edge, or a path of two edges consisting of a shorcut and an interrow edge.
But Claim~\ref{claim:at-most-2-shortcuts} implies $E'$ contains at most two shortcuts. Therefore, $E'$ contains at least $4n-2$ inrow right edges. Suppose $E'$ contains at
least $4n-1$ inrow right edges. We have already seen the contribution of source, sink, terminal, inrow up and interrow edges is
$2B^{5} + 4B^{4} + (2n+1)B^{3} + (n+1)B^{2}$. If $E'$ contains at least $4n-1$ inrow right edges, then the weight of
$E'$ is at least $2B^{5} + 4B^{4} + (2n+1)B^{3} + (n+1)B^{2} + (4n-1)B$.

However, from Equation~\ref{eqn:c-star}, we get that
\begin{align*}
C^*_n &= 2B^{5} + 4B^{4} + (2n+1)B^{3} + (n+1)B^{2} + (4n-2)B + (n+1)^2\\
&= 2B^5 + 4B^4 + (2n+1)B^3 + (n+1)B^2 + (4n-2)B + 4n^2\\
&< 2B^5 + 4B^4 + (2n+1)B^3 + (n+1)B^2 + (4n-1)B,
\end{align*}
since $B=18n^2 > 4n^2$.
\end{proof}

From Claim~\ref{claim:at-most-2-shortcuts} and the proof of Claim~\ref{claim:exactly-4n-2-inrow-right}, we can make the
following observation:
\begin{observation}
\label{claim:exactly-2-shortcuts} $E'$ contains exactly two shortcuts.
\end{observation}

Let the shortcuts used be $e_{i''}$ and $f_{j''}$. Recall that Claim~\ref{claim:one-in-out} implies that at most one edge incident to $P$ and
at most one edge incident to $Q$ is used in $E'$. Therefore, if we show that $i'=j'$, then it follows that $E'$ represents $\beta=i'=j'$.

\begin{claim}
The following inequalities hold:
\begin{itemize}
\item $i''\geq i'$ and $j''\leq j'$
\item $i''\geq j''$
\end{itemize}
\label{claim:properties-of-shortcut-indices}
\end{claim}
\begin{proof}
To use the shortcut $e_{i''}$, we need to use the lower half of a pair of interrow edges from column $C_{2i''-1}$.
Claim~\ref{claim:interrow-monotone} implies $i'\leq \ell_1$ and the pairs of interrow edges used are monotone from left to right. Hence
$i''\geq i'$. Similarly, to use the shortcut $f_{j''}$, we need to use the upper half of an interrow edge from Column
$C_{2n+2j''-1}$. Claim~\ref{claim:interrow-monotone} implies $n+j'\geq \ell_{n+1}\geq n+j''$. Hence $j''\leq j'$.

Since we use the shortcut $e_{i''}$ it follows that $\ell_{n-i''+1} = i''$. Similarly, since we use the shortcut $f_{j''}$ it follows that $\ell_{n-j''+2}=n+j''$. As $1\leq i'',j''\leq n$ it follows that $n+j''>i''$. By mono    tonicity of the $\ell$-sequence shown in Claim~\ref{claim:interrow-monotone}, we have $n-j''+2> n-i''+1$, i.e., $i''\geq j''$.
%
\end{proof}

\begin{theorem}
\label{thm:connector-represents} The weight of $E'$ is exactly $C^*_n$, and $E'$ represents some integer $\beta\in [n]$.
\end{theorem}
\begin{proof}
As argued above it is enough to show that $i'=j'$. We have already seen $E'$ has the following contribution to its weight:
\begin{itemize}
\item The source edge incident to $p_{i'}$ has weight $B^{5}+(n-i'+1)$ by Claim~\ref{claim:one-in-out}.
\item The sink edge incident to $q_{j'}$ has weight $B^{5}+ j'$ by Claim~\ref{claim:one-in-out}.
\item The terminal edges incur weight $4B^4$ by Claim~\ref{claim:4-blue-dotted}.
\item The inrow up edges incur weight $(2n+1)B^{3}$ by Claim~\ref{claim:one-red-inrow}.
\item The interrow edges incur weight $(n+1)B^2$ by Claim~\ref{claim:one-interrow}.
\item The inrow right edges incur weight $(4n-2)B$ by Claim~\ref{claim:exactly-4n-2-inrow-right}.
\item The shortcut $e_{i''}$ incurs weight $n\cdot i''$ and $f_{j''}$ incurs weight $n(n-j''+1)$ by
    Claim~\ref{claim:exactly-2-shortcuts}.
\end{itemize}
Thus we already have a weight of
\begin{equation}
C^{**}=(2B^{5}+(n-i'+j'+1))+ 4B^4 + (2n+1)B^{3} + (n+1)B^2 + (4n-2)B + n(n-j''+i''+1)
\label{eqn:c-star-star}
\end{equation}
Observe that adding any edge of non-zero weight to $E'$ (other than the ones mentioned above) increases the weight $C^{**}$ by
at least $B$, since Claim~\ref{claim:at-most-2-shortcuts} does not allow us to use any more shortcuts.
Equation~\ref{eqn:c-star} and Equation~\ref{eqn:c-star-star} imply $C^{**}+B - C^*_n = B - n(i'-j')-(j''-i'') \geq 20n^3 - n(i'-j')-(j''-i'')\geq 0$, since $i', i'', j', j''\in [n]$. This implies that the weight of $E'$ is exactly $C^{**}$. We now show that in fact $C^{**}-C^*_n \geq 0$, which will imply that $C^{**}=C^*_n$. From Equation~\ref{eqn:c-star} and Equation~\ref{eqn:c-star-star}, we have $C^{**} - C^*_n = (j'-i')+n(i''-j'')$. We now show that this
quantity is non-negative. Recall that from Claim~\ref{claim:properties-of-shortcut-indices}, we have $i''\geq j''$.
\begin{itemize}
\item If $i''>j''$ then $n(i''-j'')\geq n$. Since $j',i'\in [n]$, we have $j'-i'\geq 1-n$. Therefore, $(j'-i')+n(i''-j'')\geq n+(1-n) =1$
\item If $i''=j''$ then by Claim~\ref{claim:properties-of-shortcut-indices} we have $i'\leq i'' = j'' \leq j'$. Hence $(j'-i')\geq 0$ and so $(j'-i')+n(i''-j'')\geq 0$.
\end{itemize}
Therefore $C^{**}= C^{*}_n$, i.e., $E'$ has weight exactly $C^{*}_n$. However $C^*_n=C^{**}$ implies
\begin{equation}
j'-i' + n(i''-j'') = 0
\label{eqn:equality}
\end{equation}
Since $i', j'\in [n]$ we have $n-1\geq j'-i' \geq 1-n$. If $i''\neq j''$ then $n(i''-j'')\geq n$ and hence $j'-i' + n(i''-j'') \geq (1-n)+n \geq 1$. Contradiction. Hence, we have $j''=i''$ and therefore Equation~\ref{eqn:equality} implies $j'=i'$, i.e, $E'$ is represented by $i'=j'\in [n]$.
\end{proof}

\section{Proof of Lemma~\ref{lem:main-gadget}: constructing the main gadget}
\label{main:proof}


\begin{figure}[H]
\centering
{\footnotesize \def\svgwidth{0.7\linewidth}%
\executeiffilenewer{main-new-copy.svg}{main-new-copy.pdf}%
{inkscape -z -D --file=main-new-copy.svg %
--export-pdf=main-new-copy.pdf --export-latex}%
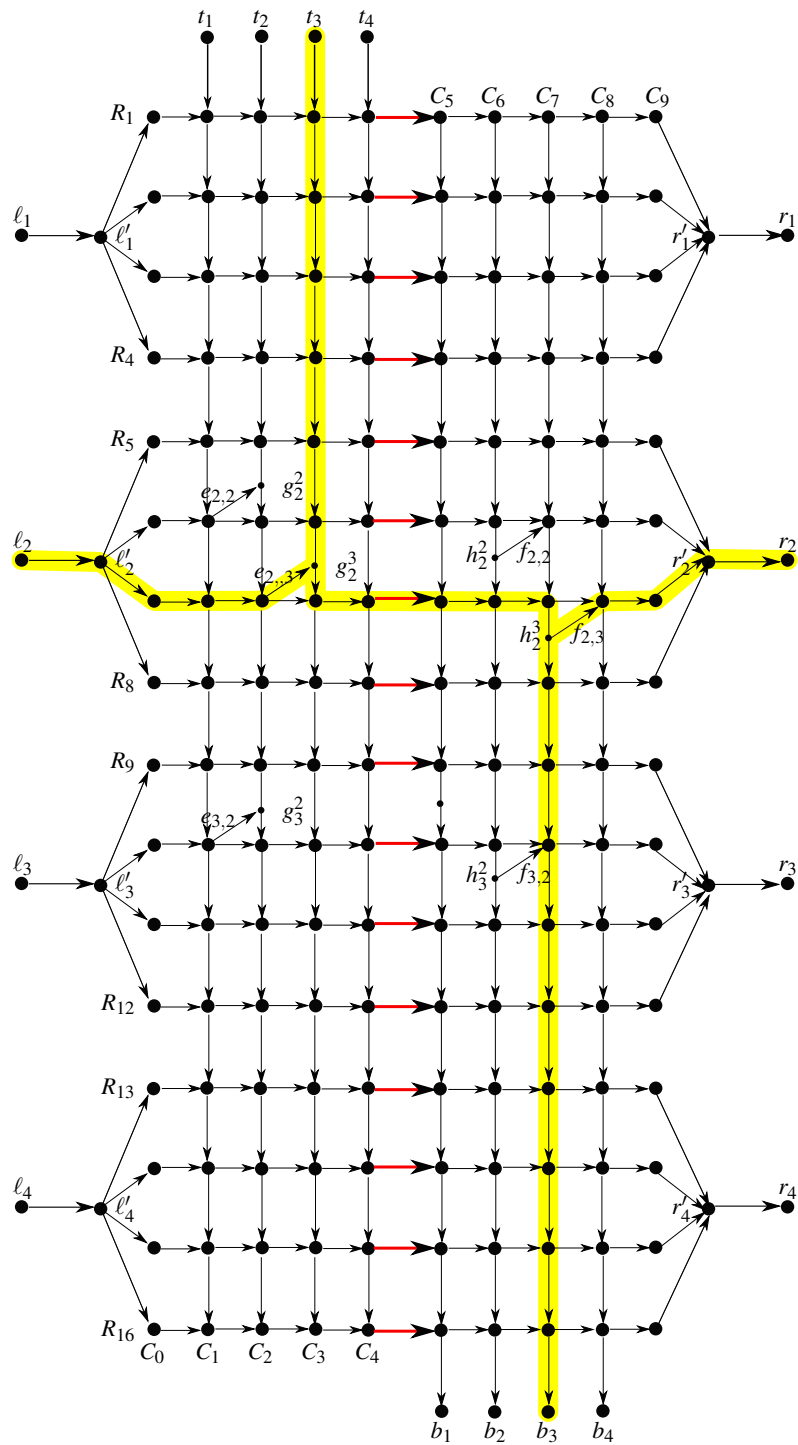%
}
\caption{The main gadget (for $n=4$) representing the set $\{(2,2),(2,3),(3,2)\}$. The highlighted edges represent the pair $(2,3)$. \label{fig:main-gadget}}
\end{figure}

We prove Lemma~\ref{lem:main-gadget} in this section, by constructing a main gadget satisfying the specifications of Section~\ref{sec:main:exists}.
Recall that, as discussed at the start of Section~\ref{sec:construction-of-g^*-scss-planar}, we may assume that $1< x,y< n$ holds for every $(x,y)\in S_{i,j}$.

\subsection{Different types of edges in main gadget}
\label{subsec:description-of-edges-in-main-gadget}

Before proving Lemma~\ref{lem:main-gadget}, we first describe the
construction of the main gadget in more detail (see
Figure~\ref{fig:main-gadget}).  The main gadget has $n^2$ rows denoted
by $R_1, R_2, \ldots, R_{n^2}$ and $2n+1$ columns denoted by $C_0,
C_1, \ldots, C_{2n+1}$. Let us denote the vertex at intersection of
row $R_i$ and column $C_j$ by $v_{i}^j$. We now describe the various
different kinds of edges in the main gadget.
\begin{enumerate}
\item \textbf{Left Source Edges}: For every $i\in [n]$, the edge $(\ell_i, \ell'_{i})$ is a left source edge.

\item \textbf{Right Sink Edges}: For every $i\in [n]$, the edge $(r'_i, r_{i})$ is a right sink edge.

\item \textbf{Top Source Edges}: For every $i\in [n]$, the edge $(t_i, v_{1}^{i})$ is a top source edge.

\item \textbf{Bottom Sink Edges}: For every $i\in [n]$, the edge $(v_{n^2}^{n+i}, b_{i})$ is a bottom sink edge.

\item \textbf{Source Internal Edges}: This is the set of  $n^2$ edges of the form $(\ell'_{i}, v_{j}^{0})$ for $i\in [n]$ and $n(i-1)+1
\leq j\leq n\cdot i$. We number the source internal edges from top to bottom, i.e., the edge $(\ell'_{i}, v_{j}^{0})$ is called the
$j^{th}$ source internal edge, where $i\in [n]$ and $n(i-1)+1 \leq j\leq n\cdot i$.

\item \textbf{Sink Internal Edges}: This is the set of $n^2$ edges of the form $(v_{j}^{2n+1}, r'_i)$ for $i\in [n]$ and $n(i-1)+1
\leq j\leq n\cdot i$. We number the sink internal edges from top to bottom, i.e., the edge $(v_{j}^{2n+1}, r'_i)$ is called the
$j^{th}$ sink internal edge, where $i\in [n]$ and $n(i-1)+1 \leq j\leq n\cdot i$.

\item \textbf{Bridge Edges}: This is the set of $n^2$ edges of the form $(v_{i}^{n}, v_{i}^{n+1})$ for $1\leq i\leq n^2$. We number
the bridge edges from top to bottom, i.e., the edge $(v_{i}^{n}, v_{i}^{n+1})$ is called the $i^{th}$ bridge edge. These edges are shown in red color in Figure~\ref{fig:main-gadget}.

\item \textbf{Inrow Right Edges}: For each $i\in [n^2]$ we call the $\rightarrow$ edges (except the bridge edge $(v_{i}^{n},
v_{i}^{n+1})$) connecting vertices of row $R_{i}$ as inrow right edges. Formally, the set of inrow right edges of row $R_{i}$ are given by $\{(v_{i}^{j}, v_{i}^{j+1})\ :\ 0\leq j\leq n-1 \} \bigcup \{(v_{i}^{j}, v_{i}^{j+1})\ :\ n+1\leq j\leq 2n\}$

\item \textbf{Interrow Down Edges}: For each $i\in [n^2 - 1]$ we call the $2n$ $\downarrow$ edges connecting vertices of row $R_{i}$
to $R_{i+1}$ as interrow down edges. The $2n$ interrow edges between row $R_{i}$ and $R_{i+1}$ are $(v_{i}^{j},
v_{i+1}^{j})$ for each $1\leq j\leq 2n$.

\item \textbf{Shortcut Edges}: There are $2|S|$ shortcut edges, namely $e_1, e_2, \ldots, e_{|S|}$ and $f_1, f_2, \ldots,f_{|S|}$.
The shortcut edge for a $(x,y)\in S$ for some $1 < x,y < n$ is defined the following way:
        \begin{itemize}
        \item Introduce a new vertex $g_{x}^y$ at the middle of the edge $(v_{n(x-1)+y-1}^{y}, v_{n(x-1)+y}^{y})$ to
            create two new edges $(v_{n(x-1)+y-1}^{y}, g_{x}^y)$ and $(g_{x}^y, v_{n(x-1)+y}^{y})$. Then the edge
            $e_{x,y}$ is $(v_{n(x-1)+y}^{y-1}, g_{x}^y)$.
        \item Introduce a new vertex $h_{x}^y$ at the middle of the edge $(v_{n(x-1)+y}^{n+y}, v_{n(x-1)+y+1}^{n+y})$
            to create two new edges $(v_{n(x-1)+y}^{n+y}, h_{x}^y)$ and $(h_{x}^y, v_{n(x-1)+y+1)}^{n+y})$. Then the
            edge $f_{x,y}$ is $(h_{x}^{y}, v_{n(x-1)+y}^{n+y+1})$.
        \end{itemize}
\end{enumerate}

\subsection{Assigning weights in the main gadget}

Define $B=11n^2$. We assign weights to the edges as follows:
\begin{enumerate}
\item Each left source edge has weight $B^4$.
\item Each right sink edge has weight $B^4$.
\item For every $1\leq i\leq n$, the $i^{th}$ top source edge $(t_i, v_{1}^{i})$ has weight $B^4$.
\item For every $1\leq i\leq n$, the $i^{th}$ bottom sink edge $(v_{n^2}^{n+i}, b_i)$ has weight $B^4$.
\item For each $i\in [n^2]$, the $i^{th}$ bridge edge $(v_{i}^{n}, v_{i}^{n+1})$ has weight $B^3$.
\item For each $i\in [n^2]$, the $i^{th}$ source internal edge has weight $B^{2}(n^2 - i)$.
\item For each $j\in [n^2]$, the $j^{th}$ sink internal edge has weight $B^{2}\cdot j$.
\item Each inrow right edge has weight $3 B$.
\item For each $(x,y)\in S$, both the shortcut edges $e_{x,y}$ and $f_{x,y}$ have weight $B$ each.
\item Each interrow down edge that does not have a shortcut incident to it has weight 2. If an interrow edge is split
    into two edges by the shortcut incident to it, then we assign a weight 1 to each of the two parts.


\end{enumerate}

Now we define the quantity $M^*_n$ stated in Lemma~\ref{lem:main-gadget}:
\begin{equation}
M^*_n = 4B^{4} + B^{3} + B^{2}n^{2} + B(6n-4) + 2(n^2-1).
\label{eqn:m-star}
\end{equation}

Next we are ready to prove the statements of Lemma~\ref{lem:main-gadget}.

\subsection{For every $(x,y)\in S$, there is a solution $E_{x,y}$ of weight $M^*_n$ that represents $(x,y)$}
\label{subsection:main-first-statement}

For $(x,y)\in S\subseteq [n]\times [n]$ define $z=n(x-1)+y$. Let $E_{x,y}$ be the union of the following sets of edges:

\begin{itemize}
\item The $x^{th}$ left source edge and $x^{th}$ right sink edge. This incurs a weight of $2B^4$.
\item The $y^{th}$ top source edge and the $y^{th}$ bottom sink edge.
 This incurs a weight of $2B^4$.
\item The $z^{th}$ bridge edge. This incurs a weight of $B^3$.
\item The $z^{th}$ source internal edge and $z^{th}$ sink internal edge. This incurs a weight of $B^{2}n^2$.
\item All inrow right edges from row $R_{z}$ except $(v_{z}^{y-1}, v_{z}^{y})$ and $(v_{z}^{n+y}, v_{z}^{n+y+1})$. This
    incurs a weight of $3B\cdot (2n-2)$.
\item The shortcut edges $e_{x,y}$ and $f_{x,y}$. This incurs a weight of $2B$.
\item The vertically downward path $v_{1}^{y}\rightarrow v_{2}^{y}\rightarrow \ldots \rightarrow v_{z}^{y}$ formed by interrow down edges in column $C_y$. This incurs a
    weight of $2(z-1)$.
\item The vertically downward path $v_{z}^{n+y}\rightarrow v_{z+1}^{n+y}\rightarrow \ldots \rightarrow v_{n^2}^{n+y}$ formed by interrow down edges in column $C_{n+y}$.
    This incurs a weight of $2(n^2 - z)$.
\end{itemize}

From the above paragraph and Equation~\ref{eqn:m-star}, it follows the total weight of $E_{x,y}$ is exactly $M^*_n$. Note that we did not include two inrow right edges, $(v_{z}^{y-1}, v_{z}^{y})$ and $(v_{z}^{n+y}, v_{z}^{n+y+1})$, from row $R_{z}$ in $E_{x,y}$. However, we can mimic the function of both these inrow right edges using shortcut edges and interrow down edges in $E_{x,y}$ as follows:
\begin{itemize}
  \item We can still travel from $v_{z}^{y-1}$ to $v_{z}^{y}$ in $E_{x,y}$ as follows: take the path $(v_{z}^{y-1}\rightarrow g_{x}^y \rightarrow v_{z}^{y})$.
  \item We can still travel from $(v_{z}^{n+y}$ to $v_{z}^{n+y+1})$ in $E_{x,y}$ via the path $(v_{z}^{n+y}\rightarrow h_{x,y} \rightarrow
v_{z}^{n+y+1})$.
\end{itemize}

The following observation follows from the previous paragraph:
\begin{observation}
\label{obs:reachability}
In $E_{x,y}$ we can reach $v_{z}^{j}$ from $v_{z}^{i}$ for any $2n+1\geq j\geq i\geq 0$.
\end{observation}
It remains to show that $E_{x,y}$ represents $(x,y)\in S$. It is easy to see that the first four conditions of Definition~\ref{defn:represents-main} are satisfied since the definition of $E_{x,y}$ itself gives the following:
\begin{itemize}
\item In $E_{x,y}$ the only outgoing edge from $L$ is the one incident to $\ell_x$
\item In $E_{x,y}$ the only incoming edge to $R$ is the one incident to $r_x$
\item In $E_{x,y}$ the only outgoing edge from $T$ is the one incident to $t_y$
\item In $E_{x,y}$ the only incoming edge to $B$ is the one incident to $b_y$
\end{itemize}
We now show that the last condition of Definition~\ref{defn:represents-main} is also satisfied by $E_{x,y}$:
\begin{enumerate}
\item There is a $\ell_x \leadsto r_x$ path in $E_{x,y}$ obtained by taking the edges in the following order:
        \begin{itemize}
        \item the left source edge $(\ell_x, \ell'_x)$,
        \item the source internal edge $(\ell'_x, v_{z}^{0})$,
        \item the horizontal path $v_{z}^0 \rightarrow v_{z}^1 \rightarrow \ldots v_{z}^n$ given by Observation~\ref{obs:reachability},
        \item the bridge edge $(v_{z}^n, v_{z}^{n+1})$,
        \item the horizontal path $v_{z}^{n+1} \rightarrow v_{z}^{n+2} \rightarrow \ldots v_{z}^{2n+1}$ given by Observation~\ref{obs:reachability},
        \item the sink internal edge $(v_{z}^{2n+1}, r'_x)$, and
        \item the right sink edge $(r'_x, r_x)$.
        \end{itemize}
\item There is a $t_y \leadsto b_y$ path in $E_{x,y}$ obtained by taking the edges in the following order:
        \begin{itemize}
        \item the top source edge $(t_y, v_{1}^{y})$,
        \item the downward path $v_{1}^y \rightarrow v_{2}^y \rightarrow \ldots v_{z}^y$ given by interrow down edges
            in column $C_y$,
        \item the horizontal path $v_{z}^y \rightarrow v_{z}^{y+1} \rightarrow \ldots v_{z}^{n}$ given by Observation~\ref{obs:reachability},
        \item the bridge edge $(v_{z}^n, v_{z}^{n+1})$,
        \item the horizontal path $v_{z}^{n+1} \rightarrow v_{z}^{n+2} \rightarrow \ldots v_{z}^{n+y}$ given by Observation~\ref{obs:reachability},
        \item the downward path $v_{z}^{n+y} \rightarrow v_{z+1}^{n+y} \rightarrow \ldots v_{n^2}^{n+y}$ given by
            interrow down edges in column $C_{n+y}$, and
        \item the bottom sink edge $(v_{n^2}^{n+y}, b_y)$.
        \end{itemize}
\end{enumerate}

Therefore, $E_{x,y}$ has weight $M^*_n$ and represents $(x,y)$.

\subsection{$E'$ satisfies the connectedness property and has weight at most $M^*_n \Rightarrow E'$ represents some $(\alpha,\beta)\in S$ and has weight exactly $M^*_n$}
\label{subsection:main-second-statement}

In this section we show that if a set of edges $E'$ satisfies the connectedness property and has weight $M^*_n$, then it
represents some $(\alpha,\beta)\in S$. We do this via the following series of claims and observations.

\begin{claim}
\label{claim:one-in-out-main} $E'$ contains
\begin{itemize}
\item exactly one left source edge,
\item exactly one right sink edge,
\item exactly one top source edge, and
\item exactly one bottom sink edge.
\end{itemize}
\end{claim}
\begin{proof}
Since $E'$ satisfies the connectedness property, it must contain at least one edge of each of the above types.
Without loss of generality, suppose we have at least two left source edges in $E'$. Then the weight of the edge set $E'$ is
greater than or equal to the sum of weights of these two left source edges and the weight of a right sink edge, the weight
of a top source edge, and the weight of a bottom sink edge. Thus if $E'$ contains at least two left source edges, then its
weight is at least $5B^{4}$.
However, from Equation~\ref{eqn:m-star}, we get that
\begin{align*}
M^*_n &= 4B^{4} + B^{3} + B^{2}n^{2} + B(6n-4) + 2(n^2-1)\\
&\leq 4B^4 + n\cdot B^3 + n\cdot B^3 + 6n\cdot B^3 + 2n\cdot B^3\\
&= 4B^4 + 10n\cdot B^3 \\
&< 5B^4,
\end{align*}
since $B=11n^2 > 10n$.
\end{proof}

Therefore, we can set up the following notation:
\begin{itemize}
\item Let $i_{L}\in [n]$ be the unique index such that the left source edge in $E'$ is incident to $\ell_{i_L}$.
\item Let $i_{R}\in [n]$ be the unique index such that the right sink edge in $E'$ is incident to $r_{i_R}$.
\item Let $i_{T}\in [n]$ be the unique index such that the top source edge in $E'$ is incident to $t_{i_T}$.
\item Let $i_{B}\in [n]$ be the unique index such that the bottom sink edge in $E'$ is incident to $b_{i_B}$.
\end{itemize}

%

\begin{claim}
\label{claim:exactly-1-bridge} The edge set $E'$ contains exactly one bridge edge.
\end{claim}
\begin{proof}
To satisfy the connectedness property, we need at least one bridge edge, since the bridge edges form a cut
between the top-distinguished vertices and the right-distinguished vertices as well as between the top-distinguished vertices
and the bottom-distinguished vertices. Suppose that the edge set $E'$ contains at least two bridge edges. This contributes a weight of
$2B^3$. We already have a contribution on $4B^4$ to weight of $E'$ from Claim~\ref{claim:one-in-out-main}. Therefore, the weight of $E'$ is at least $4B^4 + 2B^3$.
However, from Equation~\ref{eqn:m-star}, we get that
\begin{align*}
M^*_n &= 4B^{4} + B^{3} + B^{2}n^{2} + B(6n-4) + 2(n^2-1)\\
&\leq 4B^4 + B^3 + B^{2}n^{2} + 6n\cdot B + 2n^2\\
&\leq 4B^4 + B^3 + B^{2}n^{2} + 6n^{2}B^2 + 2n^{2}B^2 \\
&= 4B^4 + B^3 + 9B^{2}n^{2}\\
&< 4B^8 + 2B^3,
\end{align*}
since $B=11n^2 > 9n^2$.
\end{proof}

Let the index of the unique bridge edge in $E'$ (guaranteed by Claim~\ref{claim:exactly-1-bridge}) be $\gamma\in [n^2]$. The
connectedness property implies that we need to select at least one source internal edge incident to $\ell'_{i_L}$ and
at least one sink internal edge incident to $r'_{i_R}$. Let the index of the source internal edge incident to $\ell'_{i_L}$ be
$j_L$ and the index of the sink internal edge incident to $r'_{i_R}$ be $j_R$.

\begin{claim}
$i_L = i_R$ and $j_L = j_R=\gamma$. \label{claim:il=jl-and-ir=jr}
\end{claim}
\begin{proof}
By the connectedness property, there is a path from $\ell_{i_L}$ to some vertex in $r_{i_R}\cup b_{i_B}$. The paths starts with $\ell_{i_L}\rightarrow \ell'_{i_L}\rightarrow v_{j_L}^{1}$ and has to use the $\gamma^{th}$ bridge edge. By the construction of
the main gadget (all edges are either downwards or towards the right), this path can never reach any row $R_{\ell}$ for $\ell<
j_L$. Therefore, $\gamma\geq j_L$. By similar logic, we get $j_R\geq \gamma$. Therefore $j_{R}\geq j_{L}$.

If $j_R > j_L$, then the weight of the source internal edge and the sink internal edge is $B^{5}(n^2 - j_L + j_R) \geq
B^{5}(n^2 + 1)$. We already have a contribution of $4B^{4}+B^3$ to the weight of $E'$ from
Claim~\ref{claim:one-in-out-main} and Claim~\ref{claim:exactly-1-bridge}. Therefore, the weight of $E'$ is at least $4B^4 + B^3 + B^{2}(n^{2}+1)$. However, from Equation~\ref{eqn:m-star}, we get that
\begin{align*}
M^*_n &= 4B^{4} + B^{3} + B^{2}n^{2} + B(6n-4) + 2(n^2-1)\\
&\leq 4B^4 + B^3 + B^{2}n^2 + 6n\cdot B + 2n^2\\
&\leq 4B^4 + B^3 + B^{2}n^2 + 6n^{2}\cdot B + 2n^{2}\cdot B \\
&= 4B^4 + B^3 + B^{2}n^{2} + 8n^{2}\cdot B\\
&< 4B^4 + B^3 + B^{2}(n^{2}+1),
\end{align*}
since $B=11n^2 > 8n^2$.
Hence $j_R = j_L=\gamma$. Observing that $i_L = \lceil \frac{j_L}{n} \rceil$ and $i_R = \lceil \frac{j_R}{n} \rceil$, we
obtain $i_L=i_R$.
\end{proof}

Let $i_L=i_R = \alpha$ and $\gamma = n(\alpha-1)+\beta$. We will now show that $E'$ represents the pair $(\alpha,\beta)$. By Definition~\ref{defn:represents-main}, we need to prove the following four conditions:
\begin{enumerate}
\item The only left source edge in $E'$ is the one incident to $\ell_{\alpha}$ and the only right sink edge in $E'$ is the
    one incident to $r_{\alpha}$.

\item The pair $(\alpha,\beta)$ is in $S$.

\item The only top source edge in $E'$ is the one incident to $t_{\beta}$ and the only bottom sink edge in $E'$ is the one
    incident to $b_{\beta}$.

\item $E'$ has an $\ell_{\alpha}\leadsto r_{\alpha}$ path and an $t_{\beta}\leadsto b_{\beta}$ path.
\end{enumerate}

The first statement above follows from Claim~\ref{claim:one-in-out-main} and Claim~\ref{claim:il=jl-and-ir=jr}.
We now continue with the proof of the other three statements mentioned above:

\begin{claim}
\label{claim:all-from-row-gamma} $E'$ contains exactly $2n-2$ inrow right edges, all of them from row $R_{\gamma}$. As a
corollary, we get that there are two shortcuts incident to row $R_{\gamma}$, i.e., $(\alpha,\beta)\in S$ and also that $E'$ uses both these shortcuts.
\end{claim}
\begin{proof}
Note that by the construction of the main gadget, there can be at most two shortcut edges incident on the vertices of row $R_{\gamma}$.

Claim~\ref{claim:il=jl-and-ir=jr} implies $j_L = j_R=\gamma$. Hence the $\ell_{\alpha}\leadsto r_{\alpha}\cup b_{i_B}$ path in
$E'$ contains a $v_{\gamma}^{0}\leadsto v_{\gamma}^{n}$ subpath $P_1$. By the construction of the main gadget, we cannot
reach an upper row from a lower row. Hence this subpath $P_1$ must be the path $v_{\gamma}^{0}\rightarrow
v_{\gamma}^{2}\rightarrow \ldots \rightarrow v_{\gamma}^{n}$. This path $P_1$ can at most use the unique shortcut edge
incident to row $R_{\gamma}$ and column $C_{\beta}$ to replace an inrow right edge. Hence $P_1$ uses at least $n-1$ inrow
right edges, with equality only if $R_{\gamma}$ has a shortcut incident to it.

Similarly, the $\ell_{\alpha}\cup t_{i_T}\leadsto r_{\alpha}$ path in $E'$ contains a $v_{\gamma}^{n+1}\leadsto v_{\gamma}^{2n+1}$
subpath $P_2$. By the construction of the main gadget, we cannot reach an upper row from a lower row. Hence this subpath
$P_2$ must be the path $v_{\gamma}^{n+1}\rightarrow v_{\gamma}^{n+2}\rightarrow \ldots \rightarrow v_{\gamma}^{2n+1}$. This
path $P_2$ can at most use the unique shortcut edge incident to row $R_{\gamma}$ and column $C_{n+\beta}$ to replace an inrow
right edge. Hence $P_2$ uses at least $n-1$ inrow right edges, with equality only if $R_{\gamma}$ has a shortcut incident to it.

Clearly, the sets of inrow edges used by $P_1$ and $P_2$ are disjoint, and hence $E'$ contains at least $2n-2$ inrow right
edges from row $R_{\gamma}$. Suppose $E'$ contains at least $2n-1$ inrow right edges. Then it incurs a weight of $3B\cdot
(2n-1)$. From Claim~\ref{claim:one-in-out-main}, Claim~\ref{claim:exactly-1-bridge} and Claim~\ref{claim:il=jl-and-ir=jr} we
already have a contribution of $4B^4 + B^3 + B^{2}n^2$. Therefore the weight of $E'$ is at least $4B^4 + B^3 + B^{2}n^2+
3B\cdot (2n-1)$.

However, from Equation~\ref{eqn:m-star}, we get that
\begin{align*}
M^*_n &= 4B^{4} + B^{3} + B^{2}n^{2} + B(6n-4) + 2(n^2-1)\\
&\leq 4B^4 + B^3 + B^{2}n^2 + B(6n-4) + 2n^2\\
&< 4B^4 + B^3 + B^{2}n^2 + 3B\cdot (2n-1),
\end{align*}
since $B=11n^2 > 2n^2$.
Therefore, $E'$ can only contain at most $2n-2$ inrow right edges. Hence there must be two shortcut edges incident to row $R_{\gamma}$, which are both used by $E'$. Since $\gamma=n(\alpha-1)+\beta$, the fact that row $R_{\gamma}$ has shortcut edges incident to it implies $(\alpha,\beta)\in S$.
\end{proof}

To prove the third claim it is sufficient to show that $i_T = i_B=\beta$, since Claim~\ref{claim:one-in-out-main} implies $E'$ contains exactly one top source edge and exactly one bottom sink edge. Note that the remaining budget left for the weight of $E'$ is at most $2(n^2-1)$.

\begin{claim}
$i_T  = i_B=\beta$ \label{claim:i_T=beta=i_B}
\end{claim}
\begin{proof}
Recall that the only bridge edge used is the one on row $R_{\gamma}$. Moreover, the bridge edges form a cut between $T$ and $R\cup B$. Hence, to satisfy the connectedness property it follows that the $t_{i_T}\leadsto r_{\alpha}\cup b_{i_B}$ path in $E'$
contains a $v_{1}^{i_T}\leadsto v_{\gamma}^{n}$ subpath $P_3$. By Claim~\ref{claim:all-from-row-gamma},
all inrow right edges are only from row $R_{\gamma}$. As the only remaining budget is $2(n^2-1)$, we cannot use any other
shortcuts or inrow right edges since $B=11n^2 > 2(n^2 -1)$. Therefore, $P_3$ contains another $v_{1}^{i_T}\rightarrow v_{\gamma}^{i_T}$ subpath $P'_3$.
If $i_T\neq \beta$, then $P'_3$ incurs weight $2(\gamma-1)$. Note that we also pay a weight of 1 to use half of the interrow
edge when we use the shortcut edge (which we have to use due to Claim~\ref{claim:all-from-row-gamma}) which is incident to row $R_{\gamma}$ and column $C_{\beta}$.

Similarly, the $\ell_{\alpha}\cup t_{i_T}\leadsto b_{i_B}$ path in $E'$ contains a $v_{\gamma}^{n+1}\leadsto v_{n^2}^{n+i_B}$
subpath $P'_4$. By Claim~\ref{claim:all-from-row-gamma}, all inrow horizontal edges are only from row
$R_{\gamma}$. As the only remaining budget is $2(n^2-1)$, we cannot use any other shortcuts or inrow right edges. Therefore,
$P_4$ contains another $v_{\gamma}^{n+i_B}\leadsto v_{n^2}^{n+i_B}$ subpath $P'_4$. If $i_B\neq \beta$, then $P'_4$ incurs
weight $2(n^2 - \gamma)$. Note that we also pay a weight of 1 to use (half of) the interrow edge when we use the shortcut edge (which we have to use due to Claim~\ref{claim:all-from-row-gamma})
which is incident to row $R_{\gamma}$ and column $C_{n+\beta}$.

Suppose without loss of generality that $i_T\neq \beta$. Then $P'_3$ incurs a weight of $2(\gamma-1)$, and the half interrow edge
used incurs an additional weight of 1. In addition, path $P'_4$ incurs a weight of $2(n^2-\gamma)$. Hence the total weight incurred is
$2(\gamma-1)+ 1 + 2(n^2-\gamma) = 2(n^2 - 1)+ 1$ which is greater than our allowed budget. Hence $i_T = \beta$. It can be
shown similarly that $i_B = \beta$.
\end{proof}

\begin{claim}
$E'$ has an $\ell_{\alpha}\leadsto r_{\alpha}$ path and an $t_{\beta}\leadsto b_{\beta}$ path.
\label{claim:hori-vert-connectivity}
\end{claim}
\begin{proof}
First we show that $E'$ has an $\ell_{\alpha}\leadsto r_{\alpha}$ path by taking the following edges (in order)
\begin{itemize}
  \item The path $\ell_\alpha \rightarrow \ell'_{\alpha}\rightarrow v_{\gamma}^{0}$ which exists since $i_L = \alpha$ and $j_L = \gamma$
  \item The $v_{\gamma}^{0}\leadsto v_{\gamma}^{n}$ path $P_1$ guaranteed in proof of Claim~\ref{claim:all-from-row-gamma}
  \item The bridge edge $v_{\gamma}^{n}\rightarrow v_{\gamma}^{n+1}$ guaranteed by Claim~\ref{claim:exactly-1-bridge}
  \item The $v_{\gamma}^{n+1}\leadsto v_{\gamma}^{2n+1}$ path $P_2$ guaranteed in proof of Claim~\ref{claim:all-from-row-gamma}
  \item The path $r_\alpha \leftarrow r'_{\alpha}\leftarrow v_{\gamma}^{2n+1}$ which exists since $i_R = \alpha$ and $j_R = \gamma$
\end{itemize}

Next we show that $E'$ has an $t_{\beta}\leadsto b_{\beta}$ path by taking the following edges (in order)
\begin{itemize}
  \item The edge $t_\beta \rightarrow v_{1}^{\beta}$ which exists since $i_T = \beta$
  \item The $v_{1}^{\beta}\leadsto v_{\gamma}^{n}$ path $P_3$ guaranteed in proof of Claim~\ref{claim:i_T=beta=i_B}
  \item The bridge edge $v_{\gamma}^{n}\rightarrow v_{\gamma}^{n+1}$ guaranteed by Claim~\ref{claim:exactly-1-bridge}
  \item The $v_{\gamma}^{n+1}\leadsto v_{n^2}^{n+\beta}$ path $P_4$ guaranteed in proof of Claim~\ref{claim:i_T=beta=i_B}
  \item The edge $b_\beta \leftarrow v_{n^2}^{n+\beta}$ which exists since $i_B = \beta$
\end{itemize}

\end{proof}

Claim~\ref{claim:one-in-out-main}, Claim~\ref{claim:il=jl-and-ir=jr}, Claim~\ref{claim:all-from-row-gamma}, Claim~\ref{claim:i_T=beta=i_B} and Claim~\ref{claim:hori-vert-connectivity} together imply that $E'$ represents $(\alpha,\beta)\in S$ (see Definition~\ref{defn:represents-main}). We now show that weight of $E'$ is exactly $M^*_n$.

\begin{lemma}
Weight of $E'$ is exactly $M^*_n$
\end{lemma}
\begin{proof}
Claim~\ref{claim:one-in-out-main} contributes a weight of $4B^{4}$ to $E'$. Claim~\ref{claim:exactly-1-bridge} contributes a weight of $B^3$ to $E'$. From the proof of Claim~\ref{claim:il=jl-and-ir=jr}, we can see that $E'$ incurs weight $B^{2}n^2$ from the source
    internal edge and sink internal edge. Claim~\ref{claim:all-from-row-gamma} implies that $E'$ contains exactly $2n-2$ inrow right
    edges from row $R_{\gamma}$ and also both shortcuts incident to row $R_{\gamma}$. This incurs a cost of $3B(2n-2)+2B = B(6n-4)$. By arguments similar to that in the proof of Claim~\ref{claim:i_T=beta=i_B}, $E'$ contains at least $(\gamma-1)$ interrow edges from column $C_{\beta}$ and at least $(n^2 -\gamma)$ interrow edges from column $C_{n+\beta}$. Therefore, we have weight of $E'\geq 4B^4 + B^3 + B^{2}n^2 + B\cdot(6n-4)+ 2(\gamma-1) + 2(n^2 -\gamma)= 4B^4 + B^3 + B^{2}n^2 + B\cdot(6n-4)+ 2(n^2 -1) = M^*_n$. Hence the weight of $E'$ is exactly $M^*_n$.
\end{proof}
This completes the proof of the second statement of
Lemma~\ref{lem:main-gadget}.

\section{W[1]-hardness for SCSS in general graphs}

The main goal of this section  is to prove Theorem~\ref{thm:scss-main-hardness-general-graphs}.
We note that the reduction of Guo et al.~\cite{guo-et-al}
gives a reduction from \mcc which builds an equivalent instance of \scss with quadratic blowup in the number of terminals.
Hence using the reduction of Guo et al.~\cite{guo-et-al} only an $f(k)\cdot n^{o(\sqrt{k})}$ algorithm for SCSS can be ruled out under
ETH. We are able to improve upon this hardness by using the \csi (PSI) problem introduced by Marx~\cite{marx-beat-treewidth}. Our
reduction is also slightly simpler than the one given by Guo et al.

\begin{center}
\noindent\framebox{\begin{minipage}{6.00in}
\textbf{\csi} (PSI)\\
\emph{Input }: Undirected graphs $G=(V_G=\{g_1,g_2,\ldots,g_{\ell}\},E_G)$ and $H=(V_H,E_H)$, and a partition of $V_H$ into
disjoint subsets
               $H_1,H_2,\ldots, H_{\ell}$ \\
\emph{Question}: Is there an injection $\phi: V_G\rightarrow V_H$ such that
                \begin{enumerate}
                \item For every $i\in [\ell]$ we have $\phi(g_i)\in H_i$.
                \item For every edge $\{g_i,g_j\}\in E_G$ we have $\{\phi(g_i),\phi(g_j)\}\in E_H$.
                \end{enumerate}
\end{minipage}}
\end{center}

The PSI problem is so-called because the vertices of $H$ are partitioned into parts: one part corresponding to every vertex of $G$. Marx~\cite{marx-beat-treewidth} showed the following hardness result:

\begin{theorem}
\label{thm:marx-csi} 
Unless ETH fails, \csi cannot be solved in time $f(r)\cdot n^{o(r/\log r)}$ where $f$ is any
computable function, $r$ is the number of edges in $G$ and $n$ is the number of vertices in~$H$.
\end{theorem}

\begin{figure}[t]
\centering
\includegraphics[width=6in]{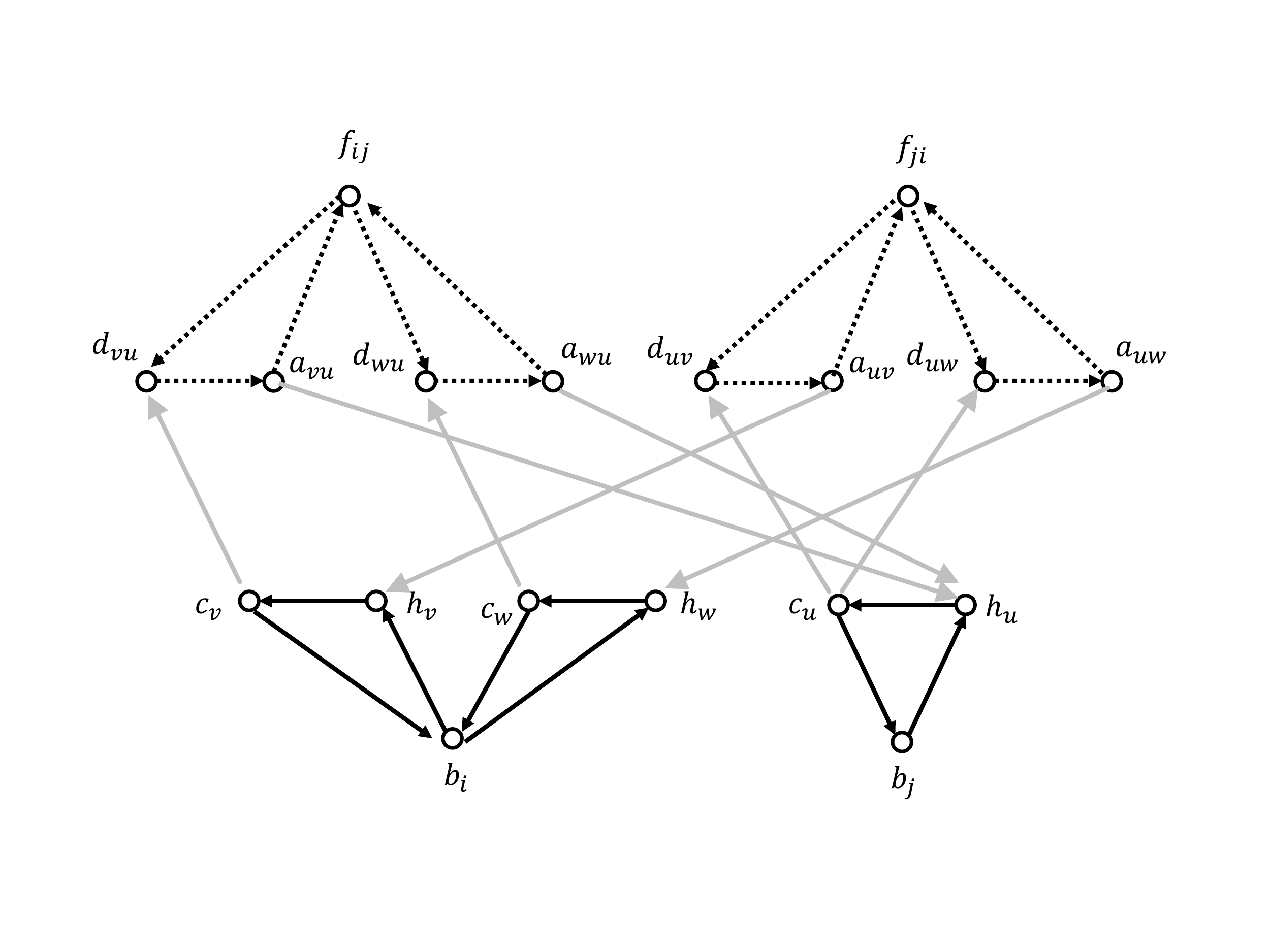}
\vspace{-15mm}
\caption{An illustration of the reduction from PSI to SCSS described in Theorem~\ref{thm:scss-main-hardness-general-graphs} for the special case when $V_G=\{g_1, g_2\}, E_G = g_1 - g_2$ and $H$ is a path on three vertices $v-u-w$ with $H_1 = \{v,w\}$ and $H_2 = \{u\}$.
\label{fig:scss}}
\end{figure}

By giving a reduction from \csi to \scss where $k=O(|E_G|)$ we will obtain a $f(k)\cdot n^{o(k/\log k)}$ hardness for SCSS under the ETH,
where $k$ is the number of terminals. Consider an instance $(G,H)$ of \csi. We now build an instance $(G^*, T^*)$ of \scss as follows:
\begin{itemize}
\item $B = \{b_{i}\ |\ i\in [\ell]\}$
\item $C = \{c_v\ |\ v\in V_H\}$
\item $H = \{h_v\ |\ v\in V_H\}$
\item $D = \{d_{uv}\cup d_{vu}\ |\ \{u,v\}\in E_H\}$
\item $A = \{a_{uv}\cup a_{vu}\ |\ \{u,v\}\in E_H\}$
\item $F = \{ f_{ij}\ |\ 1\leq i,j\leq \ell \ |\ g_{i}g_{j}\in E_G \}$
\item $V^{*}= B\cup C\cup H\cup D\cup A\cup F$
\item $E_1 = \{ (c_{v},b_{i})\ |\ v\in H_i, 1\leq i\leq \ell  \}$
\item $E_2 = \{ (b_{i},h_{v})\ |\ v\in H_i, 1\leq i\leq \ell \}$
\item $E_3 = \{ (h_{v},c_{v})\ |\ v\in V_H \}$
\item $E_4 = \{ (c_{v},d_{vu})\ |\ \{u,v\}\in E_H \}$
\item $E_5 = \{ (a_{vu},h_{u})\ |\ \{u,v\}\in E_H \}$
\item $E_6 = \{ (d_{vu},a_{vu})\ |\ \{u,v\}\in E_H \}$
\item $E_7 = \{ (f_{ij},d_{vu})\cup (a_{vu},f_{ij})\ |\ \{u,v\}\in E_H; v\in H_i; u\in H_j; 1\leq i,j \leq \ell \}$
\item $E^{*} = E_1 \cup E_2 \cup E_3 \cup E_4 \cup E_5\cup E_6\cup E_7$
\item The set of terminals is $T^*=B\cup F$.
\end{itemize}

This completes the construction of the graph $G^* = (V^*, E^*)$.
An illustration of the construction for a small graph is given in Figure~\ref{fig:scss}. In the instance of \csi we can assume the graph $G$ is connected,
otherwise we can solve the problem for each connected component. Therefore, we have that $k=|T| = \ell + 2|E_G| = O(|E_G|)$. For ease of argument, we distinguish the different types of edges of $G^*$ as follows (see Figure~\ref{fig:scss}):
\begin{itemize}
\item Edges of $E_1 \cup E_2\cup E_3$ are denoted using black  edges
\item Edges of $E_4 \cup E_5$ are denoted using light/gray edges
\item Edges of $E_6 \cup E_7$ are denoted using dotted edges
\end{itemize}

We now show two lemmas which complete the reduction from \csi to \scss.

\begin{lemma}
\label{lemma:scss-hardness-general-graphs-reduction-easy} If the instance $(G,H)$ of \csi answers YES then the instance $(G^*,T^*)$ of \scss has a solution of size $\leq 3\ell+10|E_G|$.
\end{lemma}
\begin{proof}

Suppose the instance $(G,H)$ of \csi answers YES and let $\phi$ be the injection from $V_G \rightarrow V_H$. Then we claim the
following set $M'$ of $3\ell+10|E_G|$ edges forms a solution for the \scss instance:
\begin{itemize}
\item $M_1 = \{ (h_{\phi(g_i)},c_{\phi(g_i)})\ |\ i\in [\ell] \}$
\item $M_2 = \{ (b_{i},h_{\phi(g_i)})\ |\ i\in [\ell] \}$
\item $M_3 = \{ (c_{\phi(g_i)},b_{i})\ |\ i\in [\ell] \}$
\item $M_4 = \{ (c_{\phi(g_i)},d_{\phi(g_i)\phi(g_j)})\cup (d_{\phi(g_i)\phi(g_j)},a_{\phi(g_i)\phi(g_j)})\cup
    (a_{\phi(g_i)\phi(g_j)},h_{\phi(g_j)})\ |\ g_{i}g_{j}\in E_G; 1\leq i,j\leq \ell \}$.
\item $M_5 = \{ (f_{ij},d_{\phi(g_i)\phi(g_j)})\cup (a_{\phi(g_i)\phi(g_j)},f_{ij})  \ |\ g_{i}g_{j}\in E_G; 1\leq
    i,j\leq \ell \}$.
\item $M' = M_1 \cup M_2 \cup M_3 \cup M_4\cup M_5$
\end{itemize}
First consider $i\neq j$ such that $g_{i}g_{j}\in E_G$. Then there is a $b_i\leadsto b_j$ path in $M'$, namely $b_i
\rightarrow h_{\phi(g_i)}\rightarrow c_{\phi(g_i)}\rightarrow d_{\phi(g_i)\phi(g_j)}\rightarrow
a_{\phi(g_i)\phi(g_j)}\rightarrow h_{\phi(g_j)}\rightarrow c_{\phi(g_j)}\rightarrow b_j$. Generalizing this and observing
$G$ is connected we can see any two terminals in $B$ are strongly connected. Now consider two terminals $f_{ij}$ and $b_{q}$
such that $1\leq i,j,q\leq \ell$. The existence of the terminal $f_{ij}$ implies $g_{i}g_{j}\in E_G$ and hence
$\phi(g_i)\phi(g_j)\in E_H$. There is a path in $M'$ from $f_{ij}$ to $b_q$: use the path $f_{ij}\leadsto
d_{\phi(g_i)\phi(g_j)}\rightarrow a_{\phi(g_i)\phi(g_j)}\rightarrow h_{\phi(g_j)}\rightarrow c_{\phi(g_j)}\rightarrow b_j$
followed by the $b_{j}\leadsto b_{q}$ path (which was shown to exist above). Similarly there is a path in $M'$ from $b_q$ to
$f_{ij}$: use the $b_q\leadsto b_i$ path (which was shown to exist above) followed by the path $b_i \leadsto
h_{\phi(g_i)}\rightarrow c_{\phi(g_i)}\rightarrow d_{\phi(g_i)\phi(g_j)}\rightarrow a_{\phi(g_i)\phi(g_j)}\rightarrow
f_{ij}$. Hence each terminal in $B$ can reach every terminal in $F$ and vice versa. Finally consider any two terminals $f_{ij}$ and $f_{st}$ in $F$: the terminal $f_{ij}$ can first reach $b_{i}$ and we have seen above that $b_i$ can reach any terminal in $F$.
This shows $M'$ forms a solution for the \scss instance.
\end{proof}

\begin{lemma}
\label{lemma:scss-hardness-general-graphs-reduction-hard}If the instance $(G^*,T^*)$ of \scss has a solution of size $\leq 3\ell+10|E_G|$ then the instance $(G,H)$ of \csi answers YES.
\end{lemma}
\begin{proof}
Let $X$ be a solution of size $3\ell+10|E_G|$ for the instance $(G^*, T^*)$ of SCSS. Consider a terminal $f_{ij}\in F$. The only out-neighbors of $f_{ij}$ are vertices from $D$, and hence $X$ must contain an
edge $(f_{ij},d_{vu})$ such that $v\in H_i$ and $u\in H_j$. However the only neighbor of $d_{vu}$ is $a_{vu}$, and hence $X$ has to contain this edge as well. Finally, $X$ must also contain one incoming edge into $f_{ij}$ since we desire strong connectivity. So for each terminal $f_{ij}$, we need three ``private" dotted edges in the sense that every terminal in $F$ needs three such edges in any
optimum solution. This uses up $6|E_G|$ of the budget since $|F|=2|E_G|$. Referring to Figure~\ref{fig:scss}, we can see any $f_{ij}\in F$ needs
two ``private" light edges in $X$: one edge coming out of some vertex in $A$ and some edge going into a vertex of $D$. This uses up
$4|E_G|$ more from the budget leaving us with only $3\ell$ edges.

Consider $b_i$ for $i\in [\ell]$. First we claim that $X$ must contain at least three black edges for $b_i$ to have incoming and outgoing
paths to the other terminals. The only outgoing edge from $b_i$ is to vertices of $H$, and hence we need to pick an edge $(b_i,h_{v})$ such that $v\in H_i$. Since the only out-neighbor of $h_{v}$ is $c_{v}$, it follows that $X$ must pick this edge as well. Additionally, $X$
also needs to contain at least one incoming edge into $b_i$ to account for incoming paths from other terminals to $b_i$. So each $b_i$ needs to
have at least three edges selected in order to have incoming and outgoing paths to other terminals. Moreover, all these edges
are clearly ``private", i.e., different for each $b_i$. But as seen in the previous paragraph, our remaining budget was at most $3\ell$. Hence $X$ selects exactly three such edges for each $b_i$. We now claim that once $X$ contains the edges $(b_i,h_{v})$ and $h_{v},c_{v}$ such that $v\in H_i$ then $X$ must also contain the edge $(c_{v},b_i)$. Suppose not, and for incoming towards $b_i$ the solution $X$
selects the edge $(c_{w},b_i)$ for some $w\in H_i$ such that $w \neq v$. Then since $h_{w}$ is the only neighbor of $c_{w}$, the solution $X$ would
be forced to select this edge as well. This implies that at least four edges have been selected for $b_i$, which is a contradiction. So for every $i\in [\ell]$, there
is a vertex $v_i\in H_i$ such that the edges $(b_i,h_{v_i}), (h_{v_i},c_{v_i})$ and $(c_{v_i},b_i)$ are selected in the
solution for the \scss instance. Further these are the only black edges in $X$ corresponding to $b_i$ (refer to
Figure~\ref{fig:scss}). It also follows for each $f_{ij}\in F$, the solution $X$ contains exactly three of the dotted edges (we argued above that each $f_{ij}$ needs three dotted edges, and the budget now implies that this is the maximum we can allow).

Define $\phi: V_G\rightarrow V_H$ by $\phi(g_i)=v_i$ for each $i\in [\ell]$. Since $v_i\in H_i$ and the sets $H_{1}, H_2, \ldots, H_{\ell}$ form a
disjoint partition of $V_H$, it follows that the function $\phi$ is an injection. Consider any edge $g_{i}g_{j}\in E_G$. We have seen above
that the solution $X$ contains exactly three dotted edges per $f_{ij}\in F$. Suppose for $f_{ij}\in F$ the solution $X$ contains the edges $(f_{ij}, d_{vu}),
(d_{vu},a_{vu})$ and $(a_{vu},f_{ij})$ for some $v\in H_i, u\in H_j$. The only incoming path for $f_{ij}$ is via $d_{vu}$.
Also the only outgoing path from $b_i$ is via $c_{v_i}$. If $v_{i}\neq v$ then we will need two other dotted edges to
reach $f_{ij}$, which is a contradiction since have already picked the allocated budget of three such edges. Hence, $v_i=v$. Similarly, it follows that $v_j=u$.
Finally, the existence of the vertex $d_{vu}$ implies $vu\in E_H$, i.e., $\phi(g_i)\phi(g_j)\in E_H$.
%
\end{proof}


\subsection{Proof of Theorem~\ref{thm:scss-main-hardness-general-graphs}}

Finally, we are now ready to prove Theorem~\ref{thm:scss-main-hardness-general-graphs} which is restated below:

\begin{reptheorem}{thm:scss-main-hardness-general-graphs}
Under ETH, the edge-unweighted version of the SCSS problem cannot be solved in time $f(k)\cdot n^{o(k/\log k)}$ where $f$ is any
computable function, $k$ is the number of terminals and $n$ is the number of vertices in the instance.
\end{reptheorem}
\begin{proof}
Lemma~\ref{lemma:scss-hardness-general-graphs-reduction-easy} and Lemma~\ref{lemma:scss-hardness-general-graphs-reduction-hard} together give a parameterized reduction from PSI to SCSS.
Observe that the number of terminals $k$ of the SCSS instance is $|B\cup F|=|V_G|+ 2|E_G| = O(|E_G|)$ since we had the assumption
that $G$ is connected. The number of vertices in the SCSS instance is $|V^*| = |V_G|+2|V_H|+4|E_H|+2|E_G| = O(|E_H|)$.
Therefore from Theorem~\ref{thm:marx-csi} we can conclude that under ETH there is no $f(k)\cdot n^{o(k/\log k)}$ algorithm for SCSS
where $n$ is the number of vertices in the graph and $k$ is the number of terminals.
\end{proof}

\section{W[1]-hardness for DSN in planar DAGs}

The main goal of this section is to prove Theorem~\ref{thm:dsn-w[1]-hardness} which is restated below

\begin{reptheorem}{thm:dsn-w[1]-hardness}
The edge-unweighted version of the \dsn problem is W[1]-hard parameterized by the number $k$ of terminal pairs, even when the input is restricted to planar directed acyclic graphs (DAGs). Moreover, there is no $f(k)\cdot n^{o(k)}$ algorithm for any computable function $f$, unless the ETH fails.
\end{reptheorem}

Note that this shows that the $n^{O(k)}$ algorithm of Feldman-Ruhl is asymptotically optimal. To prove Theorem~\ref{thm:dsn-w[1]-hardness}, we reduce from the
\textsc{Grid Tiling} problem introduced by Marx~\cite{marx-beat-treewidth}.

\begin{center}
\noindent\framebox{\begin{minipage}{6.00in}
\textbf{$k\times k$ \textsc{Grid Tiling}}\\
\emph{Input }: Integers $k, n$, and $k^2$ non-empty sets $S_{i,j}\subseteq [n]\times [n]$ where $1\leq i, j\leq k$\\
\emph{Question}: For each $1\leq i, j\leq k$ does there exist an entry $s_{i,j}\in S_{i,j}$ such that
\begin{itemize}
\item If $s_{i,j}=(x,y)$ and $s_{i,j+1}=(x',y')$ then $x=x'$.
\item If $s_{i,j}=(x,y)$ and $s_{i+1,j}=(x',y')$ then $y=y'$.
\end{itemize}
\end{minipage}}
\end{center}



 \begin{figure}[t]
 \centering
 \includegraphics[height=5in]{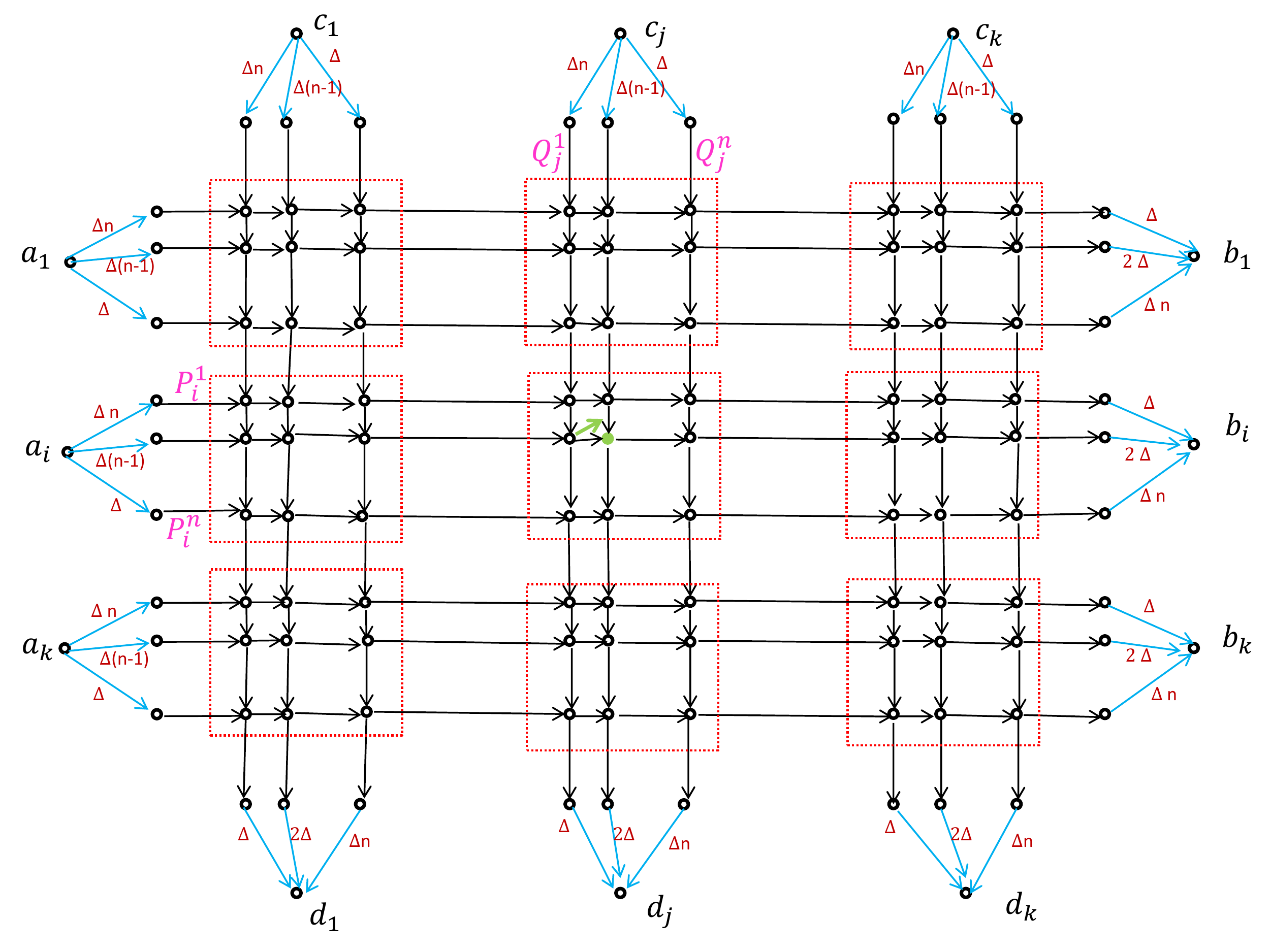}
 \caption{The instance of \dsn created from an instance of \gt.
 \label{fig:dsn}}
 \end{figure}

Consider an instance $(k,n, \{S_{i,j\ :\ 1\leq i,j \leq k}\})$ of \textsc{Grid Tiling}. We now build an instance $(G,T)$ of edge-weighted \dsn as shown in Figure~\ref{fig:dsn}.
Set $T= \{(a_i,b_i) \cup(c_i,d_i)\ : i\in [k]\}$, i.e., we have $2k$ terminal pairs. We introduce $k^2$ red gadgets where
each gadget is an $n\times n$ grid. Set the weight of each black edge to $2$.

\begin{definition}
An $a_i \leadsto b_i$ \emph{canonical} path is a path from $a_i$ to $b_i$ which starts with a blue edge coming out of $a_i$,
then follows a horizontal path of black edges and finally ends with a blue edge going into $b_i$. Similarly, a $c_j\leadsto
d_j$ \emph{canonical} path is a path from $c_j$ to $d_j$ which starts with a blue edge coming out of $c_j$, then follows a
vertically downward path of black edges and finally ends with a blue edge going into $d_j$.
\end{definition}

There are $n$ edge-disjoint $a_i \leadsto b_i$ canonical paths: let us call them $P^{1}_{i}, P^{2}_{i}, \ldots, P^{n}_i$ as
viewed from top to bottom. They are named using magenta color in Figure~\ref{fig:dsn}. Similarly we call the canonical paths
from $c_j$ to $d_j$ as $Q^{1}_{j}, Q^{2}_{j}, \ldots, Q^{n}_j$ when viewed from left to right. For each $i\in [k]$ and
$\ell\in [n]$ we assign a weight of $\Delta(n+1-\ell)$ and $\Delta\ell$ to the first and last blue edges of $P^{\ell}_{i}$,
respectively. Similarly for each $j\in [k]$ and $\ell\in [n]$ we assign a weight of $\Delta(n+1-\ell)$ and $\Delta\ell$ to the
first and last blue edges of $Q^{\ell}_{j}$, respectively. Thus the total weight of the first and the last blue edges on any canonical path
is exactly $\Delta(n+1)$. The idea is to choose $\Delta$ large enough such that in any optimum solution the paths between the
terminals will be exactly the canonical paths. We will see later that $\Delta=5n^2$ will suffice for this purpose. Any canonical path consists of the following set of edges:
\begin{itemize}
  \item Two blue edges (which sum up to $\Delta(n+1)$)
  \item $(k+1)$ black edges not inside the gadgets
  \item $(n-1)$ black edges inside each gadget
\end{itemize}
Since the number of gadgets each canonical path visits is $k$ and the weight of each black edge is 2, we have the
total weight of any canonical path is $\Delta(n+1)+2(k+1)+2k(n-1)$.

\begin{figure}
 \centering
 \includegraphics[height=3in]{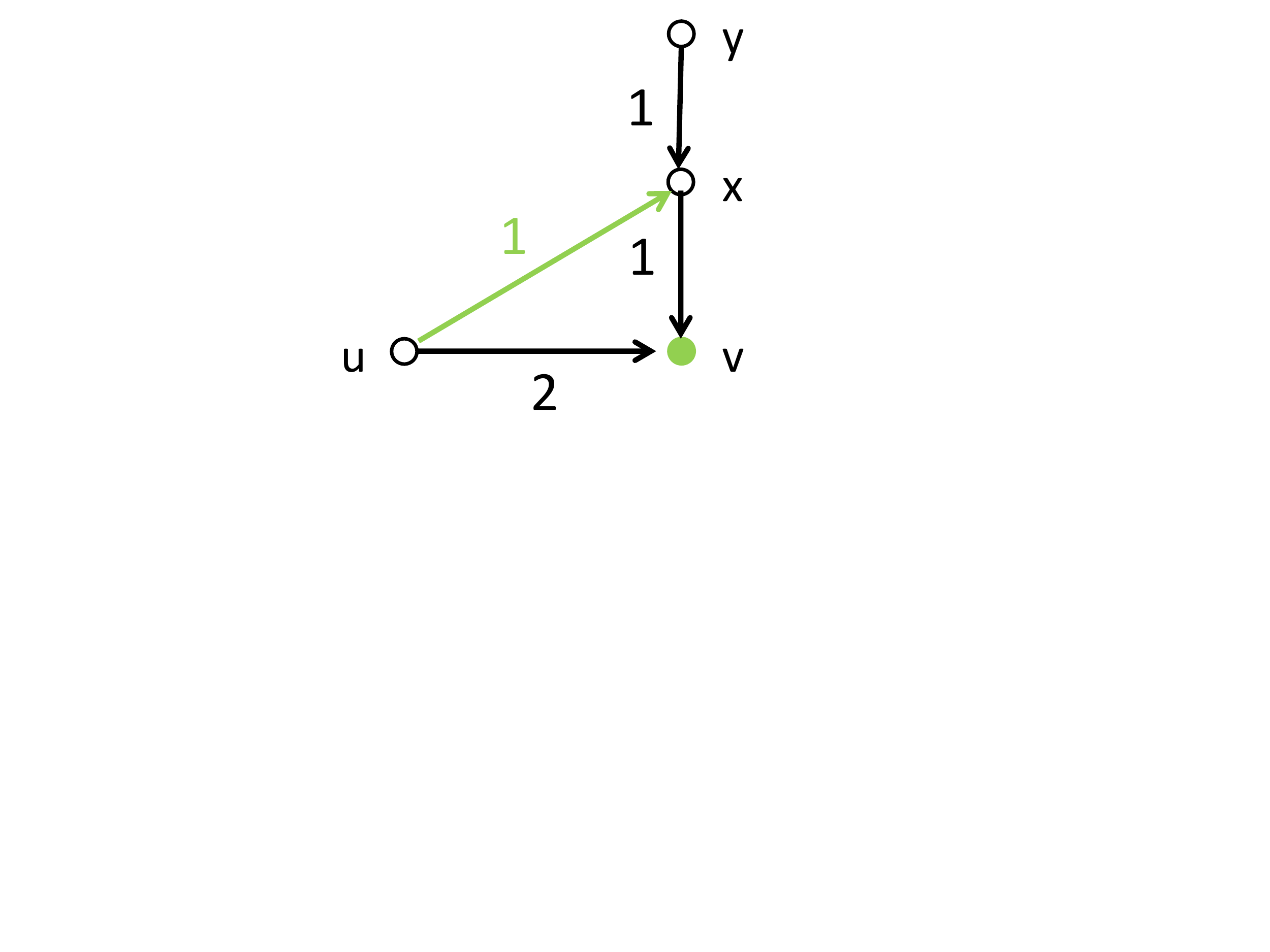}
 \vspace{-40mm}
 \caption{Let $u, v$ be two consecutive vertices on the canonical path $P^{\ell}_{i}$. Let $v$ be on the canonical path
 $Q^{\ell'}_{j}$ and let $y$ be the vertex preceding it on this path. If $v$ is a green vertex then we subdivide the edge $(y,v)$ by introducing a new vertex $x$
 and adding two edges $(y,x)$ and $(x,v)$ of weight 1. We also add an edge $(u,x)$ of weight 1. The idea is if both the edges $(y,v)$ and $(u,v)$ were
 being used initially then now we can save a weight of 1 by making the horizontal path choose $(u,x)$ and then we get $(x,v)$ for free, as it is already being used
 by the vertical canonical path.
 \label{fig:savings}}
 \end{figure}

Intuitively the $k^2$ gadgets correspond to the $k^2$ sets in the \textsc{Grid Tiling} instance. Let us denote by $G_{i,j}$ the gadget
which contains all vertices which lie on the intersection of any $a_i \leadsto b_i$ path and any $c_j \leadsto d_j$ path. If $(x,y)\in
S_{i,j}$ then we color green the vertex in the gadget $G_{i,j}$ which is the unique intersection of the canonical paths
$P_{i}^{x}$ and $Q_{j}^{y}$. Then we add a shortcut as shown in Figure~\ref{fig:savings}. The idea is if both the $a_i
\leadsto b_i$ path and $c_j \leadsto d_j$ path pass through the green vertex then the $a_i \leadsto b_i$ path can save a weight
of 1 by using the green edge and a vertical edge to reach the green vertex, instead of paying a weight of 2 to use the
horizontal edge reaching the green vertex. It is easy to see that there is a solution (without using green edges) for the DSN
instance of weight $B^* = 2k \Big ( \Delta(n+1)+2(k+1)+2k(n-1)\Big)$: each terminal pair just uses a canonical path and these
canonical paths are pairwise edge-disjoint.

The following assumption will be helpful in handling some of the border cases of the gadget construction. We may assume that $1 < \min\{x,y\}$ holds for every $(x,y)\in S_{i,j}$: indeed, we can increase $n$ by one
and replace every $(x,y)$ by $(x+1,y+1)$ without changing the problem. Hence, no green vertex can be in the first row or first column of any gadget. Combining this
fact with the orientation of the edges we get the only gadgets which can intersect any $a_i\leadsto b_i$ path are $G_{i,1},
G_{i,2}, \ldots, G_{i,k}$. Similarly the only gadgets which can intersect any $c_j\leadsto d_j$ path are $G_{1,j}, G_{2,j},
\ldots, G_{k,j}$. This completes the construction of the instance $(G,T)$ of \dsn.

Lemmas~\ref{lem:dsn-redn-easy} and \ref{lem:dsn-redn-hard} below prove that the reduction described above is indeed a correct reduction from \textsc{Grid Tiling} to DSN.

\begin{lemma}
\label{lem:dsn-redn-easy} If the instance $(k,n, \{S_{i,j\ :\ 1\leq i,j \leq k}\})$ of \textsc{Grid Tiling} has a solution then the instance $(G,T)$ of \dsn has a solution of weight at most $B^* - k^2$.
\end{lemma}
\begin{proof}
For each $1\leq i,j\leq k$ let $s_{i,j}\in S_{i,j}$ be the entry in the solution of the \textsc{Grid Tiling} instance.
Therefore for every $i\in k$ we know that each of the $k$ entries $s_{i,1}, s_{i,2}, \ldots, s_{i,k}$ have the same first coordinate $\alpha_{i}$. Similarly for every $j\in [k]$ each of the $k$ vertices $s_{1,j}, s_{2,j}, \ldots, s_{k,j}$
has the same second coordinate $\gamma_j$. For each $j\in [k]$ we use the canonical path $Q^{\gamma_j}_{j}$ to satisfy the terminal for $(c_j,d_j)$. For each $i\in [k]$, we essentially use the canonical path $P_{i}^{\alpha_i}$ with the following modifications: for each $j\in [k]$, take the shortcut green edge (as shown in Figure~\ref{fig:savings}) when we encounter the green vertex (this is guaranteed to happen since $(\alpha_i, \gamma_j)=s_{i,j}\in S_{i,j}$) in $G_{i,j}$ at intersection of the canonical paths $P_{i}^{\alpha_i}$ and $Q_{j}^{\gamma_j}$. Hence, overall we save a total of $k^2$: a saving of one per gadget. Thus, we have produced a solution for the instance $(G,T)$ of weight $2k \Big ( \Delta(n+1)+2(k+1)+2k(n-1)\Big) - k^2 = B^*- k^2$.
%
\end{proof}

We now prove the other direction which is more involved. First we show some preliminary claims:

\begin{claim}
\label{claim:vertical-canonical} Any optimum solution for $(G,T)$ contains a $c_j\leadsto d_j$ canonical path for each $j\in
[k]$.
\end{claim}
\begin{proof}
Suppose to the contrary that there is an optimum solution $N$ for $(G,T)$ which does not contain a canonical $c_{j}\leadsto d_{j}$ path for some $j\in [k]$.
From the orientation of the edges, we know that there is a $c_{j}\leadsto d_j$ path in $N$ that starts with the
blue edge from $Q_{j}^{\ell}$ and ends with a blue edge from $Q_{j}^{\ell'}$ for some $\ell'
> \ell$. We create a new set of edges $N'$ from $N$ as follows:
\begin{itemize}
  \item Add all those edges of $Q_{j}^{\ell}$ which were not present in $N$. In particular, we add the last blue edge of $Q_{j}^{\ell}$ since $\ell'>\ell$
  \item Delete the last blue edge of $Q^{\ell'}_{j}$.
\end{itemize}

It is easy to see that $N'$ is also a solution for $(G,T)$: this is because $N'$ contains the canonical path $Q_{j}^{\ell}$ to satisfy the pair $(c_j, d_j)$, and the last (blue) edge of any $c_j\leadsto d_j$ canonical
path cannot be on any $a_i\leadsto b_i$ path for any $i\in [k]$. Changing the last blue edge saves us $(\ell'-\ell)\Delta \leq \Delta = 5n^2$. However
we have to be careful since we added some edges to the solution. But these edges are the internal (black) edges of
$Q^{\ell}_{j}$, and their weight is $\leq 2(k+1) + 2k(n-1) =2kn+2 < 5n^2 = \Delta$ since $1\leq k\leq n$. Therefore we are able to create a new
solution $N'$ whose weight is less than that of an optimum solution $N$, which is a contradiction.
\end{proof}

\begin{definition}
\label{defn-almost-canonical} An $a_i\leadsto b_i$ path is called an \emph{almost canonical} path if its first and last edges are blue edges from the same $a_i \leadsto b_i$ canonical path.
\end{definition}

Hence, an $a_i\leadsto b_i$ almost canonical path looks very similar to an $a_i \leadsto b_i$ canonical path, except it can replace some of the horizontal black edges by green edges and vertical black edges as shown in Figure~\ref{fig:savings}. However, note that by definition, an almost canonical path must however end on the same horizontal level on which it began. The proof of the next claim is very similar to that of Claim~\ref{claim:vertical-canonical}.

\begin{claim}
\label{claim:horizontal-canonical} Any optimum solution for DSN contains an $a_i\leadsto b_i$ \emph{almost canonical} path for every
$i\in [k]$.
\end{claim}
\begin{proof}
Suppose to the contrary that there is an optimum solution $N$ which does not contain an almost canonical $a_{i}\leadsto b_{i}$ path for some $i\in [k]$. Hence, the $a_{i}\leadsto b_{i}$ path in $N$ starts and ends at different levels. From the orientation of the edges, we know that there is a $a_{i}\leadsto b_i$ path in the optimum solution that starts
with the blue edge from $P_{i}^{\ell}$ and ends with a blue edge from $P_{i}^{\ell'}$ for some $\ell'
> \ell$ (note that the construction in Figure~\ref{fig:savings} does not allow any $a_i\leadsto b_i$ path to climb onto an upper level).

We create a new set of edges $N'$ from $N$ as follows:
\begin{itemize}
  \item Add all those edges of $P_{i}^{\ell}$ which were not present in $N$.
  Note that in particular, we add the last blue edge of $P_{i}^{\ell}$ since $\ell'>\ell$.
  \item Delete the last blue edge of $P^{\ell'}_{i}$.
\end{itemize}

It is easy to see that $N'$ is also a solution for $(G,T)$: this is because $N'$ contains the canonical path $P_{i}^{\ell}$ to satisfy the pair $(a_i, b_i)$, and the last (blue) edge of any $a_i\leadsto b_i$ canonical
path cannot be on any $c_j\leadsto d_j$ path for any $j\in [k]$. Changing the last edge saves us $(\ell'-\ell)\Delta \leq \Delta = 5n^2$. But we
have to careful since we also added some edges to the solution. The total weight of edges added is $\leq 2(k+1) + 2k(n-1) = 2kn+2 < 5n^2 = \Delta$ since $1\leq k\leq n$. So we are able to create a new solution $N'$ whose weight is less than that of an optimum solution $N$, which is a contradiction.
\end{proof}

\begin{lemma}
\label{lem:dsn-redn-hard} If the instance $(G,T)$ of \dsn has a solution of weight at most $B^* - k^2$ then the instance $(k,n, \{S_{i,j\ :\ 1\leq i,j \leq k}\})$ of \textsc{Grid Tiling} has a solution.
\end{lemma}
\begin{proof}
Consider any optimum solution $X$. By Claim~\ref{claim:vertical-canonical} and
Claim~\ref{claim:horizontal-canonical} we know that $X$ has an $a_i\leadsto b_i$ almost canonical path and a $c_j\leadsto
d_j$ canonical path for every $1\leq i,j\leq k$. Moreover these set of $2k$ paths form a solution for DSN. Since any optimum
solution is minimal, $X$ is the union of these $2k$ paths: one for each terminal pair. For each $i,j\in [k]$ let the $a_i \leadsto b_i$ almost canonical path in $X$ be $\overline{P}_{i}^{\alpha_i}$ and the $c_j \leadsto d_j$ canonical path in $X$ be $Q_{j}^{\gamma_j}$.

The $a_i\leadsto b_i$ almost canonical path $\overline{P}_{i}^{\alpha_i}$ and $c_j\leadsto d_j$ canonical path $Q_{j}^{\gamma_j}$ in $X$ intersect in a unique  vertex in the gadget $G_{i,j}$. If each $a_i \leadsto b_i$ path was canonical instead of almost canonical, then  the weight of $X$ would have been exactly $B^*$. However we know that weight of $X$ is at most $B^* -k^2$. It is easy to see any $a_i\leadsto b_i$ almost canonical path and any $c_j\leadsto d_j$ canonical path can have at most one edge in common: the edge which comes vertically downwards into the green
vertex (see Figure~\ref{fig:savings}). There are $k^2$ gadgets, and there is at most one edge per gadget which is used for two paths in $X$. Hence for each gadget $G_{i,j}$ there is exactly one edge which is used by both the $a_i\leadsto b_i$ almost canonical path and the $c_j\leadsto d_j$ canonical path in $X$. So the endpoint of each of these common edges must be a green vertex, i.e., $(\alpha_i, \gamma_j)\in S_{i,j}$ for each $i,j\in [k]$.
\end{proof}


\subsection{Proof of Theorem~\ref{thm:dsn-w[1]-hardness}}

Finally, we are now ready to prove ~\autoref{thm:dsn-w[1]-hardness} which is restated below:
\begin{reptheorem}{thm:dsn-w[1]-hardness}
The edge-unweighted version of the \dsn problem is W[1]-hard parameterized by the number $k$ of terminal pairs, even when the input is restricted to planar directed acyclic graphs (DAGs). Moreover, there is no $f(k)\cdot n^{o(k)}$ algorithm for any computable function $f$, unless the ETH fails.
\end{reptheorem}


\begin{proof}
Given an instance $(k,n, \{S_{i,j\ :\ 1\leq i,j \leq k}\})$ of \textsc{Grid Tiling}, we use the reduction described earlier in this section to build an instance $(G,T)$ of edge-weighted \dsn (see Figure~\ref{fig:dsn} for an illustration). It is easy to see that the total number of vertices in $G$ is $O(n^{2}k^{2})$ and moreover can be constructed in $\poly(n,k)$ time. Each grid is planar (green shortcut edges do not destroy planarity), and the grids are arranged again in a grid-like manner. Figure~\ref{fig:dsn} actually gives a planar embedding of $G$. Moreover, it is not hard to observe that $G$ is a DAG.

It is known~\cite[Theorem 14.28]{fpt-book} that $k\times k$ \gt is W[1]-hard parameterized by $k$, and under ETH cannot be solved in $f(k)\cdot n^{o(k)}$ for any computable function $f$. Combining the two directions from Lemma~\ref{lem:dsn-redn-easy} and Lemma~\ref{lem:dsn-redn-hard}, we get a parameterized reduction from $k\times k$ \gt to an instance of DSN which is a planar DAG and has $O(k)$ terminal pairs. Hence, it follows that DSN on planar DAGs is W[1]-hard and under ETH cannot be solved in $f(k)\cdot n^{o(k)}$ time for any computable function $f$.
\end{proof}

Note that Theorem~\ref{thm:dsn-w[1]-hardness} shows that the $n^{O(k)}$ algorithm of
Feldman-Ruhl~\cite{feldman-ruhl} for DSN is asymptotically optimal.

\newpage

\bibliographystyle{splncs03}
\bibliography{docsdb}


\appendix

\section{Vertex-unweighted versions are more general than edge-weighted versions with
integer weights}
\label{appendix:vertex-general-than-edge}

In this section, for both the SCSS and DSN problems we show that the edge-weighted
version (with polynomially-bounded integer weights) can be solved using the
vertex-unweighted version. Hence all our hardness results from Theorem~\ref{thm:scss-main-hardness-planar-graphs},
Theorem~\ref{thm:scss-main-hardness-general-graphs} and Theorem~\ref{thm:dsn-w[1]-hardness} hold for the vertex-(un)weighted versions as well.

We give a formal proof for the DSN problem; the proof for the SCSS problem is similar. Consider an instance $I_1 = (G,T)$ of
edge-weighted DSN with integer weights where $T=\{(s_i,t_i)\ |\ i\in [k]\}$. Replace each edge of weight $\ell$ by $n\ell$ internal
vertices where $|G|=n$. Let the new graph be $G'$. Consider the instance $I_2$ of vertex-unweighted version where the set of terminals
is the same as in $I_1$.

\begin{theorem}
The instance $I_1$ of edge-weighted DSN has a solution of weight at most $C$ if and only if the instance $I_2$ of vertex-unweighted DSN has a solution with at most $Cn+n$ vertices.
\end{theorem}
\begin{proof}
Suppose there is a solution $E_1$ for $I_1$ of weight at most $C$. For each edge in $E_1$ pick all its internal vertices and two
endpoints in $E_2$. Clearly $E_2$ is a solution for $I_2$. The number of vertices in $E_2$ is $Cn+\gamma$ where $\gamma$ is the number of
vertices of $G$ incident to the edges in $E_1$. Since $\gamma \leq n$ we are done.

Suppose there is a (vertex-minimal) solution $E_2$ for $I_2$ having at most $Cn+n$ vertices. For any edge $e\in G$ of weight $c$ we need to
pick all the $cn$ internal vertices (plus the two endpoints of $e$) in $E_2$ if we actually want to use $e$ in a solution for
$I_1$. So for every edge $e\in E$ we know that $E_2$ contains either all or none of the internal vertices obtained after splitting up
$e$ according to its weight in $G$. Let the set of edges of $G$ all of whose internal vertices are in $E_2$ be $E_1=\{e_1, e_2,
\ldots, e_r\}$ and their weights be $c_1, c_2, \ldots, c_r$ respectively. Since $E_2$ is a solution for $I_2$ it follows that $E_1$ is a solution
for $I_1$. Let $S$ be the union of set of endpoints of the edges in $E_1$. Therefore $Cn+n \geq |S| + n(\sum_{i=1}^{r} c_i)$.
Since $|S|\geq 1$ we have $C\geq \sum_{i=1}^{r} c_i$, i.e., $E_1$ has weight
at most~$C$.
\end{proof}

Note that the above reduction works even in the case when the edges have zero weight\footnote{We mention this explicitly because some of the reductions in this paper do have edges with zero weight}: in this case we simply
won’t be adding any internal vertices.

\section{Treewidth and Minors}
\label{appendix:gt-defns}

\begin{definition} \textbf{(treewidth)}
Let $G$ be a given undirected graph. Let $\T$ be a tree and $B:V(\T)\rightarrow 2^{V(G)}$. The pair $(\T,B)$ is called as a \emph{tree decomposition} of an undirected graph $G$ is a tree $\T$ in which every vertex $x\in V(\T)$ has an assigned set of vertices $B_x \subseteq V(G)$ (called a bag) such that the following properties are satisﬁed:
\begin{itemize}
\item $\textbf{(P1)}$: $\bigcup_{x\in V(\T)} B_x = V (G)$.
\item $\textbf{(P2)}$: For each $\{u,v\}\in E(G)$, there exists an $x\in V(\T)$ such that $u, v\in B_x$.
\item $\textbf{(P3)}$: For each $v\in  V(G)$, the set of vertices of $\T$ whose bags contain $v$ induce a connected subtree of $\T$.
\end{itemize}
The \emph{width} of the tree decomposition $(\T,B)$ is $\max_{x\in V(\T)} |B_x|-1$. The treewidth of a graph $G$, usually denoted by $\tw(G)$, is the minimum width over all tree decompositions of $G$.

\end{definition}

\begin{definition} \textbf{(minor)}
Let $G,H$ be undirected graphs. Then $H$ is called a minor of $G$ if $H$ can be obtained from $G$ by deleting edges, deleting vertices and by contracting edges.
\end{definition}

\begin{definition} \textbf{(subdivision)}
Let $G$ be an undirected graph. An edge $e=u-v$ is subdivided by adding a new vertex $w$ and the edges $u-w$ and $v-w$. An undirected graph $H$ is called a subdivision of $G$ if $H$ can be obtained from $G$ by subdividing edges of $G$.
\end{definition}

\begin{lemma}
Subdivisions of outerplanar graphs have treewidth at most $2$.
\label{lem:tw-outerplanar}
\end{lemma}
\begin{proof}
Outerplanar graphs are known to be a subclass of series parallel graphs, and
hence have treewidth at most $2$. To prove this lemma, it is enough to show that
subdividing one edge of an outerplanar does not increase the treewidth. Let $G$
be an outerplanar graph, and $(\T,B)$ be a tree-decomposition of $G$ of width
at most $2$. Let $e=u-v$ be an edge in $G$ which is subdivided by adding a
vertex $w$ and the edges $u-w$ and $v-w$. We now build a tree decomposition for
the resulting graph $G'$. We add only one vertex to $V(\T)$: by property
$\textbf{(P2)}$, there exists $t\in V(\T)$ such that $u, v\in B_t$. Add a new
vertex $t'$ and set $B_{t'} = \{u,v,w\}$. Make $t'$ adjacent only to~$t$. It is
easy to check that $V(\T)\cup \{t'\}$ is a tree-decomposition for $G'$ of
treewidth at most $2$.
\end{proof}

\end{document}